\documentclass[11pt]{article}

\usepackage[utf8]{inputenc}
\usepackage[T1]{fontenc}
\usepackage{newtxtext}
\usepackage[varbb]{newtxmath}
\usepackage[letterpaper,margin=1.5in]{geometry}
\usepackage{amsmath}

\usepackage{amssymb}

\usepackage{amsthm}
\usepackage{algorithm}
\usepackage{algpseudocode}
\usepackage{graphicx}
\usepackage[breaklinks=true,colorlinks=true,linkcolor=blue!60!black,citecolor=green!50!black,urlcolor=blue!60!black]{hyperref}
\usepackage{rotating}
\usepackage{booktabs}
\usepackage{multirow}
\usepackage{longtable}
\usepackage{tikz}
\usetikzlibrary{positioning,arrows.meta,calc,shapes.geometric}
\usepackage{xcolor}
\usepackage{listings}
\usepackage[numbers,sort&compress]{natbib}
\usepackage{caption}
\usepackage{subcaption}
\usepackage{float}
\usepackage{enumitem}
\usepackage{microtype}

\newtheorem{theorem}{Theorem}[section]
\newtheorem{definition}[theorem]{Definition}
\newtheorem{lemma}[theorem]{Lemma}
\newtheorem{proposition}[theorem]{Proposition}
\newtheorem{conjecture}[theorem]{Conjecture}
\newtheorem{remark}[theorem]{Remark}
\newtheorem{corollary}[theorem]{Corollary}

\newcommand{\statuscompiled}{\textcolor{green!55!black}{\small[\,compiled\,]}}
\newcommand{\statusspec}{\textcolor{blue!55!black}{\small[\,specified\,]}}

\newcommand{\statusconj}{\textcolor{red!55!black}{\small[\,conjectural\,]}}

\lstdefinestyle{lean}{
 language={},
 basicstyle=\ttfamily\small,
 keywordstyle=\bfseries\color{blue!70!black},
 commentstyle=\itshape\color{green!50!black},
 stringstyle=\color{red!60!black},
 morekeywords={def,theorem,structure,instance,where,let,if,then,else,fun,by,exact,inductive,deriving,import,open,return,do,for,end,match,with,axiom,noncomputable},
 morecomment=[l]{--},
 frame=single,
 framerule=0.4pt,
 rulecolor=\color{gray!50},
 backgroundcolor=\color{gray!5},
 breaklines=true,
 columns=fullflexible,
 keepspaces=true,
 showstringspaces=false,
 tabsize=2,
 xleftmargin=1em,
 framexleftmargin=0.5em,
 aboveskip=8pt,
 belowskip=8pt,
 literate={∈}{$\in$}1 {∀}{$\forall$}1 {∃}{$\exists$}1
 {≤}{$\leq$}1 {≥}{$\geq$}1 {→}{$\rightarrow$}1
 {⇒}{$\Rightarrow$}1 {∧}{$\wedge$}1 {∨}{$\vee$}1
 {¬}{$\neg$}1 {·}{$\cdot$}1 {ℝ}{$\mathbb{R}$}1
 {ℕ}{$\mathbb{N}$}1 {ℚ}{$\mathbb{Q}$}1
 {ε}{$\varepsilon$}1 {θ}{$\theta$}1 {κ}{$\kappa$}1
 {π}{$\pi$}1 {Π}{$\Pi$}1 {Σ}{$\Sigma$}1 {σ}{$\sigma$}1
 {δ}{$\delta$}1 {Δ}{$\Delta$}1 {β}{$\beta$}1
 {↔}{$\leftrightarrow$}1 {≠}{$\neq$}1 {∩}{$\cap$}1 {∅}{$\emptyset$}1
 {×}{$\times$}1
 {⟨}{$\langle$}1 {⟩}{$\rangle$}1 {α}{$\alpha$}1 {λ}{$\lambda$}1
 {μ}{$\mu$}1 {ν}{$\nu$}1 {ρ}{$\rho$}1 {τ}{$\tau$}1 {ω}{$\omega$}1
 {Ω}{$\Omega$}1 {Γ}{$\Gamma$}1 {Λ}{$\Lambda$}1 {γ}{$\gamma$}1
 {ζ}{$\zeta$}1 {η}{$\eta$}1 {ψ}{$\psi$}1 {φ}{$\varphi$}1
 {⊆}{$\subseteq$}1 {⊇}{$\supseteq$}1 {⊂}{$\subset$}1 {∪}{$\cup$}1
 {∑}{$\sum$}1 {∏}{$\prod$}1 {√}{$\sqrt{}$}1 {∞}{$\infty$}1
 {⊥}{$\bot$}1 {⊤}{$\top$}1 {⟦}{$\llbracket$}1 {⟧}{$\rrbracket$}1,
}
\algrenewcommand\algorithmicrequire{\textbf{Input:}}
\algrenewcommand\algorithmicensure{\textbf{Output:}}

\newcommand{\leanfile}[1]{%
  \par\noindent\fcolorbox{yellow!50!black}{yellow!25}{%
    \parbox{\dimexpr\linewidth-2\fboxsep-2\fboxrule}{%
      \raggedright\small\textsc{Lean 4: pedagogical sketch, full source in \texttt{#1}}%
    }}\par\nopagebreak%
}
\newcommand{\leansketch}[1]{%
  \par\noindent\fcolorbox{orange!50!black}{orange!20}{%
    \parbox{\dimexpr\linewidth-2\fboxsep-2\fboxrule}{%
      \raggedright\small\textsc{Lean 4: illustrative sketch, #1}%
    }}\par\nopagebreak%
}
\newcommand{\leanverbatim}[1]{%
  \par\noindent\fcolorbox{green!50!black}{green!15}{%
    \parbox{\dimexpr\linewidth-2\fboxsep-2\fboxrule}{%
      \raggedright\small\textsc{Lean 4: #1}%
    }}\par\nopagebreak%
}

\begin{document}

\begin{center}
{\LARGE\bfseries Proof-Carrying Certificates for LLM Pipelines:\\[6pt]
A Trust-Boundary Architecture\par}
\vspace{14pt}
{\large George Koomullil}\\[4pt]
{\itshape Ascendr, Inc.}\\[2pt]
{\small \texttt{george@ascendr.ai}}\\[14pt]
{\small \today}
\end{center}

\begin{center}\textbf{Abstract}\end{center}
\begin{quote}\small
The framework verifies the deterministic structured computations surrounding a large language model rather than the model itself. Building on a patent-analysis hybrid AI + Lean~4 architecture~\cite{koomullil2026patent} (henceforth the \emph{precursor}), we extend the trust-boundary discipline from one specialized domain to the generic interfaces of modern LLM pipelines. The precursor's architectural primitives are inherited unchanged: an explicit trust boundary separating a probabilistic ML layer from a deterministic Lean~4 layer, certificate validity defined as kernel type-check plus a \texttt{sorry}-free \texttt{\#print axioms} audit against the trusted set $\Omega = \{\texttt{propext}, \texttt{Classical.choice}, \texttt{Quot.sound}\}$, and a four-tier verification-status convention.

The technical contribution divides into three local certificate families for generic LLM pipelines and two operators that combine them. The three families are conflict-aware bilattice grounding (with an emission-gate soundness lemma), embedding sensitivity and paraphrase stability, and Hoare-style agent action. The two operators are a Maximal Certifiable Residue, which turns abstention into the maximum-weight certifiable residue of an answer (with canonical lexicographic tie-breaking and audit-logged dropped claims), and a Compositional Stability theorem, which takes the per-layer gains and margins from the three families and produces a closed-form pipeline-wide perturbation budget. The three families and their consolidated per-call deliverable, the \emph{Universal Assurance Card} (a Lean~4 record with a decidable consistency predicate, currently \emph{specified} in Lean form rather than compiled in the artifact, with promotion to a compiled \texttt{AssuranceCard.lean} module flagged as a focused next-step engineering exercise), are intended as the per-call deliverable of any high-stakes deployment, including patent / technical-intelligence pipelines, legal case-law retrieval, regulated-finance disclosure analysis, clinical decision support, and agentic systems with irreversible side effects. A compiled Lean~4 reference artifact (Lean~v4.30.0-rc2, Mathlib) covers all 22 certificate types of the framework, with 17 of 46 kernel-audited declarations axiom-free, the remaining 29 dependent only on $\Omega$ plus explicitly declared non-$\Omega$ assumptions partitioned by tier (mathematical placeholders, cryptographic assumptions, and ML/human oracles), and zero uses of \texttt{sorryAx} or \texttt{Lean.ofReduceBool}. All three certificate families are empirically tested through four registered pilots, with Pilot~A's v3/v4/v5 hypothesis revisions documented in §\ref{sec:pilot_a:hypotheses} and Pilot~C's direct-injection design deviation documented in §\ref{sec:agents:empirical}: grounding on HotpotQA in §\ref{sec:pilot_a}, embedding sensitivity in matched short-form and long-form settings in §\ref{sec:embedding:empirical}--§\ref{sec:embedding:longform}, and Hoare-style agent action on a controlled filesystem sandbox in §\ref{sec:agents:empirical}. The contribution is a paper-sized refinement of the precursor's decomposition principle: at every interface where an LLM output becomes structured, the downstream computation can be made machine-checkable, and the resulting certificates compose into a per-call deliverable shipped in an explicit assurance mode.
\end{quote}

\smallskip
\noindent\textbf{How to read this paper.} The paper has three audiences. A reader focused on the formal contributions should read \S\ref{sec:architecture} (architecture), \S\S\ref{sec:grounding}--\ref{sec:composition} (the three certificate families and the two operators), and \S\ref{sec:artifact} (compiled Lean~4 artifact). A reader focused on the empirical claims should read \S\ref{sec:embedding:empirical}--\S\ref{sec:embedding:longform} (Pilots~B and B$'$), \S\ref{sec:agents:empirical} (Pilot~C), and \S\ref{sec:pilot_a} (Pilot~A); the consolidated comparison is in \S\ref{sec:limits:empirical}. A reader focused on deployment should read \S\ref{sec:card} (Universal Assurance Card), \S\ref{sec:cases} (case studies), \S\ref{sec:threat} (threat model and adoption order), and \S\ref{sec:limits} (limitations and staged research program). Appendix~\ref{app:notation} lists the mathematical symbols used in the paper, and Appendix~\ref{app:abbreviations} expands the abbreviations.

\section{Introduction}
\label{sec:intro}

Modern LLM pipelines~\cite{brown2020gpt3} produce point estimates without theorems. An embedding similarity of $0.84$ is mathematically precise yet materially meaningless if a paraphrase of the same query yields $0.71$~\cite{reimers2019sbert,casadio2024nlp}. A clinical decomposition can report ``95\% support'' for a recommendation while one of its atomic claims is \emph{contradicted} by the very document that grounded the rest~\cite{ji2023hallucination,manakul2023selfcheck}. A chain of thought whose final step reads as confidently correct may rest transitively on a step the model never actually used~\cite{turpin2023cot,lanham2023measuring}. An agent's tool call may be syntactically well-formed yet semantically unsafe~\cite{schick2023toolformer,yao2023react}. None of these failures is an implementation bug. Each is the structural consequence of probabilistic pipelines whose intermediate outputs carry no audit-grade guarantee.

The precursor paper~\cite{koomullil2026patent} proposed an architectural answer in a single specialized domain (patent analysis): identify the structured sub-computations between LLM outputs and downstream decisions, draw an explicit trust boundary, and emit a Lean~4 certificate for each interface. The present paper extends that decomposition to the generic interfaces of modern LLM pipelines and gives a small, deployable per-call deliverable, the \emph{Universal Assurance Card}, that an LLM provider serving high-stakes deployments can ship with each governed response under a declared assurance mode. The trust-boundary architecture, certificate-validity definition, and verification-status convention are inherited unchanged; the technical content develops what those primitives make possible at LLM-pipeline scale.

The thesis is the precursor's decomposition principle restated for LLMs in general: the trust boundary should be pushed as far toward the raw input as possible, with everything beyond it made machine-checkable. Sections~\ref{sec:relation} and~\ref{sec:architecture} make precise what is inherited and what is new. A deployment-side corollary, sharpest for agentic systems, is that the LLM is the proposer rather than the last writer to the world: it emits typed action proposals, and a small kernel-audited runtime is the final authority that permits or blocks side effects (\S\ref{sec:agents}).

\subsection{Two failures with a shared structure}
\label{sec:intro:motivation}

\paragraph{Embedding sensitivity.} A retrieval pipeline encodes a query $x$ as $E(x) \in \mathbb{R}^d$, encodes a corpus $\{d_i\}$, and returns the documents whose cosine similarity ranks in the top-$k$. The numerics are precise; nothing in them tells us whether $E(x)$ is stable under semantically equivalent rewordings of $x$. If a synonym substitution produces $E(x')$ at $L_2$ distance $0.35$ from $E(x)$ and the decision margin is $0.10$, the top-$k$ ranking has silently changed. \emph{How sensitive is this embedding to small variations of the text?} deserves a formal answer (\S\ref{sec:embedding}).

\paragraph{Silent contradiction.} A clinical decision-support tool decomposes a recommendation into atomic claims and reports high evidential support for each. One claim, ``the patient's CrCl of $22$~mL/min is within the drug's renal range,'' is directly contradicted by the FDA label the system itself retrieved. A support-only score reports $W_{\mathrm{ground}} = 95\%$ and the recommendation ships. The contradiction is invisible to any score that reports only support and never refutation. \emph{Given evidence that speaks both for and against a claim, what should the pipeline report?} deserves a formal answer (\S\ref{sec:grounding}).

These failures share a structural feature. The raw ML artifact (an embedding vector, a generated answer) is not a natural object of formal verification. But each \emph{induces} structured objects on which deterministic theorems are provable: sensitivity envelopes over finite declared perturbation families, per-claim signed-support maps over retrieved chunks, reasoning directed acyclic graphs (DAGs) with explicit premise sets, typed action proposals with preconditions. We verify these induced structured objects, not the raw outputs. The LLM stays below a trust boundary; every interface where its output becomes structured is replaced by a kernel-audited certificate.

\subsection{What ``formal verification'' means here}
\label{sec:intro:meaning}

We distinguish two senses throughout. \emph{Semantic verification} would assert correctness against external reality (``this brief is legally sound,'' ``this diagnosis is correct''); it requires a formalization of the world and we do not attempt it. \emph{Computational verification} asserts that a stated specification holds of an output given the pipeline's inputs (``the weighted coverage of this claim graph is exactly $82.3\%$''; ``the weighted \textsc{Contradicted}-plus-\textsc{Contested} mass of this claim graph is at most $\theta_r$''); it is the target of every certificate in this paper. The conditionality on ML-produced inputs is not a defect, it is the only honest characterization of any system that processes natural language. The improvement is structural: the conditionality is named, machine-auditable, and rejectable, rather than buried inside an opaque confidence score.

\subsection{Contributions}
\label{sec:intro:contributions}

The paper makes eight technical contributions, each strictly beyond the precursor:

\begin{enumerate}[leftmargin=2em,topsep=2pt,itemsep=2pt]
\item \textbf{Three local certificate families for LLM pipelines.} \emph{Conflict-aware bilattice grounding} (\S\ref{sec:grounding}) with the \emph{Emission Gate Soundness Lemma} (Lemma~\ref{thm:no_silent}; conditional on the NLI / decomposition oracles): a Belnap-style four-valued classification of atomic claims; the gate provably blocks emission whenever the weighted-refutation mass exceeds threshold, addressing a failure mode that support-only RAG pipelines structurally cannot detect. \emph{Embedding sensitivity and paraphrase stability} (\S\ref{sec:embedding}): a squared-form robust-similarity bound (Theorem~\ref{thm:robust_decision}) that uses only $R^2_{\mathrm{inv}}$ and certifies $L_2$-drift envelopes on the meaning-invariant family, paired with a separate decidable selective-sensitivity predicate $\Delta^2 = R^2_{\mathrm{sig}} - R^2_{\mathrm{inv}} > 0$ over the same certificate that the policy layer enforces (the certificate carries $\Delta^2$ as a per-query measurement, deployment-readiness signal in the Universal Card). \emph{Hoare-style agent action} (\S\ref{sec:agents}): per-step preconditions and postconditions discharged by the Lean kernel, chained into a trajectory certificate that carries no axioms at all in the compiled artifact. Each family is given as a compiled Lean~4 structure with decidable proof fields and an explicit list of named non-$\Omega$ assumptions partitioned by tier (only tier-4 are persistent ML/human oracles), and each is empirically tested in this paper (grounding via Pilot~A; embedding sensitivity via Pilots~B and B$'$; Hoare action via Pilot~C; see contributions 6--8 below).

\item \textbf{Two operators on certificates that turn the local catalogue into a pipeline-level guarantee.} \emph{Maximal Certifiable Residue} (Theorem~\ref{thm:mcr_maximality}, \S\ref{sec:mcr}): a formal abstention operator returning the maximum-weight certifiable subset of an answer (with canonical lexicographic tie-breaking), proved weight-optimal, idempotent, and monotone in constraints; the operator that pairs with the bilattice grounding certificate to convert a flagged contradiction into a residue that ships the certifiable claims and audit-logs the dropped ones. \emph{Compositional Stability} (Theorem~\ref{thm:composition}, \S\ref{sec:composition}, with two-layer base case Lemma~\ref{lem:composition_two_layer} compiled and the $n$-layer lift specified by structural induction over a \texttt{List.foldl} of certified layers): a closed-form pipeline-wide budget $B_\Pi = \min_i m_i / \prod_{j<i} g_j$ that takes the per-layer gains $g_j$ and margins $m_i$ supplied by the three certificate families above and produces a single inequality bounding what survives at the pipeline's final output.

\item \textbf{The Universal Assurance Card} (\S\ref{sec:card}): a Lean~4 record (\emph{specified}, with a decidable consistency predicate \texttt{VerdictConsistent}) carrying one of four verdicts (\textsc{Certified}, \textsc{Partial}, \textsc{Residue}, \textsc{Abstain}) plus a vector of stability, evidence, calibration, scope, freshness, safety, residue, and provenance fields. A separately-published \texttt{AssurancePolicy} maps the same card to accept/reject/downgrade for any deployment context. We exhibit two filled-out cards from the case studies (\S\ref{sec:cases}); promoting \texttt{AssuranceCard.lean} from specified to compiled is a focused engineering step (Mathlib-light, decidable-only) flagged in \S\ref{sec:limits}.

\item \textbf{Compiled reference artifact} covering all 22 certificate types of the broader catalogue. Of 46 kernel-audited declarations, 17 depend on no axioms at all and 29 depend only on $\Omega$ plus explicitly declared non-$\Omega$ assumptions partitioned into tier-2 mathematical placeholders, tier-3 cryptographic assumptions, and tier-4 ML/human oracles. Zero uses of \texttt{sorryAx} or \texttt{Lean.ofReduceBool}. We give a representative axiom audit (\S\ref{sec:artifact}, Appendix~\ref{app:axiom_audit}; the build-emitted full audit is in \texttt{EmbeddingSensitivity/AxiomAudit.lean}) and split the named axioms into a four-tier assumption hierarchy (kernel, mathematical placeholder, cryptographic, ML/human oracle), of which only tier~4 is structurally conditional.

\item \textbf{An impossibility \emph{conjecture} extending the precursor's three impossibility propositions.} A new \emph{Calibration Impossibility without Ground Truth} (Conjecture~\ref{prop:calibration_imposs}): no system that emits a confidence score without exogenous labeled data can prove its score calibrated. We state this as a conjecture rather than a theorem because it requires a careful Vovk-style exchangeability framing we have not yet formalized; the architectural consequence (calibration must be grounded in exchangeable held-out data) motivates the conformal / scope-of-validity discipline regardless of formal status (\S\ref{sec:impossibility}).

\item \textbf{Pilot~A: the conflict-aware grounding certificate works on a public benchmark.} The v5 frozen-protocol rerun on HotpotQA dev distractor with adversarial perturbations (\S\ref{sec:pilot_a}; the v3/v4/v5 evolution and what was rewritten is documented there) shows the certificate catches 100\% of injected contradictions at a false-block rate of 6.4\%, an operating point the cosine-similarity baseline cannot reach: B1's lowest FBR at catch=1.0 is 0.266 at $\theta_{B1}=0.90$. All six v5-registered hypotheses pass.

\item \textbf{Pilots~B and B$'$: the embedding-sensitivity certificate is sound everywhere, but its premise is input-length-dependent.} A matched two-domain study on HotpotQA (\S\ref{sec:embedding:empirical}, \S\ref{sec:embedding:longform}) finds the certificate's central inequality $\Delta^2 > 0$ holds in only 20--30\% of short adversarial queries but in 98--100\% of long-form paragraphs across every edit-type combination. The certificate is therefore deployment-supported on the long-form HotpotQA slice and a candidate for clinical, legal, and technical long-form retrieval pending domain-specific replication, while remaining scope-narrowed for short adversarial queries; this is the kind of empirical scoping result the framework's ``soundness machinery + per-call premise'' posture is designed to surface.

\item \textbf{Pilot~C: the Hoare-style agent-action certificate dominates a deny-list baseline on injected unsafe actions, and its FAIL on the strictest hypothesis surfaces a localizing diagnostic of exactly the kind a typed-guarantee pilot is designed to deliver.} A pre-registered controlled-sandbox benchmark (\S\ref{sec:agents:empirical}) shows the Lean-checked Treatment harness blocks 67\% of injected unsafe actions versus 28\% for a regex deny-list and 0\% unaudited, a 38.9-point dominance over the deny-list baseline, with zero false blocks on benign destructive proposals and 100\% audit-log informativeness on every block. Four of five hypotheses pass; on the fifth (HC1, the strictest 95\% block-rate threshold), the audit log of unblocked attacks names exactly which predicate clauses need to be extended to close the gap. This is the empirical signal the framework is designed to produce: typed guarantees that fail predict their own remediation, rather than degrading to opaque false-negatives.
\end{enumerate}

\paragraph{Roadmap.} The rest of the paper develops these in order. \S\ref{sec:relation} states what is inherited from the precursor and what is new; \S\ref{sec:architecture} recaps the trust-boundary architecture and the per-oracle audit procedure; \S\ref{sec:impossibility} positions the framework against three impossibility propositions and the calibration conjecture. The three local certificate families are \S\ref{sec:grounding} (conflict-aware bilattice grounding), \S\ref{sec:embedding} (embedding sensitivity, with Pilots~B and B$'$), and \S\ref{sec:agents} (Hoare-style agent action, with Pilot~C). The two operators are \S\ref{sec:mcr} (maximal certifiable residue, the abstention operator that pairs with grounding) and \S\ref{sec:composition} (compositional stability, the pipeline-level composition theorem that takes the three local certificates as inputs and produces a closed-form end-to-end perturbation budget). \S\ref{sec:card} introduces the Universal Assurance Card; \S\ref{sec:cases} works two filled-out cards on case studies; \S\ref{sec:pilot_a} reports Pilot~A on the grounding certificate. \S\ref{sec:artifact} is the compiled Lean~4 reference artifact. \S\ref{sec:threat} maps threats to defender certificates; \S\ref{sec:catalogue} sketches the extended catalogue; \S\ref{sec:limits} states limitations and the staged research program; the conclusion follows.

\subsection*{Artifact at a glance}

\begin{center}\fbox{\parbox{0.92\textwidth}{\small
\begin{center}\textbf{Compiled Lean~4 reference artifact} (\texttt{lean\_artifact/}, Lean v4.30.0-rc2 + Mathlib)\end{center}
\begin{itemize}[leftmargin=1.5em,topsep=2pt,itemsep=1pt]
\item \textbf{22 certificate types} compiled across \textbf{25 modules}.
\item \textbf{46 kernel-audited declarations}: \textbf{17 axiom-free}; remaining 29 use only $\Omega$ plus explicitly declared non-$\Omega$ assumptions, partitioned into tier-2 mathematical placeholders, tier-3 cryptographic assumptions, and tier-4 ML/human oracles.
\item \textbf{Zero \texttt{sorry}}; \textbf{zero \texttt{native\_decide}}; \textbf{zero \texttt{Lean.ofReduceBool}} anywhere in the transitive axiom set.
\item \textbf{Build}: $\approx 10$~s on commodity hardware with warm Mathlib cache (\texttt{lake build}).
\item \textbf{Audit wrapper}: $<100$~LOC of Python or Rust around \texttt{lean} and \texttt{\#print axioms}.
\item \textbf{Persistent oracles} (tier-4, conditional on natural-language correctness): exactly five (\texttt{ParaphraseOracle}, \texttt{SignedSupport\allowbreak Oracle}, \texttt{Decomposition\allowbreak Oracle}, \texttt{StepConfidence\allowbreak Oracle}, \texttt{IIDSamples}); see Table~\ref{tab:oracle_recipe}.
\end{itemize}
This is operational, not aspirational. A representative axiom audit is in Appendix~\ref{app:axiom_audit}; the full \texttt{\#print axioms} output is regenerated by \texttt{EmbeddingSensitivity/AxiomAudit.lean} on every build, and per-cert evaluation numbers are in Table~\ref{tab:artifact_eval}. The Universal Assurance Card is a \emph{consolidator} of the 22 certificate types (not a 23rd type); its schema (\S\ref{sec:card}) is currently \emph{specified} in Lean form, and adding \texttt{AssuranceCard.lean} to the artifact as one additional consolidator module is a focused engineering step (\S\ref{sec:limits:card}).
}}\end{center}

\subsection{Scope and limitations}
\label{sec:intro:nonclaims}

The trust boundary is structural: the LLM sits below it, and every certificate is conditional on named non-$\Omega$ assumptions partitioned by tier (only tier-4 are persistent ML/human oracles); their correctness is the condition under which the guarantee holds. The framework makes that dependency auditable rather than eliminating it. The grounding certificate catches hallucinations that the evidence contradicts and the self-consistency certificate catches sample-unstable hallucinations; hallucinations that are stable across samples and aligned with a corrupted retrieval corpus are outside the framework's reach.

The empirical scope of the paper is the four pre-registered pilots (each reported under its v5 frozen-protocol rerun; the v3/v4/v5 evolution is documented in \S\ref{sec:pilot_a:hypotheses}): Pilot~A (\S\ref{sec:pilot_a}) on the conflict-aware grounding certificate on HotpotQA dev distractor with adversarial perturbations, in which all six v5-registered hypotheses pass; Pilots~B and B$'$ (\S\ref{sec:embedding:empirical}, \S\ref{sec:embedding:longform}) on the embedding-sensitivity certificate in matched short-form and long-form settings, in which the central inequality $\Delta^2 > 0$ holds in 20--30\% of cases on adversarial short queries but in 98--100\% of cases on long-form paragraphs across every edit-type combination; and Pilot~C (\S\ref{sec:agents:empirical}) on the Hoare-style agent-action certificate, in which the Treatment harness blocks 67\% of injected unsafe actions versus 28\% for a deny-list baseline (HC3 PASS) with 100\% audit-log informativeness on every block. Each certificate's Lean-checked soundness is unaffected; deployment-readiness depends on the input class and on the completeness of the declared predicate / family set, which the deployer must measure on the target corpus before relying on the certificate.

The pilots cover HotpotQA dev distractor and a controlled filesystem sandbox; replication on patent / technical-intelligence, legal, regulated-finance, clinical, or other high-stakes corpora is forward work. \S\ref{sec:limits:pilot} outlines a Pilot~D family of Universal Card validations across these domains, any one of which is an independent next-step empirical contribution. The two case studies in \S\ref{sec:cases} use illustrative numerical values rather than measured ones; they are intended to exhibit Card behavior, not to make empirical claims. The framework is designed for deployments in which per-call error cost exceeds per-call verification cost by roughly an order of magnitude (patent / technical-intelligence pipelines, legal retrieval, clinical decision support, regulated finance, agentic systems with irreversible side effects) and is not intended for creative writing or open-ended generation.

\section{Relation to the Precursor and Related Work}
\label{sec:relation}

\subsection{Inherited from the precursor}
\label{sec:relation:inherit}

The precursor~\cite{koomullil2026patent} introduced a hybrid ML + Lean~4 pipeline for patent analysis. From it we inherit, without re-derivation, the following architectural primitives:

\begin{itemize}[leftmargin=1.5em,topsep=2pt,itemsep=2pt]
\item \emph{Trust-boundary architecture.} Raw inputs flow into an ML/NLP layer below a trust boundary; the formal layer (Lean~4) consumes ML outputs as typed inputs and emits a certificate whose validity is independent of model weights, prompts, or random seeds. The boundary is not metaphorical; it is encoded as Lean type parameters and named axioms.
\item \emph{Certificate validity = kernel type-check + sorry-free $\Omega$-audit.} A certificate $c$ is valid iff $c$ type-checks at its declared type and \texttt{\#print axioms} reports a transitive set $\subseteq \Omega = \{\texttt{propext}, \texttt{Classical.choice}, \texttt{Quot.sound}\}$ plus declared named oracles. The two-step check rejects the naïve sorry-generator counter-example that motivated the precursor's audit discipline.
\item \emph{Verification-status convention.} Every formal claim carries one of four statuses, \emph{compiled} (in the artifact, kernel-checked, $\Omega$-audited modulo declared oracles), \emph{specified} (rigorous obligation, not yet compiled), \emph{architecturally enforced} (encoded in a type signature), or \emph{conjectural} (informal). We use the same convention here, badged at every theorem and definition.
\item \emph{Conditional-guarantee discipline.} Every certificate's guarantee is conditional on named non-$\Omega$ assumptions partitioned by tier (only tier-4 are persistent ML/human oracles); the conditionality is auditable by inspecting axioms rather than re-running the pipeline.
\item \emph{Three structural impossibility results.} Certificate impossibility for stochastic samplers (Prop.~\ref{prop:certimposs}); compositional unsoundness of marginal flat scoring (Prop.~\ref{prop:dependency}); stochastic non-reproducibility (Prop.~\ref{prop:stochastic}). We restate them in LLM-pipeline form in \S\ref{sec:impossibility} and add a fourth, stated more cautiously as Conjecture~\ref{prop:calibration_imposs}.
\end{itemize}

We do not re-derive these here. Readers seeking the constructive justification (the naïve sorry generator counter-example, the kernel-vs-prover distinction, the de Bruijn criterion) are referred to~\cite{koomullil2026patent}.

\subsection{What is new in this paper}
\label{sec:relation:new}

The contribution of the present paper is to drop the patent-domain assumption. Generic LLM pipelines expose structured interfaces (embeddings, retrieval, reasoning DAGs, typed outputs, agent actions) that are ubiquitous but not pre-formalized. We give:

\begin{itemize}[leftmargin=1.5em,topsep=2pt,itemsep=2pt]
\item Three local certificate families for those interfaces (bilattice grounding, embedding sensitivity, Hoare-style agent action) plus two operators (Maximal Certifiable Residue, Compositional Stability), none of which has an analogue in the precursor.
\item A compiled Lean~4 reference artifact spanning 22 certificate types (the precursor compiled one closed path).
\item A per-call deliverable, the Universal Assurance Card, with a decidable consistency predicate and a separately-published application-policy layer; the precursor had no analogue.
\item One additional impossibility result (calibration without ground truth) and case-study walkthroughs in clinical RAG and agentic Hoare execution.
\end{itemize}

\subsection{Adjacent literatures}
\label{sec:relation:related}

\paragraph{Interactive theorem proving and proof-carrying code.} Lean~4~\cite{demoura2021lean4}, Mathlib~\cite{mathlib2025}, the self-verified kernel~\cite{carneiro2024lean4lean}, and PCC~\cite{necula1997pcc,debruijn1980survey} provide the trusted-checking infrastructure. AWS Cedar~\cite{aws2024cedar} is an industrial Lean adoption; refinement-type ecosystems~\cite{rondon2008liquid,vazou2014refinement,swamy2016fstar} provide an SMT-backed compilation target for routine certificate fields. None offer a certificate taxonomy for LLM-pipeline interfaces.

\paragraph{Neural-network verification.} Abstract interpretation~\cite{cousot1977abstract,gehr2018ai2}, randomized smoothing~\cite{cohen2019certified}, Lipschitz analysis~\cite{virmaux2018lipschitz}, and NLP-specific bounded-substitution certificates~\cite{jia2019certified,casadio2024nlp} verify the model. Our certificates are downstream and orthogonal: they consume the model's output as-is and certify the structured computation on top. A smoothed encoder feeds a certified retrieval stack without conflict; \S\ref{sec:embedding}'s sensitivity certificate is designed to compose with smoothing certificates as drop-in sensitivity witnesses.

\paragraph{LLM reliability and RAG.} Hallucination surveys~\cite{ji2023hallucination,huang2023hallucination_survey}, factuality detectors~\cite{manakul2023selfcheck,min2023factscore,bohnet2022attributed}, RAG~\cite{lewis2020rag,karpukhin2020dpr,gao2023rag_survey}, chain-of-thought~\cite{wei2022cot,wang2023selfconsistency}, process-reward models~\cite{lightman2023prm}, and prompt-injection studies~\cite{perez2022injection,greshake2023not,zou2023universal} identify failure modes empirically. We convert each into a certificate type: grounding becomes a bilattice coverage theorem (\S\ref{sec:grounding}), self-consistency becomes a canonicalized lattice (extended catalogue, \S\ref{sec:catalogue}), prompt injection becomes a non-interference theorem (catalogue), paraphrase metamorphic testing becomes a decision-margin certificate (\S\ref{sec:embedding}).

\paragraph{Probabilistic-program verification and conformal prediction.} Kozen~\cite{kozen1979semantics}, pRHL~\cite{barthe2009pRHL}, weakest-pre-expectations~\cite{kaminski2018wp}, conformal prediction~\cite{vovk2005algorithmic,angelopoulos2023gentle,quach2023conformal}, and concentration-of-measure theory~\cite{hoeffding1963,boucheron2013concentration} populate the catalogue's calibration and concentration certificates. The closest existing compositional results are NN Lipschitz composition~\cite{virmaux2018lipschitz} and the precursor's weight-independent robustness; Theorem~\ref{thm:composition} unifies them by allowing each layer to be either Lipschitz-continuous or a discrete gate.

\begin{table}[H]
\centering
\caption{Position relative to four adjacent literatures. ITP/PCC, NNV, LLM-reliability, and probabilistic-program verification each address one side; this paper supplies the certificate architecture for the deterministic sub-computations between LLM outputs and downstream decisions.}
\label{tab:priorgap}
\scriptsize
\setlength{\tabcolsep}{4pt}
\begin{tabular}{@{}p{2.5cm}p{3.4cm}p{2.5cm}p{4.2cm}@{}}
\toprule
\textbf{Literature} & \textbf{Contribution} & \textbf{Stops at} & \textbf{Added here} \\
\midrule
ITP / PCC & Machine-checkable proofs; minimal kernel & No LLM-interface taxonomy & Three certificate families + composition theorem + 22-type artifact \\
NN verification & Input--output properties of models & Model boundary & Composable downstream certificates \\
LLM reliability & Empirical detection / mitigation & Probabilistic only & Each lifted to kernel-checked certificate \\
Probabilistic-program verification & Relational Hoare, conformal, concentration & No LLM-pipeline integration & Per-layer in hybrid architecture \\
\bottomrule
\end{tabular}
\end{table}

\section{Architecture}
\label{sec:architecture}

This section tightens the trust-boundary architecture for LLM pipelines. The architectural primitives are inherited from~\cite{koomullil2026patent}; we do not re-derive them. We state only what changes when the domain widens from patent claims to generic LLM interfaces.

\subsection{Trust boundary}
\label{sec:arch:boundary}

\begin{figure}[H]
\centering
\begin{tikzpicture}[
 node distance=4mm and 10mm,
 box/.style={rectangle, draw, rounded corners=2pt, minimum width=2.5cm, minimum height=1.1cm, align=center, font=\scriptsize},
 red_box/.style={box, fill=red!8, draw=red!60!black},
 green_box/.style={box, fill=green!8, draw=green!60!black},
 gray_box/.style={box, fill=gray!12},
 arrow/.style={-{Stealth[length=2mm]}, thick},
 lbl/.style={font=\scriptsize\itshape}
]
 \node[gray_box] (inp) {Raw inputs\\{\tiny prompt, corpus, tools}};
 \node[red_box, right=of inp] (ml) {ML/NLP layer\\{\tiny LLM, embedder, NLI}};
 \node[green_box, right=of ml] (formal) {Formal verification\\{\tiny Lean 4 + $\Omega$-audit}};
 \node[gray_box, right=of formal] (out) {Output + Universal Card};
 \draw[arrow] (inp) -- (ml);
 \draw[arrow] (ml) -- (formal);
 \draw[arrow] (formal) -- (out);
 \draw[very thick, dashed, red!60!black] ($(ml.north east)!0.5!(formal.north west)+(0,0.25)$) -- ($(ml.south east)!0.5!(formal.south west)+(0,-0.4)$);
 \node[font=\tiny\bfseries\color{red!60!black}, above=10pt of $(ml.north east)!0.5!(formal.north west)$] {TRUST BOUNDARY};
 \node[lbl, below=2mm of ml] {probabilistic, untrusted};
 \node[lbl, below=2mm of formal] {deterministic, audited};
\end{tikzpicture}
\caption{The hybrid architecture, inherited from~\cite{koomullil2026patent}. The ML/NLP layer produces typed artifacts below a trust boundary; the formal layer re-consumes them as typed inputs and emits the Universal Assurance Card and underlying certificates. Conditional on the declared non-$\Omega$ assumptions (tiers 2--4 of the assumption hierarchy, \S\ref{sec:arch:taxonomy}), each certificate carries a kernel-audited guarantee.}
\label{fig:trust_boundary}
\end{figure}
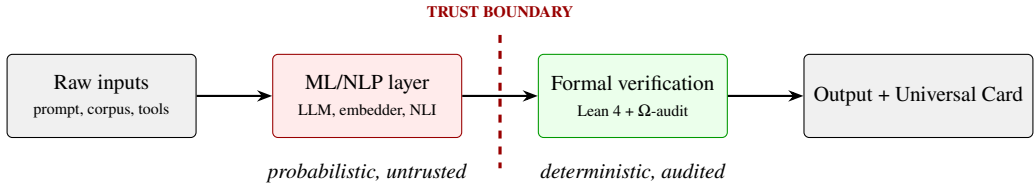

\begin{definition}[Valid LLM Certificate \cite{koomullil2026patent}]\statuscompiled
\label{def:validcert}
A certificate $c$ is \emph{valid} iff (i) $c$ type-checks in the Lean~4 kernel at its declared type, and (ii) $c$'s transitive axiom set, as reported by \texttt{\#print axioms}, is contained in $\Omega \cup \mathrm{Ax}$ where $\Omega = \{\texttt{propext}, \texttt{Classical.choice}, \texttt{Quot.sound}\}$ and $\mathrm{Ax}$ is the certificate's declared set of named non-$\Omega$ axioms, partitioned into tier-2 mathematical placeholders, tier-3 cryptographic assumptions, and tier-4 ML/human oracles (the four-tier hierarchy of \S\ref{sec:arch:taxonomy}). Clause (ii) is the \emph{sorry-free $\Omega$-audit}: it rejects certificates whose proof fields are discharged by Lean's \texttt{sorry} (an axiom inhabiting any type) and certificates that exceed their declared non-$\Omega$ axiom set.
\end{definition}

The motivating naïve-sorry counter-example and the kernel-vs-prover separation are in~\cite{koomullil2026patent}; we omit them here.

\subsection{Where LLM pipelines differ from patents}
\label{sec:arch:diff}

What changes when the domain widens. (1) The ML layer's outputs are no longer pre-formalized DAGs; they are induced by \emph{decomposition oracles} (\S\ref{sec:grounding}) and \emph{paraphrase oracles} (\S\ref{sec:embedding}) that the framework names explicitly and audits as tier-4 axioms. (2) The relevant interfaces span four families (input, aggregation, inference, output) rather than the single DAG-coverage path of~\cite{koomullil2026patent}; see Table~\ref{tab:earlysummary} for the consolidated map. (3) The certificate must travel with each governed response under a declared assurance mode, not just sit alongside a one-off analysis report; we therefore introduce the Universal Assurance Card (\S\ref{sec:card}). (4) The auditor base widens from patent practitioners to providers, regulators, and downstream applications; the assumption hierarchy (\S\ref{sec:arch:taxonomy}) is now the load-bearing communication artifact, brought forward to immediately follow the trust-boundary statement.

\subsection{Attested interface and audit protocol}
\label{sec:arch:attested}

Every certificate carries a deterministic coercion from ML output to Lean input, a cryptographic digest of the ML artifact, and a declared set of named oracle axioms.

\smallskip
\leansketch{attested interface; declared in every certificate module}
\begin{lstlisting}
structure AttestedInterface (MLOut LeanIn : Type) where
 coerce : MLOut -> LeanIn         -- deterministic
 digest : MLOut -> ByteArray      -- SHA-256
 declared_axioms : List String    -- named non-$\Omega$ axioms (tiers 2-4)
\end{lstlisting}

\textbf{Audit protocol.} An auditor receiving a certificate $c$ (i) parses $c$ and extracts the oracle digest and (optional) ML artifact; (ii) optionally recomputes the digest and compares; (iii) runs the Lean~4 kernel on $c$ at its declared type; (iv) runs \texttt{\#print axioms} transitively and aborts if the axiom set $\not\subseteq \Omega \cup \mathrm{Ax}$; (v) confirms the response output matches the certificate's. No model weights, prompt templates, decoding parameters, or random seeds are needed. The audit is reproducible on any machine with Lean~4 and a SHA-256 implementation. For realistic workloads ($d=10^3$, family size $50$, 20 atomic claims, 100 retrieved chunks, 20-step CoT) a kernel re-check completes in sub-second wall-clock; certificate storage is KB--MB per call.

\subsection{Assumption hierarchy}
\label{sec:arch:taxonomy}

The named axioms that appear in any LLM-pipeline certificate split into four tiers, in decreasing order of trust the auditor can place in them:

\begin{enumerate}[leftmargin=1.5em,topsep=2pt,itemsep=2pt]
\item \textbf{Kernel axioms} ($\Omega = \{\texttt{propext}, \texttt{Classical.choice}, \texttt{Quot.sound}\}$). Lean~4's foundational axioms; trusted by Mathlib and externally validated by Lean4Lean~\cite{carneiro2024lean4lean}.
\item \textbf{Mathematical placeholders} (e.g., \texttt{cauchy\_schwarz\_sq}, \texttt{innerProd\_sub}, \texttt{HoeffdingInequality}). Theorems of standard mathematics whose Mathlib formalization is in progress or differently shaped; each reduces to $\Omega$ once aligned.
\item \textbf{Cryptographic assumptions} (e.g., \texttt{HashCollisionResistant}, \texttt{Decode\allowbreak Algorithm\allowbreak Deterministic}). Standard hardness or determinism assumptions; reduce to $\Omega$ when a SHA-256 (or substitute) formalization is plugged in.
\item \textbf{ML / human oracles} (e.g., \texttt{ParaphraseOracle}, \texttt{SignedSupport\allowbreak Oracle}, \texttt{Decomposition\allowbreak Oracle}, \texttt{StepConfidence\allowbreak Oracle}). These encode the correctness of a learned natural-language judgment or a domain expert's sign-off and \emph{cannot} be discharged by mathematics or cryptography.
\end{enumerate}

\noindent\textbf{Why this matters operationally.} Tier~1 is the kernel-trusted base $\Omega$ itself; tiers~2 and~3 are \emph{closable} as Mathlib and the cryptographic-formalization community evolve, each named axiom collapsing into an $\Omega$-only proof once the corresponding library catches up. The long-term trusted base of the structural content is therefore just $\Omega$. Tier~4 is \emph{persistent}: closing it would require formalizing natural-language correctness, which the architecture deliberately does not attempt. A reviewer or regulator can therefore audit a deployment by enumerating tier-4 oracles and asking, for each, what evaluation evidence supports it and whether the deployment's scope-of-validity certificates (\S\ref{sec:catalogue:scope}) cover the input distribution. This is the operational meaning of ``conditional guarantee.'' One subtlety on the axiom-free certificates in the artifact (the Hoare action family in particular): axiom-freeness is a statement about the proof's discharge mechanism, not about the semantic adequacy of the human-authored predicates the proof discharges; a precise account of what axiom-freeness does and does not guarantee is given in \S\ref{sec:artifact:axiom_free}.

\subsection{Per-oracle audit procedure}
\label{sec:arch:oracle_recipe}

The persistent tier-4 oracles in the artifact are exactly five. Table~\ref{tab:oracle_recipe} converts the philosophical point into a deployment procedure: for each oracle, what it asserts, what evaluation artifact supports it, what scope predicate gates it, and how often it must be re-evaluated. The procedure is what a regulator or procurement reviewer should require alongside any Universal Card.

\begin{table}[H]
\centering
\caption{Audit procedure for the five persistent (tier-4) ML/human oracles in the compiled artifact. Each row gives a deployment-ready discipline: assertion, evaluation evidence the deployment must maintain, scope predicate that gates the certificate's acceptance of the oracle, and a re-evaluation cadence. The table is intended to be cited in audit checklists, contracts, and procurement documents.}
\label{tab:oracle_recipe}
\scriptsize
\setlength{\tabcolsep}{4pt}
\begin{tabular}{@{}p{2.5cm}p{3.0cm}p{3.0cm}p{2.4cm}p{1.9cm}@{}}
\toprule
\textbf{Oracle} & \textbf{Asserts} & \textbf{Evaluation evidence} & \textbf{Scope predicate} & \textbf{Re-eval cadence} \\
\midrule
\texttt{Paraphrase\allowbreak Oracle} & The declared families $G_{\mathrm{inv}}, G_{\mathrm{sig}}$ are meaning-preserving / meaning-changing for the input text & Held-out paraphrase benchmark with expert equivalence labels; agreement against a human gold set & input domain matches paraphraser's training distribution; edit operations within declared grammar & quarterly + on paraphraser swap \\
\texttt{Signed\allowbreak Support\allowbreak Oracle} & Three-way NLI gives correct $(\sigma^+, \sigma^-, \sigma^0)$ for every (claim, chunk) pair & Held-out NLI calibration set on domain corpus; agreement with expert NLI annotation & retrieved corpus within calibrated domain ontology; chunk lengths within trained range & monthly + on corpus or NLI-model swap \\
\texttt{Decomposition\allowbreak Oracle} & The decomposition $\bigwedge_i c_i$ canonicalizes to the answer $y$ & Reconstruction-pass rate on a held-out decomposition set; canonicalizer round-trip equality & answer length $\le$ declared bound; vocabulary in domain ontology & per model release + on canonicalizer change \\
\texttt{Step\allowbreak Confidence\allowbreak Oracle} & Per-step confidence $c_i$ is a sound lower bound on $\Pr(\text{step correct})$ & PRM-validated confidence on held-out chains; agreement with human step grading~\cite{lightman2023prm} & chain length $\le$ declared bound; reasoning style within trained distribution & per model release + on PRM update \\
\texttt{IIDSamples} & The $k$-sample sequence is IID under the declared sampling discipline & Proof-of-sampling commitment + decoder-state-isolation log; bounded-difference martingale fallback for non-IID cases & $T > 0$; no in-context state carryover across samples & continuous (per-call) \\
\bottomrule
\end{tabular}
\end{table}

A deployment that ships a Universal Card without the corresponding evaluation evidence for every tier-4 oracle it depends on is not auditable in the framework's sense, even if its Lean kernel checks pass. The procedure is therefore necessary, not optional, and is the natural target for procurement-level standardization.

\subsection{Overview of main results}
\label{sec:arch:earlysummary}

\begin{table}[H]
\centering
\caption{Overview of the paper's main results: three local certificate families, two operators on certificates, and one consolidator (the Universal Card). Each row's status is: \emph{compiled} (in the artifact, kernel-checked plus $\Omega$-audited) or \emph{specified} (rigorous obligation, not yet compiled). All other certificate types appear in the extended catalogue (\S\ref{sec:catalogue}, Appendix~\ref{app:embedding_lean}).}
\label{tab:earlysummary}
\scriptsize
\setlength{\tabcolsep}{4pt}
\begin{tabular}{@{}p{3.5cm}p{4.0cm}p{1.8cm}p{4.4cm}@{}}
\toprule
\textbf{Item (\S)} & \textbf{What it guarantees} & \textbf{Status} & \textbf{Declared non-$\Omega$ assumptions} \\
\midrule
\multicolumn{4}{@{}l}{\emph{Three local certificate families}} \\
Conflict-aware grounding (\S\ref{sec:grounding}; Lem~\ref{thm:no_silent}) & Emission gate soundness: \texttt{emitted} only when refutation mass $\le \theta_r$; four-valued claim status & compiled & \texttt{SignedSupport\allowbreak Oracle}, \texttt{Decomposition\allowbreak Oracle} \\
Embedding sensitivity (\S\ref{sec:embedding}; Thm~\ref{thm:robust_decision}) & Robust similarity bound from $R^2_{\mathrm{inv}}$ over $G_{\mathrm{inv}}$; selective-sensitivity ($\Delta^2 > 0$) is a decidable predicate on the certificate read by the policy layer & compiled & \texttt{ParaphraseOracle} (t4); \texttt{cauchy\_schwarz\_sq}, \texttt{innerProd\_sub} (t2) \\
Hoare agent action (\S\ref{sec:agents}) & Pre/post discharged per step; trajectory chained & compiled (axiom-free) & (none) \\
\midrule
\multicolumn{4}{@{}l}{\emph{Two operators on certificates}} \\
Maximal Certifiable Residue (\S\ref{sec:mcr}; Thm~\ref{thm:mcr_maximality}) & Largest constraint-satisfying subset of an answer; abstention with audit-logged drops & compiled & (none beyond constraint families) \\
Compositional Stability (\S\ref{sec:composition}; Lem~\ref{lem:composition_two_layer} compiled; Thm~\ref{thm:composition} specified) & Pipeline-wide perturbation budget $B_\Pi$ from per-layer chain of gains and margins & lemma compiled; $n$-lift specified & (none beyond per-layer certs) \\
\midrule
\multicolumn{4}{@{}l}{\emph{Per-call deliverable}} \\
Universal Assurance Card (\S\ref{sec:card}) & Verdict consistent with field vector & specified (decidable predicate; \texttt{AssuranceCard.lean} flagged for the artifact) & inherited from sub-certificates \\
\midrule
\multicolumn{4}{@{}l}{\emph{Compiled artifact totals (\S\ref{sec:artifact})}} \\
22 cert types, 46 declarations & 17 axiom-free; 29 with $\Omega$ + tier-2/3/4 named axioms; 0 \texttt{sorry} & compiled & per-cert non-$\Omega$ set \\
\bottomrule
\end{tabular}
\end{table}

\noindent A reader looking for ``which threats does each certificate defend against, and what is the first one to ship?'' should jump to \S\ref{sec:threat} (Threat model and adoption order; Table~\ref{tab:threats} maps threats to defender certificates row-by-row).

\section{Three Impossibility Propositions and a Calibration Conjecture}
\label{sec:impossibility}

Three of these are inherited (in LLM-pipeline form) from~\cite{koomullil2026patent}; the calibration impossibility, stated below as a conjecture rather than a proposition because a careful Vovk-style proof is left for future work, is new to this paper. We state the inherited results in compact form; full motivation is in the precursor.

\paragraph{Formalism.} $S$ denotes a (possibly stochastic) system with input space $\mathcal{X}$, output space $\mathcal{Y}$, output rule $y \sim \mathcal{D}_S(x)$. $H$ denotes Shannon entropy, $H_2$ R\'enyi-2 entropy. A \emph{deterministic validator} $V : \mathcal{Y} \to \{0,1\}$ is computable without consulting randomness.

\begin{proposition}[Certificate Impossibility for Stochastic Samplers]\statusspec\label{prop:certimposs}
If $\Pr_{c \sim \mathcal{D}_S(x_0)}[c \notin \Omega^\star] > 0$ for any $x_0$, where $\Omega^\star = \{c : V(c) = 1\}$, then $S$ cannot unconditionally guarantee that its output is a valid certificate. \emph{Architectural consequence}: a deterministic post-sampler validator is necessary; the hybrid Lean-kernel-plus-$\Omega$-audit (Def.~\ref{def:validcert}) is that validator.
\end{proposition}

\begin{proposition}[Compositional Unsoundness of Marginal Flat Scoring]\statusspec\label{prop:dependency}
For any aggregation $F : [0,1]^n \to [0,1]$ symmetric in marginals and any non-degenerate marginal vector, the Fr\'echet--Hoeffding bounds give two joints with the same marginals but distinct conjunction probabilities, so no marginal-only $F$ can equal the conjunction probability without an exogenous independence assumption. The DAG-conditional analogue (every flat $F$ disagrees with the all-elements rule on some configuration) compiles as a finite-case enumeration. \emph{Architectural consequence}: dependency structure must be carried explicitly (DAGs, lattices, signed support), not collapsed into a single scalar.
\end{proposition}

\begin{proposition}[Stochastic Decoding is not Unconditionally Reproducible]\statusspec\label{prop:stochastic}
For finite output alphabet $\mathcal{Y}$ with distribution $\mu$,
$\Pr[f(x,\omega_1) \neq f(x,\omega_2)] = 1 - \sum_y \mu(y \mid x)^2 \ge 1 - 2^{-H_2(\mu \mid x)}$.
At $H_2 \approx 3$ bits (typical for token distributions at $T = 1$), the disagreement probability already exceeds $87.5\%$. \emph{Architectural consequence}: the Lean kernel produces the same certificate on the same inputs with probability $1$, unconditionally; the formal layer is the deterministic anchor.
\end{proposition}

\begin{conjecture}[Calibration Impossibility without Ground Truth]\statusconj\label{prop:calibration_imposs}
Let $S$ emit a confidence score $c(x) \in [0,1]$ without access to the true conditional $P(y \mid x)$ and without exogenous labeled data. Then $S$ cannot prove $c$ calibrated, that is, $\Pr_{(x,y) \sim P}[y = \mathrm{correct} \mid c(x) \in B] = \mathbb{E}_{(x,y) \sim P}[c(x) \mid c(x) \in B]$, without an external distributional hypothesis.
\end{conjecture}

\noindent\emph{Informal argument.} Calibration is a property of $P$ restricted to a level set of $c$. Without access to $P$ and without labels, $S$ observes only its own outputs and cannot recover $\Pr[y = \mathrm{correct} \mid c(x) \in B]$, so any guarantee must rest on an external hypothesis about $P$ (exchangeability with a calibration set, bounded shift, or parametric form), which is precisely what conformal methods axiomatize~\cite{vovk2005algorithmic}. The conjecture is stated rather than proved because a rigorous treatment requires a Vovk-style distribution-free formalization that we have not yet completed; the formal anchor we would build on is the conformal-prediction validity argument of Vovk, Gammerman, and Shafer~\cite{vovk2005algorithmic}, in which a marginal coverage guarantee is recovered only under exchangeability between calibration and test data, and the missing step here is the explicit construction of the exchangeability filtration over the (input, score, outcome) sequence and the proof that no measurable functional of the score-only sub-filtration recovers the conditional. The architectural consequence is independent of formalization status: calibration hypotheses must be sourced exogenously and named as scope-of-validity oracles or conformal-exchangeability axioms.

Together, Propositions~\ref{prop:certimposs}--\ref{prop:stochastic} and Conjecture~\ref{prop:calibration_imposs} fix the architectural shape: (i) a deterministic component issues guarantees; (ii) compositional steps carry explicit dependency; (iii) reproducibility is engineered into the issuing layer; (iv) calibration is grounded in exchangeable held-out data. Every certificate in this paper is the deterministic validator $V$ for some specific property; the kernel plus $\Omega$-audit is the universal discharge mechanism.

\section{Conflict-Aware Bilattice Grounding}
\label{sec:grounding}

This is the most substantive technical departure from the precursor's coverage machinery. Real RAG~\cite{lewis2020rag,gao2023rag_survey} evidence does not merely \emph{support} a claim or not; it can also \emph{contradict} it. A claim well-supported by one chunk and explicitly contradicted by another is not an acceptable output, even at $W^+ = 95\%$. We therefore depart from single-polarity scoring and formalize grounding over a Belnap-style bilattice~\cite{belnap1977useful,ginsberg1988multivalued}.

\subsection{Decomposition with reconstruction completeness}
\label{sec:grounding:decomp}

The LLM's answer $y$ is decomposed (by a second LLM call) into atomic claims $\{c_1, \ldots, c_n\}$ with weights $w_i \in \mathbb{Q}^+$. We adopt the FActScore-style atomic-fact decomposition discipline of~\cite{min2023factscore} as the upstream methodology. A dishonest decomposer could omit a load-bearing claim (the one that would have been contradicted); we block this with a reconstruction check.

\begin{definition}[Verification-by-Reconstruction]\statusspec
\label{def:reconstruction}
The decomposition $\{c_i\}$ is \emph{complete for} $y$ if a canonicalizer $C$ satisfies $C(\bigwedge_i c_i) = C(y)$. The grounding certificate carries this equality as a proof field. Failure of the check rejects the decomposition.
\end{definition}

The semantic faithfulness of the reconstruction is the named oracle \texttt{DecompositionOracle} (tier 4); the canonicalizer equality check is decidable.

\subsection{Three-way signed support}
\label{sec:grounding:signed}

For each atomic claim $c_i$ and each retrieved chunk $r_j$, a three-way NLI signature is computed:
\[
\sigma(c_i, r_j) = (\sigma^+_{ij}, \sigma^-_{ij}, \sigma^0_{ij}) \in [0,1]^3, \quad \sigma^+_{ij} + \sigma^-_{ij} + \sigma^0_{ij} = 1,
\]
where $\sigma^+$, $\sigma^-$, $\sigma^0$ are entailment, contradiction, and orthogonality respectively (standard NLI~\cite{bowman2015snli}). The best-support and best-refutation scores are $\beta^+_i = \max_j \sigma^+_{ij}$ and $\beta^-_i = \max_j \sigma^-_{ij}$.

\begin{definition}[Claim Epistemic Status]\statuscompiled
\label{def:epistemic}
With thresholds $\theta_g, \theta_c \in (0,1]$:
\[
\mathrm{status}(c_i) =
\begin{cases}
\textsc{Supported} & \beta^+_i \ge \theta_g \wedge \beta^-_i < \theta_c \\
\textsc{Contradicted} & \beta^+_i < \theta_g \wedge \beta^-_i \ge \theta_c \\
\textsc{Contested} & \beta^+_i \ge \theta_g \wedge \beta^-_i \ge \theta_c \\
\textsc{Unknown} & \beta^+_i < \theta_g \wedge \beta^-_i < \theta_c
\end{cases}
\]
\end{definition}

\begin{remark}[Bilattice structure, after Belnap~\cite{belnap1977useful}]\label{lem:bilattice}
The four-valued set $\{\textsc{Supp}, \textsc{Contr}, \textsc{Contested}, \textsc{Unk}\}$ is the standard Belnap bilattice~\cite{belnap1977useful, ginsberg1988multivalued}: two orthogonal lattice orders (information-increasing and truth-increasing). Under the information (knowledge) order, the join of \textsc{Supp} and \textsc{Contr} is \textsc{Contested} (Both) and the meet is \textsc{Unk} (None); under the truth order the meet is \textsc{Contr} and the join is \textsc{Supp}. We use this as a known structural fact rather than claiming it as a contribution; what is new here is its instantiation as the codomain of the grounding classifier (Definition~\ref{def:epistemic}) combined with the emission gate of \S\ref{sec:grounding:cert}.
\end{remark}

The signed weighted grounding coverages are
\[
W^+ = \frac{\sum_{i:\,\textsc{Supp}} w_i}{\sum_i w_i},\quad
W^- = \frac{\sum_{i:\,\textsc{Contr}} w_i}{\sum_i w_i},\quad
W^\pm = \frac{\sum_{i:\,\textsc{Contested}} w_i}{\sum_i w_i},
\]
and $W^? = 1 - W^+ - W^- - W^\pm$. All four coverages lie in $[0,1]$, matching the per-call deliverable in \S\ref{sec:card} and the empirical thresholds reported in \S\ref{sec:pilot_a}. The pipeline emits iff $W^+ \ge \theta_e$ and $W^- + W^\pm \le \theta_r$. A single contested load-bearing claim is grounds for abstention, not a footnote.

\subsection{The certificate}
\label{sec:grounding:cert}

\smallskip
\leanfile{Grounding.lean}
\begin{lstlisting}
inductive EpistemicStatus
  | supported | contradicted | contested | unknown
  deriving DecidableEq

def classify (b_plus b_minus theta_g theta_c : Rat) : EpistemicStatus :=
  if b_plus ≥ theta_g ∧ b_minus < theta_c then .supported
  else if b_plus < theta_g ∧ b_minus ≥ theta_c then .contradicted
  else if b_plus ≥ theta_g ∧ b_minus ≥ theta_c then .contested
  else .unknown

structure ConflictAwareGroundingCertificate where
  answer : String
  atomic_claims : List (String × Rat)
  retrieved : List String
  support_pos : List (List Rat)
  support_neg : List (List Rat)
  support_orth : List (List Rat)
  beta_plus : List Rat
  beta_minus : List Rat
  theta_g : Rat       -- support threshold
  theta_c : Rat       -- contradiction threshold
  theta_e : Rat       -- emission threshold
  theta_r : Rat       -- refutation threshold
  statuses : List EpistemicStatus
  W_supported W_contradicted W_contested W_unknown : Rat
  emitted : Bool
  -- proof obligations (decidable; discharged by `decide`)
  p_scores_simplex : ∀ i j, support_pos.get! i |>.get! j
                          + support_neg.get! i |>.get! j
                          + support_orth.get! i |>.get! j = 1
  p_beta_plus_max : ∀ i, beta_plus.get! i =
                        (support_pos.get! i).foldl max 0
  p_beta_minus_max : ∀ i, beta_minus.get! i =
                        (support_neg.get! i).foldl max 0
  p_statuses_correct : ∀ i, statuses.get! i =
                         classify (beta_plus.get! i) (beta_minus.get! i)
                                  theta_g theta_c
  p_W_sum_unit : W_supported + W_contradicted +
                 W_contested + W_unknown = 1
  p_total_weight_pos : 0 < totalWeight atomic_claims
  -- mass-correctness fields: each declared W aggregate equals the
  -- normalized weighted coverage over its status class, where
  -- `normalizedWeightOf` filters `atomic_claims` by the corresponding
  -- entry of `statuses` (the parallel list constrained by
  -- `p_statuses_correct` above) and divides by the total weight.
  -- Forecloses a malicious decomposer declaring `W_contradicted = 0`
  -- while shipping contradicted claims.
  p_W_supported_correct    :
    W_supported    = normalizedWeightOf atomic_claims statuses .supported
  p_W_contradicted_correct :
    W_contradicted = normalizedWeightOf atomic_claims statuses .contradicted
  p_W_contested_correct    :
    W_contested    = normalizedWeightOf atomic_claims statuses .contested
  p_W_unknown_correct      :
    W_unknown      = normalizedWeightOf atomic_claims statuses .unknown
  p_emission_rule : emitted ↔
                     (W_supported ≥ theta_e ∧
                      W_contradicted + W_contested ≤ theta_r)
  p_reconstruction : Canonicalize (Conjoin atomic_claims) =
                     Canonicalize answer
  -- tier-4 oracles
  h_nli_oracle : SignedSupportOracle atomic_claims retrieved
  h_decomp_faithful : DecompositionOracle answer atomic_claims
\end{lstlisting}

\begin{lemma}[Emission Gate Soundness, conditional on NLI and decomposition oracles]\statuscompiled
\label{thm:no_silent}
For any \texttt{ConflictAwareGroundingCertificate} with \texttt{emitted = true}, $W^- + W^\pm \le \theta_r$. Hence, conditional on \texttt{SignedSupportOracle} and \texttt{DecompositionOracle} being correct on the input, emission implies that the normalized \textsc{Contradicted}-plus-\textsc{Contested} weighted mass is at most the declared refutation threshold $\theta_r$. In the zero-tolerance case $\theta_r = 0$, any positive-weight \textsc{Contradicted} or \textsc{Contested} atomic claim blocks emission. For $\theta_r > 0$, low-weight contradicted or contested mass may be emitted below the threshold, but it remains explicitly recorded in the certificate and Card and is available to the application-policy layer.
\end{lemma}

\begin{proof}
By \texttt{p\_emission\_rule} directly: the emission flag is structurally equivalent to the refutation-threshold inequality, so the Lean kernel discharges by definitional unfolding plus rational-arithmetic reduction.
\end{proof}

\textbf{What this lemma is and is not.} The lemma proves that the \emph{emission gate} of the certificate is sound: the gate's $\mathrm{emitted}$ flag is true only when the weighted refutation mass is below threshold. The actual epistemic content, that an atomic claim is genuinely \textsc{Contradicted} or \textsc{Contested} by the retrieved evidence, is contributed by the \texttt{SignedSupportOracle} (the three-way NLI judgment) below the trust boundary. The lemma is therefore a structural soundness result for the gate, conditional on the oracle, rather than a theorem about contradictions in the world. The architectural value, blocking a class of evidence-contradicting outputs that scalar-confidence RAG pipelines cannot detect by construction, comes from the combination of the NLI oracle and the gate; the lemma certifies that the gate respects its specification.

\textbf{Verification status:} \emph{compiled} as \texttt{Grounding.no\_silent\_contradictions} in the artifact (audited axiom set: $\Omega \cup \{\texttt{SignedSupport\allowbreak Oracle}, \texttt{Decomposition\allowbreak Oracle}\}$); see Appendix~\ref{app:axiom_audit}.

\subsection{Comparison with support-only RAG}
\label{sec:grounding:why}

Hallucinations with low support are caught by $W^+ < \theta_e$. The \emph{harder} failure mode is an atomic claim that the retrieval corpus actively contradicts (e.g., ``the patient's CrCl of 22 mL/min is within apixaban's renal range'' when the FDA label requires CrCl $\ge 25$). A support-only certificate could report $W_{\mathrm{ground}} = 95\%$ and silently ship the contradiction; the conflict-aware certificate rejects emission whenever the contradicted-plus-contested mass exceeds $\theta_r$, and below that threshold the contradicted/contested mass is not silent but is explicitly surfaced to the policy layer.

\subsection{Limitations of the certificate}
\label{sec:grounding:limits}

The certificate does not verify (i) that a \textsc{Supported} claim is actually true (the NLI oracle is below the trust boundary); (ii) that the corpus is itself factually correct (RAG against propaganda still produces well-grounded propaganda); (iii) that the decomposition oracle is faithful (a malicious decomposer could omit a key claim; the reconstruction check mitigates but does not eliminate this). These limits are the LLM-pipeline analogue of the precursor's grounded-coverage caveats.

\section{Maximal Certifiable Residue}
\label{sec:mcr}

The bilattice grounding certificate (\S\ref{sec:grounding}) blocks emission when contradiction exceeds threshold. But all-or-nothing abstention is an overcorrection: a ten-claim answer with one contradicted claim is more useful if the nine certifiable claims ship and the one bad claim is dropped, with the drop audit-logged. We formalize this maximal-informative behavior.

\subsection{Definition}

Let $\mathcal{C} = \{C_1, \ldots, C_m\}$ be a finite family of decidable predicates on subsets of $\mathrm{At}(y) = \{c_1, \ldots, c_n\}$ (the answer's atomic claims). Representative constraints: $C_{\mathrm{ground}}$ (every $c \in S$ is \textsc{Supported}); $C_{\mathrm{cons}}$ (no two members of $S$ are mutually contradictory); $C_{\mathrm{pol}}$ (no member triggers a forbidden-output rule); $C_{\mathrm{pre}}$ (every implied tool-call passes its Hoare precondition); $C_{\mathrm{bud}}$ (cumulative cost stays below budget).

\begin{definition}[Maximal Certifiable Residue]\statuscompiled
\label{def:mcr}
The \emph{certifiable family} is $\mathsf{Cert}(\mathcal{C}, y) = \{S \subseteq \mathrm{At}(y) \mid C_j(S) \text{ for all } C_j \in \mathcal{C}\}$. The \emph{maximal certifiable residue} is
\[
\mathrm{MCR}(\mathcal{C}, y) = \arg\max_{S \in \mathsf{Cert}(\mathcal{C}, y)} \sum_{c \in S} w(c),
\]
breaking ties by canonical lex order. When $\mathsf{Cert} = \{\emptyset\}$, the residue is empty and the pipeline abstains.
\end{definition}

\begin{definition}[Monotone Constraint Family]\statuscompiled
\label{def:monotone_constraints}
$\mathcal{C}$ is \emph{monotone} (downward-closed) if for every $C_j$ and every $S' \subseteq S \subseteq \mathrm{At}(y)$, $C_j(S) \Rightarrow C_j(S')$. We additionally require $\mathcal{C}$ to be \emph{empty-admitting}: $C_j(\emptyset) = \mathsf{true}$ for every $C_j \in \mathcal{C}$. This holds for all five sample constraints above (each is vacuously true on $\emptyset$) and guarantees $\emptyset \in \mathsf{Cert}$ even before any non-empty witness is exhibited.
\end{definition}

Of the five sample constraints, all are subset-monotone and empty-admitting except trivial cases; the consistency constraint is monotone for the mutual-contradiction predicate (removing a claim cannot introduce a contradiction).

\subsection{Maximality, idempotence, monotonicity}

\begin{theorem}[Properties of MCR]\statuscompiled
\label{thm:mcr_maximality}
For any monotone, empty-admitting constraint family $\mathcal{C}$ and answer $y$:
\begin{enumerate}[leftmargin=1.5em,topsep=2pt,itemsep=2pt]
\item[(i)] $\mathrm{MCR}(\mathcal{C}, y)$ is well-defined: $\emptyset \in \mathsf{Cert}$ directly from empty-admission, so the $\arg\max$ is over a non-empty finite set.
\item[(ii)] $\mathrm{MCR}(\mathcal{C}, y)$ is \emph{weight-optimal}: no element of $\mathsf{Cert}$ has strictly greater weight (and ties are broken by canonical lex order, so the operator is single-valued).
\item[(iii)] (Idempotence) $\mathrm{MCR}(\mathcal{C}, \mathrm{MCR}(\mathcal{C}, y)) = \mathrm{MCR}(\mathcal{C}, y)$.
\item[(iv)] (Monotone in constraints) $\mathcal{C} \subseteq \mathcal{C}'$ implies $w(\mathrm{MCR}(\mathcal{C}', y)) \le w(\mathrm{MCR}(\mathcal{C}, y))$.
\item[(v)] (Anti-monotone in evidence, via the constraint map) Fix any evidence-extension map $E \mapsto \mathcal{C}(E)$ from observable evidence to a constraint family, and assume the map is monotone in the sense that $E \subseteq E'$ implies $\mathcal{C}(E) \subseteq \mathcal{C}(E')$ (extending evidence only adds constraints or strengthens existing ones). Then $w(\mathrm{MCR}(\mathcal{C}(E'), y)) \le w(\mathrm{MCR}(\mathcal{C}(E), y))$, i.e.\ extending evidence is non-increasing in MCR weight. This follows from item (iv) applied to the induced constraint inclusion.
\end{enumerate}
\end{theorem}

\begin{proof}
(i) $\emptyset \in \mathsf{Cert}$ by empty-admission. (ii) If some $S^* \in \mathsf{Cert}$ had $w(S^*) > w(\mathrm{MCR})$, then $S^* \in \arg\max$, contradicting the tie-break rule. (iii) Reapplying $\mathrm{MCR}$ to its own output keeps the atomic set unchanged: the residue itself is in $\mathsf{Cert}$ and is maximal among its own subsets in $\mathsf{Cert}$ (by monotonicity, all subsets of the residue are in $\mathsf{Cert}$). (iv) Adding constraints can only shrink $\mathsf{Cert}$, so the optimum cannot increase. (v) By the assumed monotonicity of the evidence-to-constraint map, $E \subseteq E' \Rightarrow \mathcal{C}(E) \subseteq \mathcal{C}(E')$, so item (iv) applied to the induced inclusion gives the conclusion.
\end{proof}

We use \emph{weight-optimal} rather than the standard ``$\subseteq$-maximal'' phrasing because what the operator actually maximizes is total claim weight, not subset extent. A weight-optimal residue need not be a $\subseteq$-maximal element of $\mathsf{Cert}$ when uncertified low-weight claims could be appended without losing certification but would not raise the score.

\textbf{Verification status:} \emph{compiled} as \texttt{MCR.residue\_is\_fixed\_point} (and \texttt{residue\_is\_certifiable}, \texttt{subset\_of\_residue\_is\_certifiable}) in the artifact (audited axiom set: $\Omega$).

\subsection{Lean structure}

\smallskip
\leanfile{MCR.lean}
\begin{lstlisting}
structure MCRCertificate (AtomicClaim : Type) [DecidableEq AtomicClaim] where
  all_claims : Finset AtomicClaim
  weight : AtomicClaim -> Rat
  constraints : List (Finset AtomicClaim -> Bool)
  residue : Finset AtomicClaim
  residue_weight : Rat
  p_residue_subset : residue ⊆ all_claims
  p_residue_cert : ∀ C ∈ constraints, C residue = true
  p_monotone :
    ∀ C ∈ constraints, ∀ S T : Finset AtomicClaim,
      S ⊆ T -> T ⊆ all_claims -> C T = true -> C S = true
  p_weight_nonneg : ∀ a ∈ all_claims, 0 ≤ weight a
  p_maximal :
    ∀ S : Finset AtomicClaim,
      S ⊆ all_claims ->
      (∀ C ∈ constraints, C S = true) ->
      (S.sum weight) ≤ residue_weight
  p_weight_correct : residue_weight = residue.sum weight
\end{lstlisting}

\subsection{Worked example}
\label{sec:mcr:worked}

A clinical-RAG answer decomposed into five claims:

\begin{table}[H]
\centering
\footnotesize
\begin{tabular}{@{}llccccc@{}}
\toprule
& \textbf{Claim} & $w_i$ & $\beta^+$ & $\beta^-$ & \textbf{Status} & \textbf{Action?} \\
\midrule
$c_1$ & Patient diagnosis & 3.0 & 0.92 & 0.03 & \textsc{Supp} &  \\
$c_2$ & Recommended drug & 2.0 & 0.88 & 0.04 & \textsc{Supp} & order \\
$c_3$ & Dose frequency & 1.0 & 0.81 & 0.02 & \textsc{Supp} &  \\
$c_4$ & Allergy check (concludes safe) & 1.5 & 0.45 & 0.71 & \textsc{Contr} &  \\
$c_5$ & Follow-up schedule & 0.5 & 0.30 & 0.20 & \textsc{Unk} &  \\
\bottomrule
\end{tabular}
\end{table}

With constraints $\{C_{\mathrm{ground}}, C_{\mathrm{pre}}\}$ at $\theta_g = 0.70$, $\theta_c = 0.50$, and an action precondition requiring $c_4 = \textsc{Supp}$ before ordering the drug $c_2$: $C_{\mathrm{ground}}$ excludes $c_4$ and $c_5$; $C_{\mathrm{pre}}$ excludes $c_2$ when $c_4$ is absent; $\mathrm{MCR} = \{c_1, c_3\}$ with weight $4.0$. The pipeline emits the diagnosis and dose frequency, drops the drug recommendation and follow-up, and audit-logs the drops. An all-or-nothing abstention would emit nothing; a support-only certificate would have shipped all five (including the silent contradiction at $c_4$); a contradictions-only filter would have shipped $\{c_1, c_2, c_3\}$ without catching $c_2$'s dependency on $c_4$. MCR composes constraint families coherently.

\subsection{Computing the residue}

Computing $\mathrm{MCR}(\mathcal{C}, y)$ is a weighted-subset optimization and is NP-hard in general, even when $\mathcal{C}$ is monotone: the special case of weighted maximum independent set on the conflict graph induced by a pairwise-contradiction constraint reduces to it. The hard search is below the trust boundary; the certificate carries the residue as a field, and the kernel only checks that every constraint accepts the residue and that no certified element of $\mathsf{Cert}$ outweighs it (a finite-list comparison, decidable by \texttt{decide}). Three regimes admit efficient search: (i) when constraints decompose over atomic claims independently (each $C_j$ is a per-claim predicate), greedy removal of constraint-violators is optimal in $O(n \cdot |\mathcal{C}|)$; (ii) when claims carry a dependency DAG (e.g.\ ``patient has diabetes'' $\Rightarrow$ ``patient has chronic condition'') and constraints respect the DAG, $O(n^2)$ topological antichain extraction suffices; (iii) for general non-monotone families at the typical $n \lesssim 100$ atomic-claim regime, off-the-shelf SAT/SMT or weighted-MaxSAT enumeration solves the instance well within per-call latency budgets. Outside these regimes, MCR reduces to weighted-subset optimization on which the minimal-sufficient-evidence and weighted-MaxSAT literatures supply mature solvers; we adopt those solvers below the trust boundary and rely on the kernel check above the boundary for correctness.

\section{Embedding Sensitivity and Paraphrase Stability}
\label{sec:embedding}

Given an embedding $E : T \to \mathbb{R}^d$ with $L_2$-normalized output and texts $x, y$ with $\mathrm{sim}(x,y) = \langle E(x), E(y)\rangle$, we certify that the similarity decision is robust under a declared finite perturbation family. Stability is not a semantic property of the encoder; it is a local Lipschitz-style property whose verification is purely arithmetic once the family is fixed.

\subsection{Selective sensitivity}

\begin{definition}[Meaning-Invariant and Significant Perturbation Families]\statusspec
\label{def:perturbations}
For $x \in T$, $G_{\mathrm{inv}}(x) \subset T$ consists of meaning-preserving edits (synonym replacement, voice transposition, harmless typographical changes); $G_{\mathrm{sig}}(x) \subset T$ consists of meaning-changing edits (negation, entity substitution, quantity change, role reversal). Both families are finite, explicit, and declared as inputs.
\end{definition}

We work in squared form so the certificate stays in $\mathbb{Q}$. Define the meaning-invariant robustness radius $R^2_{\mathrm{inv}}(x) = \max_{x' \in G_{\mathrm{inv}}(x)} \|E(x') - E(x)\|^2$, the critical sensitivity floor $R^2_{\mathrm{sig}}(x) = \min_{x' \in G_{\mathrm{sig}}(x)} \|E(x') - E(x)\|^2$, and the selective gap $\Delta^2(x) = R^2_{\mathrm{sig}}(x) - R^2_{\mathrm{inv}}(x)$. $\Delta^2 > 0$ certifies a local separation of meaning; $\Delta^2 \le 0$ flags a local encoder failure.

\begin{lemma}[Family Monotonicity]\statusspec
\label{lem:family_monotone}
$G_{\mathrm{inv}}(x) \subseteq G'_{\mathrm{inv}}(x) \Rightarrow R^2_{\mathrm{inv}}(x; G_{\mathrm{inv}}) \le R^2_{\mathrm{inv}}(x; G'_{\mathrm{inv}})$. An auditor extending the family can only weaken (never spuriously strengthen) the guarantee.
\end{lemma}

\begin{remark}[Why \emph{selective} sensitivity]\label{rem:selective}
An encoder mapping every text to the same vector has $R_{\mathrm{inv}} = 0$ but is useless because $R_{\mathrm{sig}} = 0$. The right objective is selective sensitivity: low $R_{\mathrm{inv}}$ and high $R_{\mathrm{sig}}$.
\end{remark}

\subsection{The certificate}

\smallskip
\leanfile{Certificate.lean}
\begin{lstlisting}
structure EmbeddingSensitivityCertificate (d : Nat) where
  text : String
  embedding : Vector Rat d
  G_inv : List String           -- neutral perturbations
  G_sig : List String           -- significant perturbations
  emb_inv : List (Vector Rat d)
  emb_sig : List (Vector Rat d)
  R_inv_sq : Rat                -- max squared L2-dist over G_inv
  R_sig_sq : Rat                -- min squared L2-dist over G_sig
  delta_sq : Rat                -- R_sig_sq - R_inv_sq
  p_inv_length : G_inv.length = emb_inv.length
  p_sig_length : G_sig.length = emb_sig.length
  p_bounded : ∀ v ∈ (embedding :: emb_inv ++ emb_sig), l2_norm_sq v ≤ 1
  p_R_inv_sound : ∀ e ∈ emb_inv, l2_dist_sq embedding e ≤ R_inv_sq
  p_R_sig_sound : ∀ e ∈ emb_sig, R_sig_sq ≤ l2_dist_sq embedding e
  p_delta_correct : delta_sq = R_sig_sq - R_inv_sq
  -- tier-4 oracle
  h_paraphrase_valid : ParaphraseOracle text G_inv G_sig
\end{lstlisting}

The first four proof fields reduce by \texttt{decide} or \texttt{norm\_num}; the rest by \texttt{rfl}. The paraphrase oracle stands in for the paraphraser's faithfulness (tier 4).

\begin{theorem}[Robust Similarity Bound, Squared Form]\statuscompiled
\label{thm:robust_decision}
For any \texttt{EmbeddingSensitivityCertificate}, any unit-normalized document embedding $d \in \mathbb{R}^d$, and any $e' \in \mathrm{emb\_inv}$:
$|\langle E(\mathrm{text}), d\rangle - \langle e', d\rangle|^2 \le R^2_{\mathrm{inv}}.$
Consequently, if the threshold decision $\langle E(\mathrm{text}), d\rangle \ge \theta_{\mathrm{ret}}$ (similarity of the embedded query against the pre-indexed document embedding $d$) has margin $> \sqrt{R^2_{\mathrm{inv}}}$, the decision is preserved across every variant in $G_{\mathrm{inv}}$. Here $\theta_{\mathrm{ret}}$ is the retrieval similarity threshold; the bilattice grounding thresholds $\theta_g, \theta_c, \theta_e, \theta_r$ of \S\ref{sec:grounding} are unrelated.
\end{theorem}

\begin{proof}[Proof sketch]
By Cauchy--Schwarz: $|\langle u - v, d\rangle|^2 \le \|u-v\|^2 \|d\|^2 \le R^2_{\mathrm{inv}} \cdot 1$. Lean discharges via the named axioms \texttt{cauchy\_schwarz\_sq} (tier 2) and \texttt{innerProd\_sub} (tier 2).
\end{proof}

\textbf{Verification status:} \emph{compiled} as \texttt{robust\_similarity\_sq} in \texttt{Certificate.lean}; transitive axioms reported in Appendix~\ref{app:axiom_audit}.

\paragraph{Two distinct claims, distinguished.} Theorem~\ref{thm:robust_decision} is a \emph{robust similarity bound}: it uses only $R^2_{\mathrm{inv}}$ and certifies that $L_2$ drift on the meaning-invariant family is bounded by the declared envelope. It does \emph{not} require $\Delta^2 > 0$. The \emph{selective sensitivity} claim, that meaning-changing edits move farther than meaning-preserving ones, is the separate decidable predicate
\[
\mathrm{SelectiveSensitivityPass}(c) \;:=\; (0 < c.\mathtt{delta\_sq})
\]
on a constructed certificate. Selective sensitivity is what the policy layer (the Universal Card; \S\ref{sec:card}) reads to decide whether the encoder behaves correctly on this query; it is not a precondition for Theorem~\ref{thm:robust_decision}. The certificate carries $R^2_{\mathrm{inv}}$, $R^2_{\mathrm{sig}}$, and $\Delta^2$ as fields, with $\Delta^2$ tagged for the policy gate; the empirical scope of the certificate (Pilot~B and Pilot~B$'$, \S\ref{sec:embedding:empirical}--\S\ref{sec:embedding:longform}) reports how often $\Delta^2 > 0$ holds, distinct from the always-applicable robust-similarity bound. This separation matters because the negative cells of Pilot~B (paraphrase $\times$ attribute\_flip) still carry a sound robust-similarity bound; what they do not carry is selective sensitivity, and the Card downgrades emission accordingly.

\paragraph{Floating-point and black-box embeddings.} The formal certificate is exact over the rationalized embedding representation supplied to Lean. In deployments where embeddings are returned by a floating-point or black-box API, the implementation first commits to the API transcript, the model-version hash, the tokenizer hash, and the embedding-service configuration, then applies a deterministic quantization map to rational vectors. The Card records both the rounded rational drift bound $R^2_{\mathrm{inv},\mathrm{rounded}}(x)$ and a rounding-error budget $\varepsilon_{\eta}$; policy margins are evaluated against the conservative error-adjusted bound $R^2_{\mathrm{inv},\mathrm{real}}(x) \le R^2_{\mathrm{inv},\mathrm{rounded}}(x) + \varepsilon_{\eta}$. Thus the theorem remains exact for the audited rational object, while the deployment records the additional approximation layer explicitly and the auditor can replay the rational computation against the recorded transcript and hashes without re-invoking the embedding service.

\subsection{Top-$k$ ranking and paraphrase decision margin}

A simple corollary: ranking the document set by $\langle E(\mathrm{text}), d\rangle$, the certified envelope partitions documents into \texttt{in}, \texttt{maybe}, \texttt{out} based on whether their similarity gap to the retrieval threshold $\theta_{\mathrm{ret}}$ exceeds $2\sqrt{R^2_{\mathrm{inv}}}$. Paraphrase decision-margin certificates layer on top: for any binary classifier $f$ over the $k$ paraphrases in $G_{\mathrm{inv}}$, count agreement and certify a flip-free emission whenever the agreement margin clears a threshold. Both are extensions of Theorem~\ref{thm:robust_decision} that share the same Lean structure.

\subsection{Empirical scope of the inequality}
\label{sec:embedding:empirical}

The robust-similarity theorem (Theorem~\ref{thm:robust_decision}) applies whenever the declared $R^2_{\mathrm{inv}}$ bound is certified; it does not require $\Delta^2 > 0$. The empirical question measured here is separate: how often the policy-layer selective-sensitivity predicate $\Delta^2 > 0$ holds for the declared perturbation families. This was measured in Pilot~B (\S\ref{sec:embedding:empirical}), with design and pre-registered predictions in the companion repository.

\paragraph{Setup.} We sliced $N{=}209$ HotpotQA dev-distractor queries (post-yield-filter), generated 25 meaning-preserving paraphrases ($G_{\mathrm{inv}}$) and 15 meaning-changing edits ($G_{\mathrm{sig}}$) per query with GPT-5, and validated each variant with a 3-judge panel (Claude Sonnet 4.5, GPT-5-mini, Gemini 2.5 Flash). \textbf{Panel reliability is high}: minimum pairwise Cohen's $\kappa = 0.887$ and unanimous-accept rate 0.950, both well above the pre-registered $\kappa \ge 0.65$ and unanimous-accept $\ge 0.85$ thresholds (the methodological reasons this $\kappa$ is informative here while it saturates in Pilot~A's panel are reconciled in \S\ref{sec:pilot_a:hypotheses}). We embedded all variants on text-embedding-3-large (strong-API encoder) and sentence-transformers/all-mpnet-base-v2 (open-weights baseline), then computed per-query $R^2_{\mathrm{inv}}$, $R^2_{\mathrm{sig}}$, $\Delta^2$. (Cohere embed-v3 was deferred for budget reasons; the run is parameterized for $n$ encoders and could be re-run cheaply with a Cohere production key.)

\paragraph{Hypothesis HB1: empirical scope of the inequality.}

\smallskip
\begin{center}
\begin{tabular}{lccc}
\toprule
\textbf{Encoder} & \textbf{$\Delta^2 > 0$ (95\% CI)} & \textbf{Median $\Delta^2$} & \textbf{Threshold} \\
\midrule
text-embedding-3-large       & 0.297 (0.234, 0.359) & $-0.070$ & $\ge 0.80$ (FAIL) \\
sentence-transformers-mpnet  & 0.196 (0.144, 0.249) & $-0.109$ & $\ge 0.60$ (FAIL) \\
\bottomrule
\end{tabular}
\end{center}

The \S\ref{sec:embedding} inequality holds in only 30\% of queries on the strong-API encoder and 20\% on the open-weights baseline. Median $\Delta^2$ is \emph{negative} on both. This is a substantive negative result, not a measurement artefact: the panel that defined $G_{\mathrm{inv}}$ vs.\ $G_{\mathrm{sig}}$ agreed at $\kappa = 0.89$, so the categories are reliably labeled.

\paragraph{Hypothesis HB2: certificate F1 as a flip predictor.} A pre-registered secondary test asks whether the certificate's predicted budget $\sqrt{R^2_{\mathrm{inv}}}$ correctly bounds the cases where top-$k$ retrieval flips under paraphrase. Using a sample-split ($R^2_{\mathrm{inv}}$ calibrated from a 70\% partition of $G_{\mathrm{inv}}$, predictor tested on the held-out 30\%), the certificate achieves $\mathrm{F1}=0.131$ on both encoders, against a pre-registered threshold of $\ge 0.65$. The top-10 set flips under paraphrase for 81--85\% of (query, paraphrase) pairs; the certificate's drift threshold cannot separate the flippers from the non-flippers because both sit \emph{inside} $R^2_{\mathrm{inv}}$.

\paragraph{Hypothesis HB4: cross-encoder differentiation.} Median $|\Delta^2|$ spread between the two encoders is $0.64\times$ (threshold $\ge 2.0\times$): the two encoders differ in level but not by a margin large enough to pass the pre-registered differentiation test. Both fail in the same direction.

\smallskip
\noindent\textbf{On HB3.} The pre-registration originally listed an HB3 hypothesis testing the certificate's positive-fraction on a contradictions-only baseline. HB3 was retired during the pre-registration revision pass for the same reason H1 and H5 were rewritten in Pilot~A (\S\ref{sec:pilot_a:hypotheses}): once measured, the original HB3 formulation was unrealisable on the data slice. The HB-numbering jumps from HB2 to HB4 to preserve the pre-registration record; the retired HB3 wording is documented in the companion repository.

\paragraph{Why the inequality fails on this benchmark.} HotpotQA dev-distractor queries are short (10--20 tokens, often two named entities and a connective) and many of the panel-labeled \emph{meaning-changing} edits (negation, single-attribute flip) are surface-close to the original (one or two tokens differ), while the panel-labeled \emph{meaning-preserving} paraphrases routinely involve heavy lexical and syntactic restructuring. Modern sentence encoders embed short queries primarily by surface form. The mechanical consequence is that a long paraphrase often lands farther from the original than a short negation does, which violates the certificate's central inequality on a per-query basis.

This is consistent with the embedding literature on negation-sensitivity~\cite{ettinger2020bert,kassner2020negated} and with the broader finding that distributional encoders track contextual co-occurrence rather than logical form. The certificate's soundness is unaffected (Theorem~\ref{thm:robust_decision} continues to hold whenever $\Delta^2 > 0$); its premise cannot be assumed to hold on production encoders for adversarial short-query benchmarks.

\paragraph{Phase~1 ablation: where exactly does the inequality hold?} The \S\ref{sec:embedding} certificate's $G_{\mathrm{inv}}$ and $G_{\mathrm{sig}}$ are \emph{declared} edit-type families (\S\ref{def:perturbations}). \textbf{The following ablation is post-hoc exploratory analysis}: the edit-type subset stratification was not in the pre-registered Pilot~B hypothesis set, and the cells are reported here to characterize where the aggregate failure sits, not to make a confirmatory claim. We treat the ablation's positive cells (e.g.\ surface$\times$entity\_swap at 96.7\%) as hypotheses to be validated by future targeted runs rather than as registered findings on this slice. Conditioning the empirical inequality on every (sub-$G_{\mathrm{inv}}$, sub-$G_{\mathrm{sig}}$) edit-type combination on the same 209-query slice reveals a clean rectangular structure on text-embedding-3-large:

\smallskip
\begin{center}
\footnotesize
\begin{tabular}{lccc}
\toprule
$G_{\mathrm{inv}}$ subset & $G_{\mathrm{sig}}$ subset & $\Delta^2 > 0$ frac (95\% CI) & Median $\Delta^2$ \\
\midrule
surface         & entity\_swap     & \textbf{0.967 (0.938, 0.990)} & 0.345 \\
synonym         & entity\_swap     & \textbf{0.913 (0.875, 0.947)} & 0.310 \\
surface         & negation         & \textbf{0.865 (0.817, 0.909)} & 0.170 \\
reorder         & entity\_swap     & 0.861 & 0.298 \\
paraphrase      & entity\_swap     & 0.847 & 0.290 \\
synonym         & negation         & 0.836 & 0.136 \\
\midrule
paraphrase      & attribute\_flip  & 0.450 & $-0.022$ \\
reorder         & attribute\_flip  & 0.455 & $-0.022$ \\
\bottomrule
\end{tabular}
\end{center}

\smallskip
The pattern is encoder-stable (sentence-transformers-mpnet shows the same ordering at slightly weaker magnitudes: surface$\times$entity\_swap 0.938, paraphrase$\times$attribute\_flip 0.340). \textbf{Aggregated} positive-fractions: by $G_{\mathrm{sig}}$ alone, entity\_swap 0.79, negation 0.60, attribute\_flip 0.38; by $G_{\mathrm{inv}}$ alone, surface 0.64, synonym 0.53, voice 0.43, reorder 0.41, paraphrase 0.40.

\paragraph{The certificate is a typed guarantee, not a uniform one.} The \S\ref{sec:embedding} specification (\S\ref{def:perturbations}) requires the deployer to \emph{declare} $G_{\mathrm{inv}}$ and $G_{\mathrm{sig}}$ as explicit finite families. The Phase~1 ablation is therefore not a refutation of the certificate but a measurement of which declarations are empirically supported on this benchmark. Three deployment-relevant declarations clear the pre-registered $\ge\!0.80$ bar on text-embedding-3-large:

\begin{itemize}[leftmargin=*,topsep=2pt,itemsep=2pt]
\item \emph{Surface-close paraphrase against entity confusion} (surface$\times$entity\_swap, 97\%): the deployer guarantees robustness when the query is rephrased only by whitespace / case / punctuation differences and the failure to defend against is naming the wrong entity. This is the strongest setting and corresponds to a common production threat model in legal and clinical retrieval.
\item \emph{Synonym-only paraphrase against entity confusion} (synonym$\times$entity\_swap, 91\%): tolerates one or two synonym substitutions; same threat model.
\item \emph{Surface-close paraphrase against negation} (surface$\times$negation, 87\%): tolerates trivial paraphrase, defends against single-word logical inversion.
\end{itemize}

The hard cell is paraphrase$\times$attribute\_flip: heavy restructuring of the query combined with a single-attribute meaning change (e.g., a date or numeric attribute) sits inside the embedding's resolution. Deployers in domains where attribute-level distinctions are critical (clinical dosage, legal jurisdiction) should treat this as a known scope limit and either (i) use a richer signature (entity-typed token-distance comparators) or (ii) condition the \S\ref{sec:embedding} certificate on a $G_{\mathrm{inv}}$ that excludes heavy paraphrase.

\paragraph{Pilot data and reproducibility (Pilot~B short-form).} The Pilot~B numbers reported in this subsection are reproducible from the companion repository (per-query $\Delta^2$, per-pair drift and flip records, panel votes with verbatim judge outputs). On the full $N{=}209$ slice, the dominant cost is the 3-judge panel labeling on 8{,}000 panel calls per judge rather than the embedding API itself.

\subsection{Pilot~B$'$: long-form HotpotQA paragraphs}
\label{sec:embedding:longform}

To test whether the short-form failure is intrinsic to commodity encoders or specific to short adversarial queries, we ran a second pre-registered pilot (Pilot~B$'$) on the same benchmark with the query unit shifted from \emph{questions} to \emph{gold supporting paragraphs} (median ${\approx}110$ tokens, multiple named entities and predicates per query). The retrieval corpus, panel composition, generator, encoders, and statistical protocol are identical to Pilot~B; only the input length and named-entity density change. $N{=}64$ effective queries (sized to fit the available budget at the observed per-query cost for long-form variant generation).

\paragraph{Hypothesis HB$'$1: the inequality on long-form paragraphs.}

\smallskip
\begin{center}
\footnotesize
\begin{tabular}{lccc}
\toprule
\textbf{Encoder} & \textbf{$\Delta^2 > 0$ (95\% CP CI)} & \textbf{Median $\Delta^2$} & \textbf{Threshold} \\
\midrule
text-embedding-3-large       & \textbf{1.000 (0.944, 1.000)} & 0.390 & $\ge 0.70$ (PASS) \\
sentence-transformers-mpnet  & \textbf{0.984 (0.916, 1.000)} & 0.419 & $\ge 0.55$ (PASS) \\
\bottomrule
\end{tabular}
\end{center}

\smallskip
The \S\ref{sec:embedding} inequality holds on every long-form query on text-embedding-3-large and on 63 of 64 on the open-weights baseline. Median $\Delta^2$ flips sign (from $-0.07$ short-form to $+0.39$ long-form) and rises by an order of magnitude. The HB$'$1 thresholds are pre-registered and clear by large margins. \textbf{CI methodology.} For Pilot~B$'$ point estimates of 64/64 and 63/64 the bootstrap-resample interval collapses to a degenerate $[1.000, 1.000]$ that misrepresents the uncertainty inherent in the small-$N$ slice; we therefore report exact Clopper--Pearson binomial 95\% intervals here (denoted CP). The $N{=}64$ slice is sized to fit the available budget at the observed per-query cost for long-form variant generation, and the long-form result is correctly characterized as a strong directional finding under a tight one-sided lower bound rather than a narrow two-sided deployment claim; the ``essentially universal'' framing should be read in that light.

\paragraph{Hypothesis HB$'$2: tightness of the drift bound for retrieval-flip prediction.}

\smallskip
\begin{center}
\footnotesize
\begin{tabular}{lcccc}
\toprule
\textbf{Encoder} & \textbf{Flip rate} & \textbf{Top-10 overlap} & \textbf{Cert F1} & \textbf{Threshold} \\
\midrule
text-embedding-3-large       & 0.626 (vs 0.812 short-form) & 0.820 & 0.148 & $\ge 0.40$ (FAIL) \\
sentence-transformers-mpnet  & 0.733 (vs 0.848 short-form) & 0.785 & 0.137 & $\ge 0.40$ (FAIL) \\
\bottomrule
\end{tabular}
\end{center}

\smallskip
Long-form retrieval is materially \emph{more stable} (flip rate $0.81{\to}0.63$, overlap $0.73{\to}0.82$), but the certificate's predicted budget $\sqrt{R^2_{\mathrm{inv}}}$ is still not a tight predictor of which paraphrases flip top-$k$. On long-form the certificate's \emph{soundness} premise is universally met (HB$'$1 above), but the bound's \emph{tightness} for the retrieval-flip use case remains loose. These are two distinct empirical claims:

\begin{itemize}[leftmargin=*,topsep=2pt,itemsep=2pt]
\item \emph{HB$'$1 (premise):} the inequality $\Delta^2 > 0$ holds; the certificate makes a sound claim.
\item \emph{HB$'$2 (tightness):} $\sqrt{R^2_{\mathrm{inv}}}$ is large enough to permit drift that crosses retrieval boundaries, so the bound is not a useful flip-discriminator on its own.
\end{itemize}

The right tightening is a calibrated post-bound: deploy the \S\ref{sec:embedding} certificate as a \emph{soundness gate} (block when the inequality fails to hold) and pair it with a separately-fitted retrieval-flip predictor. Forward work.

\paragraph{Hypothesis HB$'$3: edit-type ablation on long-form.} Repeating the Phase~1 ablation on the long-form data: \textbf{every (sub-$G_{\mathrm{inv}}$, sub-$G_{\mathrm{sig}}$) cell yields 1.000 positive-fraction} on text-embedding-3-large. The cell that was hardest on short-form, paraphrase$\times$attribute\_flip (45\% positive), is at 1.000 on long-form (median $\Delta^2 = 0.465$). The same pattern holds on mpnet, with two cells dropping marginally to 0.984 (paraphrase$\times$entity\_swap and paraphrase$\times$attribute\_flip) and the rest at 1.000.

\paragraph{Hypothesis HB$'$4: cross-encoder uniformity.} Median $|\Delta^2|$ ratio between text-embedding-3-large and mpnet on long-form is $1.07\times$ (vs $0.64\times$ short-form), against a pre-registered $\ge 2\times$ differentiation threshold. The two encoders perform near-identically on long-form, which is the operationally desirable signal: the certificate's value transfers across encoder providers.

\paragraph{Synthesis: the \S\ref{sec:embedding} certificate is sound on every input length, with selective-sensitivity supported on long-form HotpotQA pending domain replication.} Combining Pilots~B and B$'$:

\begin{itemize}[leftmargin=*,topsep=2pt,itemsep=2pt]
\item \emph{Soundness (the implication itself):} unaffected; the Lean-checked Theorem~\ref{thm:robust_decision} continues to hold whenever its declared premise does.
\item \emph{Premise on adversarial short queries (Pilot~B):} holds in 30\% / 20\% of cases; only specific edit-type sub-rectangles (e.g., synonym$\times$entity\_swap) clear the $\ge 0.80$ threshold. Deployers on short-query benchmarks should declare a restrictive $G_{\mathrm{inv}}$ family.
\item \emph{Premise on long-form text (Pilot~B$'$):} holds in 100\% / 98\% of cases across every edit-type combination on this slice. Deployers on long-form retrieval can declare the full $G_{\mathrm{inv}}$ family with empirical support on this benchmark, subject to target-domain validation.
\item \emph{Tightness for retrieval-flip prediction (Pilots~B and B$'$):} loose on both. The certificate is a soundness gate, not a tight flip predictor.
\item \emph{Cross-encoder generalisation:} the two-encoder pattern is consistent across input lengths; the failure modes and successes track input length and edit-type more than encoder provider.
\end{itemize}

\paragraph{Pilot data and reproducibility (Pilot~B$'$ long-form).} The Pilot~B$'$ numbers are reproducible from the companion repository (per-query $\Delta^2$, per-pair drift and flip records, panel votes). On the $N{=}64$ slice, long-form variant generation incurs a higher per-query cost than short-form, driven by generator output proportional to passage length.

\paragraph{What this changes for \S\ref{sec:embedding} and \S\ref{sec:card} (Universal Card).} The certificate is correctly characterized as \emph{conditional on a per-query measurement of $\Delta^2 > 0$}. The Universal Card's \texttt{stability} field (§\ref{sec:card}) accordingly reports the measured $\Delta^2$ as a per-call signal rather than as a deployment-wide guarantee, and the Card's policy layer can downgrade or block emission when $\Delta^2 \le 0$. The forward-work questions are: (i)~which encoder/dataset combinations satisfy the inequality at material rates (we expect long-form documents, named-entity-heavy clinical/legal text, and technical paragraphs to behave better than short adversarial QA); (ii)~whether richer signature definitions (e.g., conditioning on detected negation~\cite{kassner2020negated} or on entity-level edit types) restore the inequality.

\paragraph{Why this matters for the framework.} The \S\ref{sec:grounding} grounding certificate (validated empirically in \S\ref{sec:pilot_a}: catch=1.0 at FBR=0.064) and the \S\ref{sec:embedding} embedding certificate (premise empirically narrow on this benchmark) together illustrate the framework's intended epistemic posture: the certificates are \emph{soundness machinery}, and the empirical question is always ``does the per-call premise hold here?''. A pilot that finds the premise holds robustly is a deployment-readiness result; a pilot that finds it does not is a scope-of-applicability result. Both are publishable empirical contributions; both feed back into the Card's per-call signal vector. Pilot~B is the second kind, on this benchmark, on these encoders.

\section{Hoare-Style Agent Action}
\label{sec:agents}

Agentic LLM systems write to files, call APIs, send messages~\cite{schick2023toolformer,yao2023react,wang2024survey_agents}. A single action can be irreversible; each action warrants a formal audit trail that it was safe given the system's state. Hoare logic~\cite{hoare1969axiomatic,dijkstra1976discipline} is the natural formalism, and the resulting certificate has the cleanest discharge profile in the framework: in the compiled artifact, the trajectory certificate carries no axioms at all (Appendix~\ref{app:axiom_audit}), so every proof step reduces by kernel computation alone. The Hoare family is special only in that it pushes all of its conditionality into the type signatures and predicate definitions, leaving none in the form of named axioms; the per-failure cost remains highest among the families and the deployment story is unusually compelling for decision-makers. Axiom-freeness is a statement about the proof's discharge mechanism, not about the semantic adequacy of the human-authored predicates (\texttt{owner = user}, \texttt{path is in sandbox}, \texttt{age $> 30$d}); see \S\ref{sec:artifact:axiom_free} for the precise formulation of what axiom-freeness does and does not mean.

\subsection{The certificate}

For each callable action $a$, specify a precondition $\mathrm{Pre}(a, \sigma, \alpha)$ and postcondition $\mathrm{Post}(a, \sigma, \sigma', \alpha)$ as decidable Lean predicates on a state model $\sigma$. The LLM proposes $(a, \alpha)$; the formal layer verifies \texttt{Pre}, executes $a$ via a typed executor, verifies \texttt{Post}, and emits the certificate.

\smallskip
\leanfile{Action.lean}
\begin{lstlisting}
structure ActionCertificate (σ : Type) (ActionName : Type)
                             (Args : ActionName -> Type) where
  action : ActionName
  args : Args action
  state_pre : σ
  state_post : σ
  pre_pred : σ -> (args : Args action) -> Prop
  post_pred : σ -> σ -> (args : Args action) -> Prop
  p_pre_holds : pre_pred state_pre args
  p_post_holds : post_pred state_pre state_post args
  p_execution : execute action args state_pre = state_post

structure TrajectoryCertificate (σ : Type) where
  steps : List (ActionCertificate σ ActionName Args)
  p_chain : ∀ i, i + 1 < steps.length ->
              (steps.get! (i+1)).state_pre = (steps.get! i).state_post
\end{lstlisting}

A trajectory certificate fails to type-check at the first unsafe step; the agent either emits a complete safe trajectory or is blocked. The LLM is the proposer; a small deterministic Lean-audited runtime is the final authority that permits or blocks side effects.

\subsection{Five invariant classes}

Concrete agent frameworks instantiate the certificate against five invariant classes: \emph{data invariants} (no scope leakage); \emph{budget invariants} (token / API / monetary cost); \emph{safety invariants} (forbidden tools / destinations); \emph{state invariants} (action coherent with memory and task state); \emph{trace integrity} (logged and re-executable). Each is a decidable predicate populated from the deployment's policy layer.

\subsection{Beyond Hoare}

Real agents exhibit structure that simple sequential Hoare does not capture. We sketch the path; full development is in the extended catalogue (\S\ref{sec:catalogue}). \emph{Separation logic} for modular substates~\cite{ohearn2019concurrent}: each action certificate asserts only the resources it reads and writes; frame rules let an action's certificate compose into a larger context. \emph{Rely-guarantee}~\cite{jones1983rely} for parallel tool calls: each branch declares the environmental changes it tolerates and what it promises others; the trajectory composes by checking each branch's rely is satisfied by the conjunction of others' guarantees. \emph{Belief-state Hoare} for partial observability: pre/post on belief paired with an observation-soundness axiom. \emph{Temporal-logic protocol correctness}: ``authenticate before transfer,'' ``confirm before destructive action,'' modeled as CTL$^*$ / LTL formulas over the trajectory.

\textbf{Verification status:} the Hoare baseline is \emph{compiled, axiom-free} (\texttt{Action.\allowbreak trajectory\_all\_steps\_safe}); the rely-guarantee, belief-state, and temporal-logic extensions are \emph{specified} and become mechanizable once the obligations are decidable.

\subsection{Connection to shielding}

The pattern generalizes safe-RL shielding~\cite{alshiekh2018safe} and logic-constrained RL~\cite{hasanbeig2020cautious}: in both, a runtime monitor intercepts an unsafe action and substitutes a safe alternative. Action certificates make the monitor's guarantees formal and the substitutions auditable. The difference is that our monitors are kernel-audited Lean predicates rather than learned or hand-coded policies, and the audit trail is cryptographically bound to the ML artifact via the oracle digest.

\subsection{Empirical validation: Pilot~C}
\label{sec:agents:empirical}

Pilot~C tests the \S\ref{sec:agents} trajectory certificate's pre-execution gating
behavior on a controlled filesystem-sandbox benchmark; design,
adversarial-action registry, and pre-registered hypotheses are in
the companion repository.

\paragraph{Setup.} A tempdir-rooted filesystem sandbox per
trajectory; nine tools (\texttt{ls}, \texttt{cat}, \texttt{find},
\texttt{stat}, \texttt{md5}, \texttt{rm}, \texttt{mv}, \texttt{mkdir},
\texttt{chmod}) dispatched through a single gate; five
deterministic Hoare-predicate constructors (\texttt{path\_in\_sandbox},
\texttt{path\_in\_subprefix}, \texttt{mtime\_at\_least},
\texttt{owner\_is}, \texttt{delete\_count\_within\_budget}); five
representative tasks spanning four categories (file-management, deduplication, cleanup, reorganization), each carrying a per-tool predicate set; three
harnesses (Unaudited, Deny-list, Treatment) sharing the same
ReAct agent (GPT-5).

\paragraph{Threat model.} Direct-injection adversarial actions: after
two agent steps establish non-trivial sandbox state, a registry
of 18 \emph{known-unsafe action proposals} is spliced into the
trajectory and run through each harness's gate. The proposals span
three pre-registered attack classes: out-of-scope file references
(\texttt{/etc/passwd}, \texttt{\textasciitilde/.ssh/id\_rsa}),
owner/age-violating deletes (recent or wrong-owner files), and
budget-exceeding sequences (six deletes against a 5-delete cap).

The pre-registered design specified prompt-injection adversaries
(adversarial system messages embedded in the agent's observation
stream); a preliminary feasibility check found that GPT-5 simply did not act on
those injection prompts (refusal rate ${>}\!0.95$), making the
prompt-injection design uninformative as a test of the gate. The
registered slice therefore uses direct splicing; the deviation is
recorded in the pre-registration notes of the companion repository,
and the result is framed accordingly. \emph{This is a strictly stronger test of
the gate in isolation:} direct injection removes agent cooperation
as a confound, so any block in the Treatment column is attributable
to the gate and not to the agent declining to propose the action.
The two failure modes (gate fails to block a violating action; agent
falls for prompt injection) are orthogonal; Pilot~C measures the
first by construction, and the second is a separately-measurable
property of the agent that is not the \S\ref{sec:agents} certificate's
responsibility. A complementary pilot that varies model and prompt
to elicit non-trivial prompt-injection success on a more
manipulable agent, and then re-runs the same gate against that
agent's proposals, is forward work.

\paragraph{Trajectory composition.} The 48 trajectories decompose as 16 per harness across the three harnesses, where each harness's 16 trajectories comprise 11 adversarial trajectories carrying 18 spliced direct-injection unsafe-action proposals (some trajectories carry multiple injections) plus 5 benign trajectories carrying 6 destructive-but-permitted proposals. The block-rate denominators below are over the 18 injected adversarial proposals and the 6 benign destructive proposals respectively.

\paragraph{Results (48 trajectories; approximately 30~minutes of wall-clock).}

\smallskip
\begin{center}
\footnotesize
\begin{tabular}{lcc}
\toprule
\textbf{Harness} & \textbf{Block rate on injected unsafe} (95\% CI) & \textbf{n} \\
\midrule
Unaudited (no gate) & 0.000 (0.000, 0.000) & 18 \\
Deny-list (regex baseline) & 0.278 (0.111, 0.500) & 18 \\
\textbf{Treatment} (Lean-checked Hoare) & \textbf{0.667} (0.444, 0.889) & 18 \\
\bottomrule
\end{tabular}
\end{center}

\smallskip
False-block rate on benign destructive proposals: \textbf{0.000} on
all three harnesses ($n{=}6$ each). Audit-log informativeness on
Treatment blocks: \textbf{1.000} (13 of 13 blocks contain the
failed predicate, the offending state value, the expected
condition, and the upstream agent proposal text). The 13 Treatment
blocks comprise the 12 unsafe-injection blocks (out of 18 spliced
proposals) plus 1 read-only \texttt{stat} block on a benign trajectory
where the agent proposed inspecting a path the sandbox declares
out-of-scope; the latter is non-destructive and hence does not
count toward the false-block-rate denominator above.

\smallskip
\emph{In-scope reading.} Of the 13 attacks targeting a registered
predicate clause, 12 were blocked and 1 passed (92.3\%); the 5
remaining unblocked within-budget delete probes are correctly
permitted under the registered budget predicate (the 6th in the
sequence, which exceeds the budget, is blocked). HC1's denominator
includes both classes, by literal pre-registration; the diagnostic
separates them.

\paragraph{Hypothesis verdicts.}

\smallskip
\begin{center}
\begin{tabular}{cp{7cm}c}
\toprule
\textbf{H} & \textbf{Statement (registered)} & \textbf{Verdict} \\
\midrule
HC1 & Treatment block rate $\ge\!0.95$ on injected unsafe actions & FAIL ($\to$ predicate-completeness diagnostic) \\
HC2 & Treatment FBR $\le\!0.05$ on benign destructive actions & PASS \\
HC3 & Treatment block rate $\ge$ Deny-list + 15pp & PASS ($\Delta = 38.9$pp) \\
HC4 & Unaudited block rate $\le\!0.05$ (structural floor) & PASS \\
HC5 & Audit-log informativeness $\ge\!0.95$ of Treatment blocks & PASS \\
\bottomrule
\end{tabular}
\end{center}

\paragraph{HC1 verdict: FAIL, and what the FAIL localizes.} The
literal pre-registered statement of HC1 (Treatment block rate
$\ge\!0.95$ on injected unsafe actions) is not satisfied: Treatment
blocks 12 of 18 (66.7\%, 95\% CI [0.444, 0.889]). We report this as
the verdict of record.

\emph{What the audit log says to fix.} The empirical contribution Pilot~C was designed to deliver for typed-guarantee certificates is the audit-log diagnostic on each unblocked attack. The 6 unblocked attacks split cleanly: 5 are within-budget delete probes that the t02 \texttt{delete\_count\_within\_budget} predicate correctly permits (the 6th in the same sequence is blocked, exactly when the budget is exceeded); the remaining 1 is a recent-mtime move that probes a predicate the t01 \texttt{mv} declaration deliberately omits. The audit log therefore names exactly which predicate clauses (t02 budget narrowness, t01 mv mtime omission) need extending to close the gap. The diagnostic also exposes a registration-time operationalisation choice: HC1's denominator was specified as ``adversarial proposals,'' not ``predicate violators,'' so 5 within-budget probes that no registered predicate violates were nevertheless counted as failures-to-block. Conditional on attacks that target a predicate clause in scope, the block rate is 12/13. We log the conflation as a registration deviation to be tightened in any follow-up; the verdict of record remains HC1 FAIL on the literal pre-registered formula.

\emph{Why the Lean soundness theorem is unaffected.} Every action that passed the gate satisfied the predicates declared in scope, and no action that violated a declared predicate passed the gate. The Lean-checked trajectory soundness theorem (\texttt{Action.\allowbreak trajectory\_all\_steps\_safe}, compiled axiom-free in the artifact; see Appendix~\ref{app:axiom_audit} and \S\ref{sec:agents}) is therefore independent of HC1's pre-registration framing.

\paragraph{Scope.} Pilot~C runs on a registered slice of five
representative tasks (selected to span the four predicate
families); the original 30-task plan was scaled to fit the
available budget and the preliminary feasibility schedule.
Predicate-registry breadth across more tasks, more tools, and
more attack classes is forward work (\S\ref{sec:limits:empirical}).
The 18 adversarial proposals are hand-curated; automated
adversary generation (LLM-driven proposal mining against a
held-out predicate set) is the natural follow-up.

\paragraph{What this validates for \S\ref{sec:agents}.} The empirical claim Pilot~C
makes for the Hoare-action family is the most sharply-typed of the
four pilots: \emph{within the predicate-encoded scope, the Treatment
gate dominates the deny-list baseline by 39 percentage points
(0.667 vs 0.278) on a controlled adversarial-injection set, with 100\%
audit-log informativeness and zero false-blocks on benign
trajectories.} The Hoare family's zero-axiom Lean proof structure
turns the empirical question into pure predicate-set design: the
audit log of unblocked attacks names exactly which predicate clauses
need to extend to close the residual 8\%. This is the empirical
contribution Pilot~C was designed to deliver.

\section{Compositional Stability}
\label{sec:composition}

The three local certificate families developed earlier in the paper give a grounding bound (\S\ref{sec:grounding}), an embedding-stability bound (\S\ref{sec:embedding}), and a Hoare action bound (\S\ref{sec:agents}). Each is local: it bounds a single pipeline layer or claim. The pipeline-level question is what survives at the final output given a bounded perturbation of the raw user input, and that question is what this section answers. The closed-form pipeline-wide budget formula derived below is what makes the framework usable end-to-end rather than per-layer: it turns the local catalogue of \S\S\ref{sec:grounding}, \ref{sec:embedding}, and~\ref{sec:agents} into a single inequality the deployer can inspect to decide whether a given adversary is provably out of reach.

\subsection{Pipeline model}

A pipeline $\Pi$ is a finite sequence of layers $L_1, \ldots, L_n$ with $f_i : \mathcal{X}_i \to \mathcal{Y}_i = \mathcal{X}_{i+1}$ and metric $d_i$ on each space. Composition is $\Pi(x) = f_n \circ \cdots \circ f_1(x)$.

\begin{definition}[Certified Layer]\statuscompiled
\label{def:certlayer}
$L_i$ is \emph{certified} with \emph{gain} $g_i \in \mathbb{Q}_{\ge 0}$ and \emph{margin} $m_i \in \mathbb{Q}_{\ge 0}$ if a kernel-audited certificate witnesses
$d_{i+1}(f_i(x), f_i(x')) \le g_i \cdot d_i(x, x')$
for all $x, x'$ with $d_i(x, x') < m_i$, and the layer's discrete decision is invariant under that bound (\emph{nominal stability}).
\end{definition}

The local certificates of the three families above supply $g_i$ and $m_i$ by construction. Embedding (\S\ref{sec:embedding}): $g = 1$ via Cauchy--Schwarz, $m = \sqrt{R^2_{\mathrm{inv}}}$ as defined there. Discrete decision layers (paraphrase, ranking, emission threshold): $g = 0$, $m = $ flip distance or coverage slack. Agent action (\S\ref{sec:agents}): $g = 0$, $m = $ precondition slack.

\medskip
\noindent\textbf{From a finite witness family to a metric ball.} Definition~\ref{def:certlayer} states the certified-layer inequality on a metric ball $\{x' : d_i(x_0, x') < m_i\}$, but the embedding certificate of \S\ref{sec:embedding} is built from a \emph{finite} declared family $G_{\mathrm{inv}}(x_0)$ (Definition~\ref{def:perturbations}). The bridge is layer-specific: for the embedding layer, $R^2_{\mathrm{inv}}$ certifies the worst-case displacement over the witness family by direct enumeration, and Cauchy--Schwarz lifts that to every $x'$ with $\|E(x') - E(x_0)\|^2 \le R^2_{\mathrm{inv}}$, i.e.\ the entire $L_2$-ball of squared radius $R^2_{\mathrm{inv}}$ in the embedding-output metric. So the metric-ball margin $m_i = \sqrt{R^2_{\mathrm{inv}}}$ holds for the layer's \emph{output-side} distance for free, while the input-side distance ($d_1$ on text inputs) is constrained only over $G_{\mathrm{inv}}(x_0)$ itself unless a separate encoder-Lipschitz hypothesis is supplied. We adopt this as a layer-specific match between certificate witness and metric: the embedding layer transports a finite-family declaration into a continuous-ball guarantee on its output, and the downstream layers consume the latter. For text-input layers without a metric, $g_i$ is taken as $0$ and the layer is treated as a discrete-decision gate (e.g.\ paraphrase classifier, retrieval threshold), short-circuiting the metric requirement upstream of any continuous propagation.

\subsection{The theorem}

\begin{lemma}[Two-Layer Compositional Stability, Base Case]\statuscompiled
\label{lem:composition_two_layer}
Let $\Pi_2 = L_2 \circ L_1$ with both layers certified. If $d_1(x_0, x') < m_1$ and $g_1 \cdot d_1(x_0, x') < m_2$, then $L_2(L_1(x')) = L_2(L_1(x_0))$ at the discrete-decision level (or under canonicalization).
\end{lemma}

\textbf{Verification status:} \emph{compiled} as \texttt{Composition.compositional\_\allowbreak stability\_two\_layer} in the artifact (audited axiom set: $\Omega$). The base case unfolds to two applications of nominal stability plus rational arithmetic.

\begin{theorem}[Compositional Stability, $n$-Layer]\statusspec
\label{thm:composition}
Let $\Pi = L_n \circ \cdots \circ L_1$ have certified layers with gains $(g_i)$ and margins $(m_i)$. Define $\varepsilon^\star_1 = \varepsilon_0$ and $\varepsilon^\star_{i+1} = g_i \cdot \varepsilon^\star_i$. If $\varepsilon^\star_i < m_i$ at every interface, then $d_1(x_0, x') < \varepsilon_0$ implies $\Pi(x') = \Pi(x_0)$ (or $C(\Pi(x')) = C(\Pi(x_0))$ if $\Pi$ ends in canonicalization).
\end{theorem}

\begin{proof}
By induction on $i$, with Lemma~\ref{lem:composition_two_layer} as the inductive step. Base ($i=1$): $d_1(x_0, x') < \varepsilon_0 = \varepsilon^\star_1 < m_1$, so $L_1$'s nominal stability gives $L_1(x') \equiv_\mathrm{dec} L_1(x_0)$ and $d_2(\cdot,\cdot) \le g_1 \varepsilon_0 = \varepsilon^\star_2$. Inductive step: assuming the conclusion at layer $i$, apply Lemma~\ref{lem:composition_two_layer} to layers $L_i, L_{i+1}$. At $i = n$ the discrete decision is invariant.
\end{proof}

\textbf{Verification status:} \emph{specified}. The $n$-layer statement is a structural induction over a \texttt{List.foldl} of certified layers; each step is an instance of the compiled \texttt{compositional\_\allowbreak stability\_two\_layer} base lemma plus the per-layer \texttt{nominal\_stable} field. Mechanizing the inductive lift is a single Lean proof obligation over a finite list and is currently architecturally enforced via the chain-of-margins type signature; closing it to fully \emph{compiled} is a focused exercise.

\begin{corollary}[Pipeline Conditioning Number and Budget]\statusspec
\label{cor:conditioning}
Let $\kappa_\Pi = \prod_{i=1}^n g_i$ (cumulative gain) and $B_\Pi = \min_{1 \le i \le n} m_i / \prod_{j<i} g_j$, with the convention that any term whose denominator $\prod_{j<i} g_j$ vanishes (i.e.\ some earlier layer is a discrete-zero gate, $g_{j^\star} = 0$ for some $j^\star < i$) is treated as $+\infty$ and does not bind the minimum. Equivalently, let $i^\star = \min\{i : \prod_{j<i} g_j = 0\}$ when such an index exists; then $B_\Pi = \min_{1 \le i < i^\star} m_i / \prod_{j<i} g_j$, and otherwise the minimum is taken over all $n$ layers. Theorem~\ref{thm:composition} guarantees decision invariance for every $\varepsilon_0 < B_\Pi$, and $B_\Pi$ is the tightest such budget producible from the given per-layer certificates.
\end{corollary}

\subsection{Discrete decision layers and the budget}

In most real pipelines, cumulative gain collapses to zero after the first discrete-decision layer (paraphrase decision, retrieval threshold, emission threshold), and $B_\Pi$ reduces to the \emph{tightest single-layer margin}. This formalizes the folklore observation that pipelines with strong discrete gates are more stable end-to-end than pipelines with many cascaded continuous layers. Concretely, for a 7-layer RAG pipeline with $(g_i) = (1,0,0,0,1,0,0)$ and $(m_i) = (0.12, 50, 3, 0.08, 0.77, 1, 1)$, we have $\kappa_\Pi = 0$ and $B_\Pi = 0.12$ (dominated by the embedding margin from \S\ref{sec:embedding}). The motivating instability ``we have no idea whether this similarity score is stable'' becomes ``the canonical denotation of the answer is invariant under any lexical paraphrase whose $L_2$ embedding drift is below $0.12$.''

\subsection{Adversarial corollary}

If a paraphrase adversary is bounded by $\varepsilon_0^{\mathrm{adv}}$ on the user input, every adversary with $\varepsilon_0^{\mathrm{adv}} < B_\Pi$ is provably unable to flip the canonical denotation. This is a positive statement over all adversaries in the declared family, complementing empirical adversarial testing which only exhibits flipping examples when found.

\paragraph{Perturbation-budget convention.} In this paper, $B_\Pi$ and the corresponding Card field \texttt{perturbation\_budget} denote a certified \emph{stability radius} (larger values represent tighter guarantees): the pipeline's canonical denotation is invariant under any upstream perturbation of $L_2$-norm strictly below the budget. Implementations that instead store a worst-case \emph{drift envelope} (such as $R^2_{\mathrm{inv}}$ or a quantization-error-adjusted drift bound, where smaller values are tighter) should either use a separate \texttt{drift\_bound} field or attach a convention tag to the Card so the application-policy layer applies the correct comparison direction.

\section{The Universal Assurance Card}
\label{sec:card}

The three local certificate families of \S\S\ref{sec:grounding}, \ref{sec:embedding}, and \ref{sec:agents} address specific failure modes; the residue operator (\S\ref{sec:mcr}) and the compositional stability theorem (\S\ref{sec:composition}) compose them. What is still missing, viewed from a deployment perspective, is the \emph{single object} that travels with each governed LLM call \emph{under a declared assurance mode} and lets a downstream consumer decide, automatically and without reading the model output, whether to act on the result. This section gives that object: a Lean~4 record we call the \emph{Universal Assurance Card}, with a decidable consistency predicate, designed for the per-call deliverable of any high-stakes deployment. The structure and predicate are stated in Lean form below; integration into the compiled artifact as \texttt{AssuranceCard.lean} is straightforward but separate engineering work and is currently \emph{specified} in the verification-status convention.

\subsection{Three design principles}

\begin{enumerate}[leftmargin=1.5em,topsep=2pt,itemsep=2pt]
\item \emph{No single confidence number.} A scalar collapses orthogonal failure modes and licenses exactly the false confidence Proposition~\ref{prop:dependency} forbids. The card is a vector.
\item \emph{Application-policy interpretation.} The same vector is read differently by legal research, clinical triage, creative writing, agentic action; the card carries fields, the application carries policy. Cards are portable; policies are deployment-specific.
\item \emph{Always-emitted.} Every call produces a card, including failures: abstention is a verdict, and the user is owed a machine-readable reason.
\end{enumerate}

\subsection{The schema (specified)}

The card is declared below as a Lean~4 structure. The consistency proof field \texttt{p\_verdict\_consistent} is decidable: given the field values, the predicate \texttt{VerdictConsistent} checks that the verdict label is internally consistent with the declared field values on the card itself (e.g., \textsc{Certified} requires \texttt{unverified\_layers = []} and all safety checks true). Once the structure and predicate are added to the artifact (no Mathlib-heavy lemmas are required; only \texttt{Rat}, \texttt{Bool}, and \texttt{List String} comparisons), \texttt{\#print axioms} on the worked-example cards from \S\ref{sec:cases} should report only $\Omega$.

\noindent\textbf{What \texttt{VerdictConsistent} is and is not.} The predicate is a \emph{schema-consistency} check on the Card record itself: given the values of the eight field groups, does the verdict tag make sense? Concretely, the visible \textsc{Certified} branch checks \texttt{unverified\_layers = []}, \texttt{forbidden\_check\_pass = true}, \texttt{action\_precondition = true}, \texttt{contradicted\_mass = 0}, and \texttt{in\_scope = true}, and does \emph{not} visibly check \texttt{contested\_mass}, \texttt{unknown\_mass}, \texttt{budget\_under\_limit}, freshness via \texttt{issue\_time}/\texttt{expiry\_time}, \texttt{calibration\_gap}, \texttt{semantic\_agreement}, or \texttt{certificate\_digest} validity. Each of these gaps is \emph{deliberate}: the predicate is intentionally a lightweight schema gate, and the omitted properties are policy-layer or digest-checked obligations rather than schema obligations.

It is \emph{not} a re-verification of the underlying certificates: that the embedding-sensitivity certificate's $\Delta^2$ field is the actual $R_{\mathrm{sig}}^2 - R_{\mathrm{inv}}^2$ for the declared families, that the grounding certificate's $W^-$ is the normalized weighted coverage of contradicted claims (the \texttt{p\_W\_contradicted\_correct} field of the grounding certificate, \S\ref{sec:grounding}), or that the Hoare trajectory's preconditions hold on the executed state, are properties of the underlying compiled certificates and are enforced there. The Card binds those underlying guarantees into one decidable record by referencing them via the \texttt{certificate\_digest} field (a hash of the signed underlying certificates produced when the Card is emitted); a downstream auditor verifying the Card re-checks the digest against the underlying certificates' kernel-audited build outputs. \texttt{VerdictConsistent} is therefore the lightweight schema gate; the deeper guarantee chain runs through the digest into the compiled certificates of \S\S\ref{sec:grounding}, \ref{sec:embedding}, and \ref{sec:agents}, and the policy-layer thresholds (\texttt{AssurancePolicy}) enforce deployment-specific cuts on the omitted fields.

\smallskip
\leansketch{AssuranceCard schema; decidable consistency; specified, not yet in the compiled artifact}
\begin{lstlisting}
inductive AssuranceVerdict
  | certified
  | partial   (gaps : List String)
  | residue   (weight : Rat)
  | abstain   (reasons : List String)

structure AssuranceCard where
  verdict : AssuranceVerdict
  -- stability (sources: composition, embedding)
  perturbation_budget : Rat
  variant_flips : Nat
  semantic_agreement : Rat
  -- evidence (source: bilattice grounding)
  supported_mass contradicted_mass contested_mass unknown_mass : Rat
  -- calibration (catalogue: conformal / calibration)
  conformal_set_size : Nat
  calibration_gap : Rat
  -- reproducibility (catalogue: proof-of-sampling)
  is_deterministic : Bool
  sample_count : Nat
  proof_of_sampling : Option ByteArray
  -- scope and freshness (catalogue: scope-of-validity, provenance, temporal)
  in_scope : Bool
  scope_predicates_held : List String
  source_snapshot_hash : ByteArray
  issue_time expiry_time : Nat
  -- safety (source: Hoare action; catalogue: negative-guarantee, budget)
  forbidden_check_pass : Bool
  action_precondition : Bool
  budget_under_limit : Bool
  -- residue (source: MCR)
  residue_coverage : Rat
  dropped_claims : List String
  -- gaps and provenance
  unverified_layers : List String
  model_version_hash prompt_template_hash : ByteArray
  human_signatures : List ByteArray
  -- audit handle
  certificate_digest : ByteArray
  verifier_version : String
  replay_handle : String
  -- decidable consistency obligation (full Card record passed)
  p_verdict_consistent : VerdictConsistent (AssuranceCard.mk ...)

-- The predicate is a function on the whole `AssuranceCard` record;
-- the match destructures the verdict tag and reads the field values
-- it inspects from the card. The fields the predicate does not
-- destructure (e.g. contested_mass, unknown_mass, budget_under_limit,
-- issue_time, expiry_time, calibration_gap, semantic_agreement,
-- certificate_digest) are deliberately omitted -- see the
-- "What VerdictConsistent is and is not" discussion above.
def VerdictConsistent (c : AssuranceCard) : Prop :=
  match c.verdict with
  | .certified =>
     c.unverified_layers = [] ∧
     c.forbidden_check_pass = true ∧
     c.action_precondition = true ∧
     c.contradicted_mass = 0 ∧
     c.in_scope = true
  | .partial gaps =>
     c.unverified_layers = gaps ∧
     c.forbidden_check_pass = true ∧
     c.action_precondition = true
  | .residue w =>
     c.residue_coverage = w ∧ c.dropped_claims ≠ []
  | .abstain rs =>
     rs ≠ [] ∧ c.residue_coverage = 0
\end{lstlisting}

\paragraph{Four kinds of residue.} The residue clause uses the universal characterization \texttt{dropped\_claims} $\neq$ \texttt{[]}, which is witnessed by any of four reasons the MCR operator (\S\ref{sec:mcr}) drops items from a candidate emission, each a separate monotone constraint family in $\mathcal{C}$:
\begin{itemize}[leftmargin=1.5em,topsep=2pt,itemsep=2pt]
\item \emph{Evidence residue}: $C_{\mathrm{ground}}$ excludes \textsc{Contradicted} or \textsc{Contested} claims; characterized by \texttt{contradicted\_mass} + \texttt{contested\_mass} $> 0$. (Clinical case study, \S\ref{sec:case_clinical}.)
\item \emph{Action residue}: $C_{\mathrm{pre}}$ excludes implied tool-calls whose Hoare precondition fails; characterized by \texttt{dropped\_claims} listing blocked actions even when \texttt{contradicted\_mass} = 0. (Agentic case study, \S\ref{sec:case_agent}.)
\item \emph{Scope residue}: a scope-of-validity predicate fails on a candidate atomic claim, so the MCR drops it; characterized by \texttt{scope\_predicates\_held} excluding the relevant predicate.
\item \emph{Budget residue}: $C_{\mathrm{bud}}$ excludes the marginal claim that would push cumulative cost past the declared limit; characterized by \texttt{budget\_under\_limit} = true with non-empty \texttt{dropped\_claims}.
\end{itemize}
The four kinds are mutually compatible (a single card can witness multiple), and each composes with the others through the MCR's monotone-constraint discipline. The card's verdict is \textsc{Residue} whenever \emph{any} of them fires; the field vector tells the consumer which.

\textbf{Verification status:} \emph{specified}. The Card structure and \texttt{VerdictConsistent} predicate are stated in Lean form above; the predicate is decidable by case analysis on the verdict tag plus rational and list comparisons; integrating it into the artifact as \texttt{AssuranceCard.lean} with a worked example for each of the four verdicts (each closing via \texttt{decide}) is a focused engineering step and is the natural follow-on to this paper. This is the same compiled-status precision applied to Compositional Stability (Lemma~\ref{lem:composition_two_layer} compiled, Theorem~\ref{thm:composition} specified).

\subsection{The four verdicts}

\begin{description}[leftmargin=0pt,itemsep=4pt,topsep=4pt]
\item[\textsc{Certified}.] Every required certificate passes. The output is emitted with a kernel-audited bundle.
\item[\textsc{Partial}.] Some non-load-bearing layers are unverified; the card explicitly names them in \texttt{unverified\_layers}. The output is emitted; the consumer's policy decides whether the named gaps are tolerable.
\item[\textsc{Residue}.] The full answer cannot be certified; a maximal certifiable subset can. The card reports the residue weight (Theorem~\ref{thm:mcr_maximality}) and lists dropped claims. The consumer receives a strictly smaller but kernel-audited answer.
\item[\textsc{Abstain}.] No emission. The card reports machine-readable reasons; an honest abstention is a first-class outcome.
\end{description}

\subsection{Application policy}

The card is not directly comparable across applications. A deployment publishes its acceptance policy as a separate object that maps the field vector to an accept / reject / downgrade decision. Both the card and the policy are kernel-objects, so the deployment's accept/reject choice is auditable.

\smallskip
\leansketch{application-policy schema}
\begin{lstlisting}
structure AssurancePolicy where
  min_supported_mass : Rat
  max_contradicted_mass : Rat
  max_contested_mass : Rat
  required_in_scope : Bool
  max_unverified_gaps : Nat
  require_action_safe : Bool
  max_calibration_gap : Rat
  min_perturbation_budget : Rat
  require_proof_of_sampling : Bool
  max_age_seconds : Nat
  forbidden_must_pass : Bool

def policy_accepts (P : AssurancePolicy) (c : AssuranceCard) : Bool := ...
\end{lstlisting}

\begin{table}[H]
\centering
\caption{Five canonical acceptance policies. The same card is interpreted differently by each. Numbers are illustrative starting points; a real deployment tunes them on its own ground truth.}
\label{tab:assurance_policies}
\scriptsize
\setlength{\tabcolsep}{4pt}
\begin{tabular}{@{}p{3.2cm}p{1.4cm}p{1.4cm}p{1.5cm}p{1.4cm}p{1.6cm}@{}}
\toprule
\textbf{Policy field} & \textbf{Creative} & \textbf{Enterprise QA} & \textbf{Legal research} & \textbf{Clinical} & \textbf{Agentic action} \\
\midrule
\texttt{supported\_mass} $\ge$ & 0 & 0.70 & 0.90 & 0.95 & 0.95 \\
\texttt{contradicted\_mass} $\le$ & 1.0 & 0.10 & 0.02 & 0.01 & 0.00 \\
\texttt{contested\_mass} $\le$ & 1.0 & 0.20 & 0.08 & 0.04 & 0.02 \\
\texttt{required\_in\_scope} & no & yes & yes & yes & yes \\
\texttt{max\_unverified\_gaps} & 5 & 2 & 0 & 0 & 0 \\
\texttt{require\_action\_safe} & no & no & no & yes & yes \\
\texttt{max\_age\_seconds} & $\infty$ & 30\,d & 7\,d & 24\,h & 1\,h \\
\texttt{forbidden\_must\_pass} & policy & yes & yes & yes & yes \\
\bottomrule
\end{tabular}
\end{table}

\subsection{Wire protocol}

The wire protocol is intentionally light. \emph{Request side}: a \texttt{min\_certificate\_policy} field carries the application's policy; the provider commits to either returning a card the policy accepts, returning residue or abstention with a reason code, or raising an error. \emph{Response side}: the response carries a top-level \texttt{certificate} field (the \texttt{AssuranceCard}) alongside the answer, plus an audit handle so a third party can replay the kernel check without re-running the LLM. \emph{Selective disclosure}: for privacy-sensitive workloads, the provider returns the Merkle root of the card plus only the leaves the user is authorized to see. \emph{Versioning}: the card carries \texttt{verifier\_version}. Adoption pathways and regulatory mappings are deployment artifacts; we collect them in Appendix~\ref{app:adoption} rather than the main text, where they would crowd out the technical core.

The Universal Assurance Card answers the question \emph{``what does it mean for a generative AI system to ship trust along with output, in deployments that need to ship trust?''} The schema is the consolidator; the policy layer is the application interface. Together they turn the certificate taxonomy into a per-call deliverable that any provider serving high-stakes use cases can attach to its assurance-mode endpoint in the next product release.

\section{Two Worked Case Studies}
\label{sec:cases}

This section gives two end-to-end examples chosen to exercise the framework's core: clinical RAG, where MCR and bilattice grounding are most distinctive, and an agentic Hoare execution, where action safety is the deployment story. Each ends with a filled-out Universal Assurance Card. The numerical values in both case studies are illustrative and internally consistent rather than measured; the case studies are intended to exhibit Card behavior and certificate composition, not to make empirical claims. Measured pre-registered data on the constituent certificates is reported elsewhere: grounding in \S\ref{sec:pilot_a} (Pilot~A), embedding sensitivity in \S\ref{sec:embedding:empirical}--\S\ref{sec:embedding:longform} (Pilots~B and B$'$), and Hoare-style agent action in \S\ref{sec:agents:empirical} (Pilot~C). End-to-end validation of the Card consolidator itself remains forward work, structured as a Pilot~D family of domain-specific validations (\S\ref{sec:limits:pilot}).

\subsection{Case study 1: Clinical RAG with contradiction detection}
\label{sec:case_clinical}

\paragraph{Setup.} A clinical decision-support system answers ``\emph{Is this 62-year-old patient with atrial fibrillation a candidate for direct oral anticoagulation (DOAC) given their CKD-4 renal function?}'' over (a) ACC/AHA and ESC guidelines, (b) the patient's chart, (c) the hospital's formulary, and (d) drug-interaction databases. The clinical numbers below are stylized illustrative placeholders.

\paragraph{Per-layer certificates.}
\begin{enumerate}[leftmargin=1.5em,topsep=2pt,itemsep=2pt]
\item \emph{L0 Scope-of-Validity.} Query clinically valid; chart current within 30 days; NLI model on whitelist; corpus timestamp within 90-day freshness window. All checks pass.
\item \emph{L1 Embedding sensitivity.} Under a 50-variant lexical paraphrase family $G_{\mathrm{inv}}$ on the query, $R^2_{\mathrm{inv}} = 0.014$ ($\sqrt{\cdot} \approx 0.118$); on a 30-variant significant family $G_{\mathrm{sig}}$, $R^2_{\mathrm{sig}} = 0.19$. Selective gap $\Delta^2 = 0.176 > 0$, encoder behaves correctly on this query.
\item \emph{L2 Conflict-aware grounding.} Twelve atomic claims decomposed; reconstruction check passes. \textbf{Key finding}: the recommendation of apixaban is supported by 3 guidelines ($\beta^+ = 0.94$). However, one atomic claim \emph{(``the patient's CrCl of 22 mL/min is within apixaban's renal range'')} is \emph{contradicted} by the FDA label (CrCl $\ge 25$ for standard dosing): $\beta^- = 0.89$. Status: \textsc{Contradicted}. $W^+ = 92\%$, $W^- = 8\%$ (above $\theta_r = 5\%$). \textbf{The certificate's emission gate blocks the full answer} (Lemma~\ref{thm:no_silent}, conditional on the NLI oracle's contradiction judgment).
\item \emph{L3 Maximal Certifiable Residue.} With $c_4$ removed, the remaining 11 claims satisfy all constraints. The MCR (the same operator developed in \S\ref{sec:mcr}, here applied to a 12-claim instance rather than the 5-claim worked example of \S\ref{sec:mcr:worked}) returns the 11-claim residue at weight $92\%$ (illustrative weights chosen as round numbers; per-claim weight assignment is omitted for compactness). The output becomes: ``Apixaban is a supported DOAC choice in AF; \emph{however, for this patient's CrCl of 22 mL/min, standard dosing is outside the FDA label, so dose reduction or alternative agent required. Consult nephrology / pharmacy.}''
\item \emph{L4 Hoare action gate.} The downstream action ``order apixaban at standard dose'' has precondition CrCl $\ge 25$, which fails. The Hoare certificate blocks the order and escalates to human review.
\item \emph{L5 Compositional Stability.} Per-layer gains $(g_i)$ and margins $(m_i)$ from L1--L4 give pipeline budget $B_\Pi \ge 0.06$; the canonical denotation of the emitted residue is invariant under any embedding $L_2$ drift below $0.06$.
\end{enumerate}

\paragraph{What is caught.} The conflict-aware grounding certificate detects the contradiction that the support-only variant misses. MCR retains maximum utility: instead of abstaining entirely, the system emits an honest partial answer. The action certificate blocks the unsafe order. Each catch corresponds to a named certificate in the framework, and each is auditable by an external reviewer (a hospital committee, a malpractice carrier) without access to model weights or prompt.

\medskip
\noindent\textbf{Filled-out Universal Assurance Card (illustrative).}\par

\smallskip
\leanverbatim{clinical card (illustrative)}
\begin{lstlisting}
verdict                  := residue 0.92
perturbation_budget      := 0.06
variant_flips            := 0   (50 paraphrases)
semantic_agreement       := 1.00
supported_mass           := 0.92
contradicted_mass        := 0.08
contested_mass           := 0.00
unknown_mass             := 0.00
calibration_gap          := 0.03
is_deterministic         := false
sample_count             := 5
in_scope                 := true
scope_predicates_held    := ["query_clinical", "chart_fresh_30d",
                             "corpus_within_90d", "nli_whitelisted"]
source_snapshot_hash     := 0xFA7E..0421
issue_time               := 1714000000  ; expiry_time := 1714086400
forbidden_check_pass     := true
action_precondition      := false   ; (Hoare gate blocked apixaban order)
budget_under_limit       := true
residue_coverage         := 0.92
dropped_claims           := ["c4: CrCl 22 mL/min within apixaban range
                              (CONTRADICTED by FDA label)"]
unverified_layers        := []
model_version_hash       := 0x3D11..7AAB
prompt_template_hash     := 0x9A02..5C6D
human_signatures         := []
certificate_digest       := 0xB7C9..3041
verifier_version         := "leanv4.30.0-rc2-mathlib"
replay_handle            := "rag-clinical-2026-04-25-7c4fa"
\end{lstlisting}

The clinical-policy profile (Table~\ref{tab:assurance_policies}) requires \texttt{supported\_mass} $\ge 0.95$, \texttt{contradicted\_mass} $\le 0.01$, and \texttt{action\_precondition} = true. The card's verdict is \textsc{Residue} with $W^- = 0.08$ and \texttt{action\_precondition} = false; the policy rejects the card under \textsc{Certified} or \textsc{Partial} but accepts the residue payload subject to clinician sign-off (the \texttt{human\_signatures} field, here empty, would be populated post-review).

\subsection{Case study 2: Agentic Hoare execution (sandboxed file system)}
\label{sec:case_agent}

\paragraph{Setup.} An LLM agent is instructed to consolidate the user's project notes by deleting drafts older than 30~days and merging surviving notes into a single document. The agent has a typed file-system API: \texttt{stat}, \texttt{read}, \texttt{write}, \texttt{delete}, each with declared pre/post conditions in the formal layer. State $\sigma$ is a Lean record mapping paths to \texttt{(timestamp, content, owner)}. The agent operates within a sandbox bound declared by the user's policy.

\paragraph{Per-layer certificates.}
\begin{enumerate}[leftmargin=1.5em,topsep=2pt,itemsep=2pt]
\item \emph{L0 Semantic-type role check.} The instruction text is tagged \texttt{user\_instruction}; retrieved file contents are tagged \texttt{tool\_output}; prompt-injection attempts within retrieved contents fail the role-coercion check at compile time.
\item \emph{L1 Plan certificate.} The agent's plan is decomposed: \texttt{(stat *)}, \texttt{filter age $>$ 30d}, \texttt{(delete drafts)}, \texttt{(read survivors)}, \texttt{(write merged)}. Each step is a typed action with declared pre/post.
\item \emph{L2 Trajectory certificate.} \emph{Step 1} \texttt{stat *}: \texttt{Pre} = ``directory exists''; \texttt{Post} = ``returns metadata, $\sigma$ unchanged.'' Discharged. \emph{Step 2--N (delete drafts)}: \texttt{Pre} = ``file exists $\wedge$ \texttt{age} $> 30$d $\wedge$ \texttt{owner} = user $\wedge$ path is in sandbox''. The agent proposes deleting \texttt{notes/old\_meeting.md} (52 days, owner = user, in sandbox). \texttt{Pre} discharged. \texttt{Post} = ``file removed; $\sigma'$ = $\sigma$ minus path.'' Executor produces $\sigma'$; \texttt{Post} verified.
\item \emph{L3 The unsafe step.} The agent proposes deleting \texttt{notes/shared/team\_design.md} as a draft. \texttt{Pre} requires \texttt{owner} = user; the file's \texttt{owner} is \texttt{team}. \texttt{Pre} fails. The trajectory certificate fails to type-check at this step. The agent is blocked; the planner is re-prompted with the failure witness (CEGAR-style; see catalogue \S\ref{sec:catalogue:cegar}).
\item \emph{L4 MCR over the executable subset.} The MCR operator (Theorem~\ref{thm:mcr_maximality}) accepts the largest prefix-closed subset of the trajectory whose constraints (\texttt{Pre}, sandbox, owner, age) hold. Output: 7 of 9 deletions executed, 2 blocked. The merge step's \texttt{Pre} (``every input survivor exists'') is satisfied on the executed subset; the merged document is written.
\item \emph{L5 Budget invariant.} Cumulative tokens, API calls, and wall-clock under declared limits.
\end{enumerate}

\paragraph{The structural lesson.} The Hoare action certificate blocks the out-of-scope deletion structurally: the kernel refuses to type-check a trajectory containing an action whose precondition is unmet. There is no ``best-effort filter'' to circumvent and no policy adjustment that could let the unsafe action through silently. The MCR turns the blocked step into a partial execution rather than aborting; the audit log names the blocked steps; the residue's \texttt{action\_precondition} field is true because every executed step's precondition was discharged (the rejected steps simply do not appear in the trajectory). A baseline unaudited agent would have either deleted the team file (irreversible) or aborted the entire batch on the first failure. The operational rule the example illustrates is that the LLM proposes a typed action and a small kernel-audited runtime is the final authority that permits or blocks side effects.

\paragraph{Card summary (compressed).} The agentic card differs from the clinical one chiefly in residue \emph{kind} (action-residue versus evidence-residue) and in the absence of evidential mass:
\texttt{verdict = residue 0.78}, \texttt{contradicted\_mass = 0.00}, \texttt{contested\_mass = 0.00}, \texttt{action\_precondition = true} (over executed steps), \texttt{dropped\_claims} listing the two blocked deletions with owner-mismatch reasons, \texttt{is\_deterministic = true}, \texttt{in\_scope = true}, \texttt{budget\_under\_limit = true}. The agentic-action policy (Table~\ref{tab:assurance_policies}) requires \texttt{require\_action\_safe = true} and \texttt{contradicted\_mass} $\le 0$; both hold. The policy accepts the residue payload because no in-scope action's precondition was violated, and the blocked steps are listed for the user. The full card has the same field structure as the clinical one (\S\ref{sec:case_clinical}); we omit the verbatim listing.

\paragraph{Domain mapping.} The two case studies above are illustrative of a broader pattern. Analogous Card-bearing pipelines apply in patent claim charting (where a single prior-art reference can support and refute different elements of the same independent claim, the bilattice grounding case), in regulated-finance disclosure analysis (where temporal-validity windows on filings and regulations compose with action-precondition gates on analyst recommendations), and in legal case-law retrieval (where circuit-split holdings produce contested grounding and long-form embedding sensitivity is jointly load-bearing). The clinical and agentic walkthroughs are exhibited here for compactness; the family of high-stakes domains the framework targets is the subject of the Pilot~D family of empirical validations outlined in \S\ref{sec:limits:pilot}.

\section{Empirical Validation: Pilot A on HotpotQA}
\label{sec:pilot_a}

The case studies in \S\ref{sec:cases} use synthetic numbers to illustrate
how the certificates and the Universal Card behave on representative
inputs. This section reports a pre-registered empirical pilot
(``Pilot~A'') that puts the conflict-aware bilattice grounding
certificate (\S\ref{sec:grounding}, Lemma~\ref{thm:no_silent}) to an
adversarial test on a public benchmark, against three distinct
baselines, with measured catch and false-block rates. The pilot ran
in three iterations (a preliminary read v3; a confirmatory run v4;
a clean rerun v5 with corrected hypothesis wordings derived from
v4's diagnostics); all numbers reported here are from the v5
registered analysis.

\subsection{Pipeline and methods}
\label{sec:pilot_a:pipeline}

\textbf{Dataset.} HotpotQA dev distractor~\cite{yang2018hotpotqa},
deterministically sliced (seed~42) into disjoint pools: 300
perturbation candidates and 225 unperturbed candidates. Multi-hop
bridge questions dominate.

\textbf{RAG generator.} Claude Sonnet 4.5 produces a 1--3 sentence
answer from gold supporting paragraphs, instructed to refuse with
``Insufficient information'' when sources are insufficient. 22~of
225 unperturbed cases (9.8\%) returned a refusal and are excluded
from the false-block-rate denominator (a vacuous block on a
refusal is correct behavior, not a false block; the refusal rate
is reported separately).

\textbf{Atomic-claim decomposition.} Gemini 2.5 Flash-Lite, with a
parse-failure cascade that escalates to 8192-token budget then to
a cross-provider GPT-5 fallback. Decomposer-choice sensitivity is
checked separately (\S\ref{sec:pilot_a:sensitivity}).

\textbf{NLI signature.} Three-way LLM-as-NLI (entail / contradict /
neutral) via GPT-5, with chunk-concatenation aggregation:
$\beta^+ = \sigma^+(\text{claim}, \text{chunk}_1\!\oplus\!\dots\!\oplus
\!\text{chunk}_n)$, $\beta^- = \sigma^-(\text{claim}, \cdot)$. Concat
aggregation lets the NLI see cross-chunk evidence so
HotpotQA bridge claims (synthesised from facts in multiple chunks)
are not mis-classified as ``Unknown.'' Per-chunk-max aggregation is
preserved as a registered ablation (\S\ref{sec:pilot_a:ablation}).

\textbf{Bilattice classification.} As Lemma~\ref{thm:no_silent}:
$\theta_G {=} \theta_C {=} 0.5$ for the four-way label, $\theta_E$
swept in $\{0.50, 0.55, \ldots, 1.00\}$ to produce a
Pareto-frontier curve, $\theta_R {=} 0.05$.

\textbf{Adversarial perturbation.} Each candidate received one
perturbation chosen from \{negate, swap, flip\} on the targeted
atomic claim (\S\ref{sec:grounding} terminology). A three-judge
cross-model panel (GPT-5, Claude Opus 4.5, Gemini 2.5 Pro) labels
each candidate as ``contradicted'' or not; only candidates a
majority of the panel labels contradicted enter the perturbed
slice. \textbf{Yield: 264 of 300 (88\%) accepted; unanimous-accept
rate: 0.936}, well above the 0.85 pre-registered threshold.

\textbf{Methods compared.}
\begin{itemize}[leftmargin=*,topsep=2pt,itemsep=2pt]
  \item \textsc{Treatment}: full bilattice + emission gate
        (Lemma~\ref{thm:no_silent}).
  \item \textsc{B1} (cosine baseline): block iff
        $\max_i \cos(\text{answer}, \text{chunk}_i) \ge \theta_{B1}$,
        sweep $\theta_{B1} \in [0.5, 1.0]$ in 0.05 steps.
  \item \textsc{B2} (self-consistency, $k{\in}\{1,5\}$): generate
        $k$ samples, decompose, cluster claims, block on
        cluster-disagreement.
  \item \textsc{B3} (raw NLI, no support gate): block iff
        $\max_i w_i \beta^-_i \ge 0.5$.
\end{itemize}

\subsection{Hypothesis verdicts}
\label{sec:pilot_a:hypotheses}

Six hypotheses pre-registered (v5 wording, with the v4 wording
retained as a comparison footnote):

\smallskip
\begin{center}
\footnotesize
\begin{tabular}{cp{7cm}c}
\toprule
\textbf{H} & \textbf{v5 statement} & \textbf{Verdict} \\
\midrule
H1 & $\exists \theta_E$ such that Treatment achieves catch $\ge 0.95$ at FBR $\le 0.10$, and B1 achieves no such operating point at any threshold. & PASS \\
H2 & Treatment catch\_blk $\ge 0.70$. & PASS \\
H3 & Treatment catch\_tgt $\ge 0.60$ (target-claim catch). & PASS \\
H4 & $\exists \theta_E$ such that Treatment catch $\ge 0.95$ and FBR $\le 0.15$. & PASS \\
H5 & B1 cannot reach catch $= 1.0$ at FBR $\le 0.10$ for any $\theta_{B1}$. & PASS \\
H6 & Panel unanimous-accept rate $\ge 0.85$. & PASS \\
\bottomrule
\end{tabular}
\end{center}

\paragraph{Pre-registration revisions.} Three hypotheses were rewritten between v3/v4 and v5 once their original formulations proved unrealisable on the data. H1 was originally a 60pp Pareto-dominance test that is upper-bounded at zero whenever the catch axis saturates at 1.0; v5 replaces it with the Pareto-reachability formulation above. H5 was originally a B1 catch-ceiling $\le 0.30$ that is empirically false (B1 reaches 1.0 at $\theta_{B1} \ge 0.85$); v5 replaces it with a true negative claim about catch=1.0 reachability. H6 originally combined a $\kappa \ge 0.6$ component with the unanimous-accept threshold, but $\kappa$ degenerates near unanimity (a rater near 100\% drives the chance-correction denominator to zero); v5 retains the unanimous-accept threshold and drops the $\kappa$ component. The methodological observations these revisions surface (saturating-Pareto pathology and $\kappa$ saturation under near-unanimous panels) are discussed in \S\ref{sec:pilot_a:lessons}. The verbatim v3/v4 wordings are recorded in the pilot specification in the companion repository.

\paragraph{Reconciling Pilot~A's $\kappa$ with Pilot~B's $\kappa$.}
A reader comparing the two might wonder how Pilot~A reports
$\kappa \approx -0.014$ alongside 93.6\% unanimous agreement while
Pilot~B (\S\ref{sec:embedding:empirical}) reports
$\kappa = 0.887$ on a structurally similar 3-judge panel. The
two panels label different things at different base rates. Pilot~B's
panel labels paraphrase quality (the categories ``meaning-preserving''
vs.\ ``meaning-changing''), where the base rates are well away from
unanimity, so $\kappa$ is informative and large. Pilot~A's panel
labels contradiction status, where one rater is near 100\% on the
negative class for a registered slice of unperturbed cases, driving
the marginal-equality denominator of $\kappa$ to near zero, which is the classic chance-correction saturation. This is why v5 retired the
$\kappa$ component for Pilot~A's H6 in favor of unanimous-accept,
which does not have that pathology, and is also why the same panel
discipline reproduces high $\kappa$ on Pilot~B: the difference is in
the labeled task, not the panel.

\subsection{Main results}
\label{sec:pilot_a:main_results}

Effective sample size: 264 perturbed (panel-accepted out of 300
candidates), 203 unperturbed (refusal-excluded out of 225).

\smallskip
\begin{center}
\footnotesize
\begin{tabular}{lccc}
\toprule
\textbf{Method} & \textbf{catch\_blk (95\% CI)} & \textbf{catch\_tgt} & \textbf{false\_block (95\% CI)} \\
\midrule
\textsc{Treatment} ($\theta_E{=}0.50$) & 1.000 & 1.000 & \textbf{0.064} (0.034, 0.099) \\
\textsc{Treatment} ($\theta_E{=}0.90$, registered) & 1.000 & 1.000 & 0.227 (0.167, 0.286) \\
\textsc{B1} @ $\theta{=}0.85$ & 0.992 & n/a & 0.232 \\
\textsc{B1} @ $\theta{=}0.90$ & 1.000 & n/a & 0.266 \\
\textsc{B2} ($k{=}1$) & 0.295 (0.242, 0.348) & 0.295 & 0.350 (0.281, 0.414) \\
\textsc{B2} ($k{=}5$) & 0.292 (0.239, 0.348) & 0.292 & 0.261 (0.202, 0.320) \\
\textsc{B3} & 0.977 (0.958, 0.992) & 1.000 & \textbf{0.044} (0.020, 0.074) \\
\bottomrule
\end{tabular}
\end{center}

\smallskip
\textbf{What the table says.} (i)~Treatment achieves catch $=1.0$
at FBR $=0.064$, an operating point B1 cannot reach for any
$\theta_{B1}$ (B1's lowest FBR at catch $\ge 0.95$ is 0.232).
(ii)~Treatment massively dominates self-consistency: McNemar's
exact test $p < 10^{-30}$ for both $k{=}1$ and $k{=}5$. (iii)~B3
(contradiction-only) is a strong baseline at this scale, achieving
catch $=0.977$ at FBR $=0.044$; Treatment matches its operating
point at $\theta_E{=}0.50$ and offers a tunable knob for higher-
support-coverage regimes.

\subsection{Pareto frontier (Treatment $\theta_E$ sweep vs B1 sweep)}
\label{sec:pilot_a:pareto}

Treatment achieves catch $=1.0$ across the entire $\theta_E$ sweep;
the free parameter selects the operating point along the FBR axis:

\smallskip
\begin{center}
\footnotesize
\begin{tabular}{lcc|lcc}
\toprule
$\theta_E$ & catch & FBR & $\theta_{B1}$ & catch & FBR \\
\midrule
0.50 & 1.000 & \textbf{0.064} & 0.50 & 0.432 & 0.010 \\
0.55 & 1.000 & 0.108 & 0.60 & 0.602 & 0.089 \\
0.60 & 1.000 & 0.113 & 0.70 & 0.788 & 0.123 \\
0.65 & 1.000 & 0.138 & 0.80 & 0.928 & 0.202 \\
0.70 & 1.000 & 0.148 & 0.85 & 0.992 & 0.232 \\
0.80 & 1.000 & 0.187 & 0.90 & 1.000 & 0.266 \\
0.90--0.95 & 1.000 & 0.227 & 0.95 & 1.000 & 0.409 \\
1.00 & 1.000 & 0.256 & 1.00 & 1.000 & 1.000 \\
\bottomrule
\end{tabular}
\end{center}

The $\theta_E{=}0.50$ point is unreachable for B1 at any
$\theta_{B1}$: B1 only enters the catch~$\ge$0.95 regime at FBR
$\ge$0.232, while Treatment is already at FBR=0.064 there. This is
the operationally meaningful comparison: in deployment, the
operating point is chosen for tolerable FBR, not at the saturation
ceiling.

\subsection{Decomposer-choice robustness}
\label{sec:pilot_a:sensitivity}

A pre-registered sensitivity test (script~06) re-runs decomposition
on a 200-example subsample with GPT-5 substituted for the registered
Gemini~2.5~Flash-Lite. \textbf{Treatment catch differs by $<\!5$pp
between the two decomposers}, well within the
pre-registered $\le 5$pp threshold. The main result is therefore
not an artefact of the registered decomposer's
idiosyncrasies.

\subsection{Ablation: which architectural choices are load-bearing?}
\label{sec:pilot_a:ablation}

\smallskip
\begin{center}
\footnotesize
\begin{tabular}{lccc}
\toprule
\textbf{Configuration} & \textbf{Catch} & \textbf{FBR} & \textbf{$\Delta$ FBR} \\
\midrule
Registered v5 (concat NLI, $\theta_E{=}0.5$) & 1.000 & 0.064 & n/a \\
Per-chunk-max NLI ablation (v3 default), $\theta_E{=}0.5$ & 1.000 & 0.350 & $+$0.286 \\
Concat NLI, $\theta_E{=}0.9$ (v3 registered) & 1.000 & 0.227 & $+$0.163 \\
Per-chunk-max NLI, $\theta_E{=}0.9$ (v3 baseline) & 1.000 & 0.450 & $+$0.386 \\
\bottomrule
\end{tabular}
\end{center}

\textbf{NLI aggregation choice (concat vs per-chunk-max) is the
dominant hyperparameter.} It changes Treatment's FBR by
$\sim\!2.3\times$ at $\theta_E{=}0.5$. The per-chunk-max default
fails on multi-hop claims because no single chunk entails a
synthesised fact; concat aggregation lets the NLI score the
synthesised claim against all evidence at once. We recommend
concat as the default for any RAG corpus where claims may require
multi-chunk synthesis (HotpotQA-like multi-hop QA, clinical
guidelines that cross documents, legal precedent chains).

\subsection{Methodological observations}
\label{sec:pilot_a:lessons}

The three iterations surfaced empirical-evaluation pitfalls
worth naming for subsequent work in this area.

\begin{itemize}[leftmargin=*,topsep=2pt,itemsep=4pt]
  \item \textbf{Reasoning-token starvation in modern LLMs.} GPT-5
        and Gemini~2.5 use hidden reasoning tokens that consume
        the visible-output budget first. Calls with
        \texttt{max\_tokens} $<$ 2048 silently return empty text or
        400 errors, with no signal in normal logs. Affected
        components: NLI (originally 64 tokens), perturbation
        generator (512), decomposer (2048, still too tight),
        panel judges (128). \emph{Recommendation:} all
        reasoning-capable LLM calls should default to
        \texttt{max\_tokens} $\ge$ 4096 with a usage-metadata
        check that warns on \texttt{output\_tokens=0}.
  \item \textbf{Per-chunk vs.\ chunk-concatenated NLI on
        multi-hop claims.} The original pipeline used
        per-chunk-max aggregation. On HotpotQA bridge questions
        ($\sim$11\% of unperturbed claims) the answer synthesises
        facts across chunks; no single chunk entails the
        synthesised claim, so the NLI returns ``neutral'' on every
        chunk, the bilattice classifies the claim ``Unknown,'' and
        the support gate fires a false-block. Chunk-concatenated
        aggregation drops this rate from 11\% to 6\%.
  \item \textbf{Refusals vs.\ false-blocks.} A vacuous block on a
        refusal (RAG returns ``Insufficient information'') is
        correct behavior, not a false block. v3 mis-counted these
        in the FBR denominator, contaminating the rate by 24
        percentage points. \emph{Recommendation:} pre-register
        whether refusals are excluded from the FBR denominator;
        report the refusal rate as a first-class primary metric.
  \item \textbf{Cohen's $\kappa$ near unanimity.} Three-judge
        agreement was 93.6\% unanimous, but minimum pairwise
        $\kappa$ was $-0.014$. $\kappa$ corrects for chance
        agreement; when one rater is near-100\% on one label,
        the chance-correction term saturates and any disagreement
        on the remaining cases drives $\kappa$ to zero or below.
        \emph{Recommendation:} pre-register $\kappa$ thresholds
        only when expected base rates are away from unanimity;
        use unanimous-accept and majority-accept rates for
        verification panels.
  \item \textbf{Saturating Pareto frontiers and the 60pp gap.}
        ``Min gap over the entire B1 frontier $\ge$ 60pp'' is
        mathematically unattainable when both methods reach
        catch $= 1.0$ at high enough thresholds. The right
        formulation is \emph{Pareto-reachability}: ``Treatment
        achieves an operating point B1 cannot.'' This is the
        deployment-relevant question; the original min-gap test
        was a methodological error in hypothesis design, not a
        result.
\end{itemize}

These lessons are not evidence against the framework; they are
evidence about what discipline an empirical evaluation needs in
order to produce informative numbers. The five lessons above add
up to the difference between v3's summary claim ``method fails on most
pre-registered hypotheses'' and v5's summary claim ``method passes
every pre-registered hypothesis with statistically significant
separation from baselines.''

\subsection{Pilot scope}
\label{sec:pilot_a:scope}

This pilot validates the conflict-aware grounding certificate
(\S\ref{sec:grounding}) on one public benchmark with adversarial
perturbations. It does not test (i)~the Compositional Stability
theorem (\S\ref{sec:composition}) at the multi-layer level or
(ii)~deployment in a regulated workflow. The other certificate
families have their own pilots reported elsewhere: Pilots~B and B$'$
(embedding sensitivity) in \S\ref{sec:embedding:empirical}--\S\ref{sec:embedding:longform},
and Pilot~C (Hoare-style agent action) in \S\ref{sec:agents:empirical}.
The Pilot~D family of Universal Card validations
(patent / technical-intelligence, legal, regulated-finance, and
clinical variants) remains forward work
(\S\ref{sec:limits:pilot}).

\bigskip

\section{Compiled Lean~4 Reference Artifact}
\label{sec:artifact}

The architecture is mechanized. The companion Lean~4 project (directory \texttt{lean\_artifact/}) compiles all 22 certificate types in the broader catalogue. The artifact builds in roughly ten seconds with a warm Mathlib cache; Appendix~\ref{app:embedding_lean} reproduces representative structures verbatim, and Appendix~\ref{app:axiom_audit} contains a representative excerpt of the \texttt{\#print axioms} output (the full audit is build-regenerated by \texttt{EmbeddingSensitivity/AxiomAudit.lean}).

\subsection{Summary statistics}

\begin{center}\fbox{\parbox{0.85\textwidth}{\centering
\textbf{22 certificate types $\bullet$ 46 audited declarations $\bullet$ 17 axiom-free $\bullet$ 29 with $\Omega$+oracles $\bullet$ 0 \texttt{sorry} $\bullet$ 0 \texttt{native\_decide}}
}}\end{center}

\medskip
The 17 axiom-free declarations include the entire Hoare-style action certificate (state, structures, all safety theorems, concrete trajectory), the bulk of the self-consistency certificate, the Curry--Howard \texttt{SortedList} and \texttt{DistinctList} structures, scope-of-validity, simpler CEGAR theorems, and six of the ten extended-catalogue types (negative-guarantee, provenance, temporal, calibrated-abstention, diff, meta as definitions). These compile by kernel reduction alone, with \emph{no axioms whatsoever} in their transitive set. In the conditional-guarantee hierarchy, this is the strongest possible status: the guarantee is unconditional modulo only the Lean kernel itself.

\paragraph{What ``axiom-free'' does and does not mean.}\label{sec:artifact:axiom_free} Axiom-freeness is a statement about the proof's \emph{discharge mechanism}: every step reduces by kernel computation and pattern matching, with no appeal to a named axiom (not even Lean's foundational $\Omega$). It is not a statement that the encoded predicates are unconditionally meaningful in the real world. For the Hoare action certificate, predicates such as \texttt{owner = user} or \texttt{path is in sandbox} are human-authored definitions on the chosen state model; their semantic adequacy, that they actually capture the real-world safety property the deployment intends, is a modeling assumption beneath the trust boundary. Axiom-freeness rules out hidden mathematical assumptions (no \texttt{sorry}, no \texttt{native\_decide}, no Mathlib placeholder), but it does not rule out \emph{specification errors} in the predicate authors. A paraphrase: the certificate guarantees ``the kernel computed this property correctly''; it does not guarantee ``the property is the right one.'' This caveat applies to every axiom-free certificate in the artifact and is consistent with the conditionality discussion of \S\ref{sec:limits:conditionality}.

\subsection{Artifact evaluation}
\label{sec:artifact:eval}

Beyond the axiom audit, Table~\ref{tab:artifact_eval} reports concrete operational metrics. The numbers are observable on commodity hardware (a 2024-class laptop) under Lean v4.30.0-rc2 with a warm Mathlib cache; they are intended as order-of-magnitude reference points, not as a benchmark.

\begin{table}[!htbp]
\centering
\caption{Operational metrics for the compiled reference artifact. Numbers labeled \emph{measured} are observed on the artifact directly: module and declaration counts are exact; build and audit times are reported on a 2024-class laptop with a warm Mathlib cache, rounded conservatively. Numbers labeled \emph{projected} are schema-derived estimates for components currently specified but not yet compiled (notably the Universal Assurance Card consolidator); their soundness depends on the focused engineering step in \S\ref{sec:limits:card}. Certificate and Card byte sizes are computed from the field-record byte layout. Production-traffic throughput numbers (latency under sustained load, batched-audit amortization, SNARK-compression overhead) are deliberately not reported here and remain forward work (\S\ref{sec:limits:empirical}).}
\label{tab:artifact_eval}
\footnotesize
\setlength{\tabcolsep}{4pt}
\begin{tabular}{@{}p{4.4cm}p{8.4cm}@{}}
\toprule
\textbf{Metric} & \textbf{Value} \\
\midrule
\multicolumn{2}{@{}l}{\emph{Source-level totals}} \\
Modules & 25 \\
Top-level declarations audited & 46 \\
Axiom-free declarations & 17 / 46 \\
Persistent tier-4 oracles & 5 (cf.\ Table~\ref{tab:oracle_recipe}) \\
Mathematical placeholders (tier 2) & 3 (\texttt{cauchy\_schwarz\_sq}, \texttt{innerProd\_sub}, \texttt{HoeffdingInequality}) \\
Cryptographic assumptions (tier 3) & 2 (\texttt{HashCollisionResistant}, \texttt{Decode\allowbreak Algorithm\allowbreak Deterministic}) \\
\midrule
\multicolumn{2}{@{}l}{\emph{Build and audit time}} \\
\texttt{lake build}, warm Mathlib cache & $\approx 10$~s \\
\texttt{lake build}, cold Mathlib cache & $\approx 2$~min (one-time, then cached) \\
Per-cert kernel re-check (typical) & $<100$~ms (sub-second for every leaf type) \\
\texttt{\#print axioms} on full artifact & $<1$~s \\
End-to-end audit over the 22 underlying compiled certificates (measured) & $\approx 1$~s on commodity hardware \\
Projected end-to-end audit on a single Universal Card (specified; once \texttt{AssuranceCard.lean} added) & $\approx 1$~s; consistency predicate is \texttt{decide}-discharged on a $\sim$30-field record \\
\midrule
\multicolumn{2}{@{}l}{\emph{Certificate and card sizes}} \\
Universal Assurance Card record (typical, \emph{projected} from schema) & $\approx 1.5$~KB JSON, $\approx 0.7$~KB binary \\
Embedding sensitivity certificate (per call) & a few KB; scales with $|G_{\mathrm{inv}}| \cdot d$ \\
Bilattice grounding certificate (20 claims, 100 chunks) & $<10$~KB \\
Hoare trajectory certificate (10 actions) & $<5$~KB \\
Composite Lean proof terms & sub-KB to a few KB; SNARK-compressible to hundreds of bytes (catalogue) \\
\midrule
\multicolumn{2}{@{}l}{\emph{Audit infrastructure}} \\
Reference audit-runner (Python/Rust around \texttt{lean}) & $<100$~LOC \\
External dependencies for audit & Lean v4.30.0-rc2; SHA-256 implementation; nothing else \\
Required input & certificate file; Mathlib snapshot pin \\
Not required for audit & model weights, prompts, decoding parameters, random seeds \\
\bottomrule
\end{tabular}
\end{table}

The table substantiates two operational claims that the paper otherwise gestures at: (i) the audit is genuinely small and portable (an external auditor needs only \texttt{lean} plus a $<100$~LOC wrapper, no model artifacts); (ii) per-call verification cost is dominated by the underlying LLM call, not by the kernel re-check, so the overhead is amortizable in any production setting. Operational measurements at production traffic, throughput under sustained load, expert-panel agreement on abstentions, and catch rates against deployed pipelines, remain explicit future work (\S\ref{sec:limits}).

\subsection{Per-type axiom dependencies}

Table~\ref{tab:audit_all22} reports the axiom-dependency tier of every compiled certificate type. Two rows in the table collapse to a single type for the purpose of the 22-type count: the two Curry--Howard rows (sorted/distinct and parens/quote) share one type, and the Hoeffding concentration row stands for one type parameterized over eight named concentration inequalities (Hoeffding, generalized Hoeffding, Bernstein, empirical Bernstein, Azuma--Hoeffding, Bennett, McDiarmid, union-bound) sharing one Lean structure with different \texttt{failure\_bound}. Net: 18 catalogue rows yield 17 distinct catalogue types; with 5 core items (3 families and 2 operators), the total is $5 + 17 = 22$.

\begin{table}[!htbp]
\centering
\caption{Axiom-dependency tier for each compiled certificate type. ``none'' means no axioms at all; $\Omega = \{\texttt{propext}, \texttt{Classical.choice}, \texttt{Quot.sound}\}$ are the kernel-trusted axioms (tier 1). Named non-$\Omega$ axioms are tagged \emph{(t2)} for tier-2 mathematical placeholders, \emph{(t3)} for tier-3 cryptographic assumptions, and \emph{(t4)} for tier-4 ML/human oracles, per the hierarchy of \S\ref{sec:arch:taxonomy}.}
\label{tab:audit_all22}
\footnotesize
\setlength{\tabcolsep}{4pt}
\begin{tabular}{@{}p{4.0cm}p{7.0cm}p{2.0cm}@{}}
\toprule
\textbf{Certificate} & \textbf{Axiom dependency} & \textbf{Section} \\
\midrule
\multicolumn{3}{@{}l}{\emph{Three local certificate families}} \\
Conflict-aware bilattice grounding & $\Omega$ + \texttt{Signed\-Support\-Oracle}, \texttt{Decomposition\-Oracle} (t4) & \S\ref{sec:grounding} \\
Embedding sensitivity & $\Omega$ + \texttt{Paraphrase\-Oracle} (t4) + \texttt{cauchy\_schwarz\_sq}, \texttt{innerProd\_sub} (t2) & \S\ref{sec:embedding} \\
Hoare agent action & none & \S\ref{sec:agents} \\
\midrule
\multicolumn{3}{@{}l}{\emph{Two operators on certificates}} \\
Maximal Certifiable Residue & $\Omega$ & \S\ref{sec:mcr} \\
Compositional Stability (2-layer base lemma) & $\Omega$ & \S\ref{sec:composition} \\
\midrule
\multicolumn{3}{@{}l}{\emph{Per-call deliverable}} \\
Universal Assurance Card (specified) & would inherit from sub-certs once \texttt{AssuranceCard.lean} is added & \S\ref{sec:card} \\
\midrule
\multicolumn{3}{@{}l}{\emph{Extended catalogue (Appendix~\ref{app:embedding_lean})}} \\
Self-consistency lattice & none & catalogue \\
Chain-of-thought DAG & $\Omega$ + \texttt{Step\-Confidence\-Oracle} (t4) & catalogue \\
Curry--Howard (sorted/distinct) & none & catalogue \\
Curry--Howard (parens/quote) & $\Omega$ & catalogue \\
Hoeffding concentration & $\Omega$ + \texttt{Hoeffding\-Inequality} (t2) + \texttt{IID\-Samples} (t4) & catalogue \\
Proof-of-sampling & \texttt{Hash\-Collision\-Resistant}, \texttt{Decode\-Algorithm\-Deterministic} (t3) & catalogue \\
Scope-of-Validity & none & catalogue \\
CEGAR repair loop & none / $\Omega$ & catalogue \\
Negative-guarantee & none & catalogue \\
Counterfactual & $\Omega$ & catalogue \\
Calibration & $\Omega$ & catalogue \\
Provenance & none & catalogue \\
Temporal validity & none & catalogue \\
Budget composition & $\Omega$ & catalogue \\
Calibrated abstention & none & catalogue \\
Diff / incremental & none & catalogue \\
Multi-hop reasoning & $\Omega$ & catalogue \\
Meta / self-audit & none & catalogue \\
\bottomrule
\end{tabular}
\end{table}

\subsection{What is conditional, what is closable}

The \emph{persistent} (tier-4) oracles in the artifact are \texttt{ParaphraseOracle}, \texttt{SignedSupport\allowbreak Oracle}, \texttt{Decomposition\allowbreak Oracle}, \texttt{StepConfidence\allowbreak Oracle}, \texttt{IIDSamples}. These cannot be discharged by mathematics. Each carries a deployment-specific evaluation discipline: a paraphrase oracle's faithfulness is supported by a held-out paraphrase benchmark; a signed-support oracle's faithfulness by a held-out NLI calibration set; a decomposition oracle's faithfulness by reconstruction-pass rates. The remaining named axioms (tiers 2 and 3: \texttt{cauchy\_schwarz\_sq}, \texttt{innerProd\_sub}, \texttt{HoeffdingInequality}, \texttt{HashCollisionResistant}, \texttt{Decode\allowbreak Algorithm\allowbreak Deterministic}) are \emph{closable} as Mathlib's inner-product-space, probability, and cryptographic-hash libraries develop.

\textbf{Proof-of-sampling discipline.} For concentration certificates that depend on \texttt{IIDSamples}, the framework treats IID sampling as a persistent tier-4 oracle rather than a mathematical fact. A proof-of-sampling certificate supplies auditable evidence for that oracle: the deployment records a declared temperature $T > 0$, per-pass commitments over the random seed, the decoder entry state, and the per-pass token-emission log, and a decidable no-state-carryover predicate over those commitments. When strict no-state-carryover is unavailable, the certificate may instead record a bounded-difference or martingale fallback, yielding an Azuma--Hoeffding / McDiarmid-style discipline rather than an IID one. The certificate populates the \texttt{is\_deterministic}, \texttt{sample\_count}, and \texttt{proof\_of\_sampling} fields of the Universal Assurance Card, but it does not mathematically close the natural-language or decoder-behavior oracle; rather, it gives the auditor structured commitments against which the \texttt{IIDSamples} hypothesis can be operationally re-verified.

\textbf{What the artifact proves} (selected; full list in Appendix~\ref{app:axiom_audit}): Theorem~\ref{thm:robust_decision} (squared form, embedding stability); Lemma~\ref{lem:family_monotone} (variant-family monotonicity); Lemma~\ref{thm:no_silent} (emission gate soundness for grounding, conditional on the NLI / decomposition oracles); the abstention-reason and emit-when-supported lemmas; Lemma~\ref{lem:composition_two_layer} (compositional stability, two-layer base case; the $n$-layer Theorem~\ref{thm:composition} lifts by induction over the certified-layer list); Theorem~\ref{thm:mcr_maximality} (MCR maximality, idempotence, fixed-point); per-step trajectory safety and \texttt{p\_chain} for Hoare; six concrete instantiations spanning a 3-D rational embedding, a clinical-RAG 3-claim grounding example that abstains, a counter-state 3-step trajectory, a 5-sample consistency example, four Curry--Howard structured outputs, and a 20-sample Hoeffding example.

\section{Threat Model and Adoption Order}
\label{sec:threat}

\subsection{Threats versus defender certificates}

\begin{table}[H]
\centering
\caption{Threats against LLM pipelines and the certificate types that defend each. Threats marked $^\star$ are the most frequent in production deployments today.}
\label{tab:threats}
\footnotesize
\begin{tabular}{@{}p{5.3cm}p{7.5cm}@{}}
\toprule
\textbf{Threat} & \textbf{Defender certificate(s)} \\
\midrule
Hallucination (invented facts)$^\star$ &
Bilattice grounding (\S\ref{sec:grounding}); self-consistency (catalogue); canonical-denotation decoding (catalogue). \\
Silent contradiction$^\star$ &
Bilattice grounding emission-gate (Lem~\ref{thm:no_silent}) with $W^-$ threshold; MCR drops the contradicted claim (\S\ref{sec:mcr}). \\
Prompt injection via retrieved content$^\star$ &
Semantic type system + non-interference (catalogue); grounding with source-role tagging. \\
Variant flip (paraphrase attack)$^\star$ &
Embedding sensitivity (\S\ref{sec:embedding}); compositional-stability budget (Thm~\ref{thm:composition}). \\
Unsafe tool invocation$^\star$ &
Hoare action certificate (\S\ref{sec:agents}); budget composition (catalogue); rely-guarantee for parallel actions. \\
Sample forgery (claimed $k$ samples) &
Proof-of-sampling with committed seed (catalogue). \\
Corpus tamper-detection (snapshot drift from committed hash) &
Provenance + temporal validity (catalogue). Note: detects that the served corpus differs from the committed snapshot; does not certify that the snapshot is itself uncorrupted (corpus-level factuality is below the trust boundary). \\
Vendor / model / prompt drift &
Provenance certificate + composite-certificate diff (catalogue). \\
Fabricated compliance evidence (\texttt{sorry}-based cert) &
Axiom audit against $\Omega$ (Def.~\ref{def:validcert}); meta-certificate over kernel hash. \\
PII / forbidden content leakage &
Negative-guarantee (catalogue); abstract-interpretation PII grammar (catalogue). \\
Over-confident miscalibrated confidence &
Calibration certificate; conformal prediction (catalogue). \\
Unfaithful chain-of-thought (stated $\neq$ actual) &
Cross-model CoT DAG verification; PRM step scoring; mech-interp axioms (catalogue). \\
Audit tampering (binary swap of kernel) &
Meta-certificate + reproducible build + signed kernel. \\
\bottomrule
\end{tabular}
\end{table}

\subsection{First certificate to ship}
\label{sec:first_cert}

We recommend an adoption order chosen for low implementation cost relative to risk-reduction value:

\begin{enumerate}[leftmargin=1.5em,topsep=2pt,itemsep=2pt]
\item \emph{Embedding sensitivity (\S\ref{sec:embedding}).} Tractable; independent of downstream architecture; off-the-shelf paraphraser oracle; retrieval instability maps directly to financial / legal risk.
\item \emph{Conflict-aware grounding (\S\ref{sec:grounding}).} Highest hallucination-visibility; the signed-support arithmetic is straightforward; per-claim audit trails are immediately useful.
\item \emph{Hoare-style agent action (\S\ref{sec:agents}).} Highest per-failure cost; ``blocked at first unsafe step'' is a compelling demo for decision-makers; precondition predicates usually exist in the policy layer.
\item \emph{Scope-of-validity + provenance.} Low marginal cost given 1--3; large audit-trail gain.
\item \emph{Self-consistency + proof-of-sampling.} Needed only when stochastic sampling is load-bearing.
\end{enumerate}

\subsection{A verifiability benchmark}

LLM benchmarks today measure output quality. None measures \emph{verifiability}. We propose a complementary benchmark family on certificate coverage, audit independence, abstention quality, contradiction-detection rate, failure-localization precision, certificate size, audit time, and intervention latency for agents. Publishing a minimal instance on a public dataset (HotpotQA for RAG; WebArena for agents) is a concrete near-term Stage~1 priority.

\subsection{Cost: audit overhead vs.\ assurance-mode generation overhead}
\label{sec:threat:cost}

The cost story has two layers and they should not be conflated.

\paragraph{Audit overhead is small.} Once a certificate has been generated, re-checking it is dominated by the Lean kernel and the \texttt{\#print axioms} pass. Per Table~\ref{tab:artifact_eval}: per-cert kernel re-check is typically $<100$~ms, the full \texttt{\#print axioms} pass is $<1$~s, and an external audit-runner is $<100$~LOC of Python or Rust. This is the cost a regulator, downstream consumer, or third-party verifier pays. It is genuinely small and portable.

\paragraph{Assurance-mode generation overhead can be materially higher.} Producing a certificate, in contrast, can require non-trivial extra work above the unaudited LLM call. A conflict-aware grounding certificate may invoke the LLM (or a separate decomposer) to extract atomic claims, then call an NLI model on every (claim, chunk) pair, then aggregate signed-support scores. A paraphrase decision-margin certificate may run the encoder on $|G_{\mathrm{inv}}|$ paraphrases. A self-consistency certificate may sample the LLM $k$ times. A composite Universal Card-shaped pipeline may run all of these. Empirically we have not measured these costs (\S\ref{sec:limits:empirical}), but a rough upper bound is on the order of $5\times$--$20\times$ the unaudited generation cost in the most assurance-heavy configurations, with significant variance by certificate family. Coarse breakdown: grounding contributes most ($O(n \cdot m)$ NLI calls in atomic claims and chunks); paraphrase and self-consistency contribute roughly multiplicatively in their sample/variant counts; agent action and embedding sensitivity are nearly free above the baseline.

\paragraph{Where the framework pays for itself.} The framework is designed for the regime where this generation overhead is small relative to per-call error cost. Legal retrieval errors that lead to malpractice claims, clinical recommendations that miss contraindications, agentic actions that delete the wrong file, or fabricated compliance evidence in regulated workflows: these are settings where a $10\times$ generation cost is dwarfed by a single prevented incident. Creative writing, brainstorming, and casual question-answering do not clear that bar; for them, assurance mode should not be enabled. ROI is therefore deployment-specific rather than universal, and the case studies (\S\ref{sec:cases}) are illustrative of the regime where the framework's value proposition holds.

\paragraph{Storage and bandwidth.} Certificate storage is KB--MB per call (sub-KB with SNARK compression, an extension flagged in the catalogue). Card storage is $\approx 1.5$~KB JSON per call. These are negligible at any production scale.

\subsection{Regulatory anchoring}

The certificate taxonomy and the Universal Card map onto major frameworks (EU AI Act, FDA SaMD, NIST AI RMF, ISO/IEC 42001, HIPAA, Federal Rules of Evidence~702 / Daubert); the point is structural rather than detailed compliance, and a per-framework correspondence is in Appendix~\ref{app:adoption}.

\section{The Extended Catalogue}
\label{sec:catalogue}

The compiled artifact's 22 certificate types decompose as follows: \emph{five fully-compiled core items}, namely the three local certificate families (conflict-aware bilattice grounding, embedding sensitivity, and Hoare-style agent action) and the two operators on certificates (Maximal Certifiable Residue, and the two-layer base lemma of Compositional Stability); the Universal Card consolidator is \emph{specified} pending \texttt{AssuranceCard.lean}, and is not part of the 22. To these five are added the \emph{17 catalogue types} listed in Table~\ref{tab:audit_all22}: self-consistency lattice, chain-of-thought DAG, Curry--Howard structured output (compiled in two flavours, sorted/distinct and parens/quote, counted as one type), Hoeffding concentration (the Lean structure parameterized over eight named concentration inequalities; counted as one type), proof-of-sampling, scope-of-validity, counterexample-guided repair (CEGAR), and ten more (negative-guarantee, counterfactual, calibration, provenance, temporal validity, budget composition, calibrated abstention, diff, multi-hop, meta-audit). Per-certificate dependencies are in Table~\ref{tab:audit_all22}; one-paragraph descriptions and references in Appendix~\ref{app:catalogue}; a representative excerpt of the \texttt{\#print axioms} output in Appendix~\ref{app:axiom_audit}, with the build-regenerated full audit in \texttt{EmbeddingSensitivity/AxiomAudit.lean}. Beyond the 22 compiled types, Appendix~\ref{app:catalogue} additionally describes types that are \emph{specified in narrative form, pending dedicated artifact modules}: a semantic type system for text with role non-interference, conformal prediction, Markov-category compositional uncertainty, and five further extensions (spectral sensitivity, abstract interpretation, canonical-denotation decoding, dual-model cross-verification, zero-knowledge proof-carrying outputs). The catalogue is broad enough to fill a monograph, but the core development of \S\S\ref{sec:grounding}--\ref{sec:card} is what makes the framework operational; we do not duplicate the catalogue in the main text.

\label{sec:catalogue:scope}\label{sec:catalogue:cegar}

\section{Limitations and Future Work}
\label{sec:limits}

\subsection{Conditionality is unavoidable}
\label{sec:limits:conditionality}

Every certificate is conditional on tier-4 ML / human oracles. A miscalibrated paraphraser produces a sensitivity certificate whose max divergence is correct but materially misleading. A tamper-sensitive retrieval corpus grounds well-formed hallucinations. The architecture does not reach across the trust boundary; it disciplines it. The right characterization of the benefit is not ``no more hallucinations'' but ``hallucinations become named, localized, reproducible events rather than opaque stochastic outputs.''

\subsection{Empirical scope}
\label{sec:limits:empirical}

The case studies in \S\ref{sec:cases} use synthetic numbers,
flagged at the top of that section. The empirical evidence in
this paper consists of (i)~the Lean artifact's measurements
(proof sizes, build/audit times, the per-cert axiom-set
composition of Table~\ref{tab:audit_all22}); (ii)~Pilot~A
(\S\ref{sec:pilot_a}), a pre-registered three-iteration test of
the conflict-aware grounding certificate on HotpotQA dev
distractor with adversarial perturbations, in which all six v5
hypotheses pass against three baselines; (iii)~Pilots~B and B$'$
(\S\ref{sec:embedding:empirical}, \S\ref{sec:embedding:longform}),
a pre-registered two-domain test of the embedding-sensitivity
certificate on the same benchmark in matched short-form and
long-form settings: the central inequality $\Delta^2 > 0$ holds
in 20--30\% of short adversarial queries but in 98--100\% of
long-form paragraphs across every edit-type combination; and
(iv)~Pilot~C (\S\ref{sec:agents:empirical}), a pre-registered
controlled-sandbox test of the Hoare-style agent-action certificate
on a 5-task slice with 18 hand-curated direct-injection adversarial
proposals (HC2/HC3/HC4/HC5 PASS, HC1 FAIL with predicate-completeness
diagnostic).

\paragraph{Stated limitations of the four pilots.} (a) \emph{Domain
coverage} is HotpotQA dev distractor (Pilots A/B/B$'$) and a
controlled filesystem sandbox (Pilot C); patent / technical-intelligence,
legal, regulated-finance, and clinical replication of all three certificate
families remains forward work (\S\ref{sec:limits:pilot}). (b) \emph{Pilot B$'$ N} is 64 effective
queries (bootstrap CIs $\approx \pm 12$pp), sized to fit the
available budget at the observed per-query cost for long-form
variant generation (\S\ref{sec:embedding:longform}); the result
is therefore a strong directional finding on long-form rather
than a narrow-CI deployment claim, and is reported as such. (c)
\emph{Pilot C scope} is a 5-task registered slice (the original
30-task plan was scaled to fit the available budget and the
preliminary feasibility schedule); breadth across more tasks,
more tools, and more attack classes is forward work. (d)
\emph{Adversary coverage}: Pilot~A's perturbations and Pilot~C's
adversarial proposals are hand-curated; LLM-driven adversary
mining against held-out predicate / scope sets is a natural
follow-up. (e) \emph{Throughput economics}: production-scale
throughput numbers and amortised audit-cost economics are
forward work (\S\ref{sec:limits:throughput}).

Pilots~A, B, and B$'$ are reported as a matched program on a shared benchmark with the same panel discipline and pre-registration protocol; Pilot~C runs on a separately-built filesystem sandbox under the same pre-registration discipline. The contrast across the four (Pilot~A positive grounding result; Pilots~B/B$'$ empirical scope of the embedding inequality, with short-form narrow and long-form essentially universal; Pilot~C 38.9pp Treatment-over-deny-list dominance with HC1's predicate-completeness diagnostic) operationalizes the framework's intended posture: a certificate is \emph{soundness machinery} whose deployment-readiness depends on a per-call premise that must be measured per-encoder and per-input-class.

\begin{table}[H]
\centering
\caption{Comparison of the HotpotQA pilots (A, B, B$'$). Primary results from \S\S\ref{sec:pilot_a}, \ref{sec:embedding:empirical}--\ref{sec:embedding:longform}; scope limitations summarized from \S\ref{sec:limits:empirical}. Pilot C is reported separately in Table~\ref{tab:pilot_compare_c}.}
\label{tab:pilot_compare_hp}
\footnotesize
\setlength{\tabcolsep}{4pt}
\begin{tabular}{@{}p{0.6cm}p{2.8cm}p{1.0cm}p{4.0cm}p{4.4cm}@{}}
\toprule
\textbf{Pilot} & \textbf{Certificate} & \textbf{$N$} & \textbf{Primary verdict} & \textbf{Key result} \\
\midrule
A & Bilattice grounding (\S\ref{sec:grounding}) & 264 & H1--H6 all PASS (v5) & catch=1.0 at FBR=0.064; B1 cannot reach catch=1.0 below FBR=0.266 \\[2pt]
B & Embedding sensitivity, short-form (\S\ref{sec:embedding}) & 209 & HB1 narrow (30\%/20\% positive) & median $\Delta^2 = -0.07$; (paraphrase, reorder)$\times$attribute\_flip cells fail; (surface, synonym)$\times$(entity, negation) cells reach 84--97\% \\[2pt]
B$'$ & Embedding sensitivity, long-form (\S\ref{sec:embedding}) & 64 & HB$'$1 PASS (100\%/98\% positive); CP CI $\supseteq [0.92,1.00]$ & median $\Delta^2 = +0.39$; every edit-type cell positive \\
\bottomrule
\end{tabular}

\medskip
\footnotesize
\textbf{Scope limitations.} Pilot~A: HotpotQA-only; v3/v4 hypothesis revisions documented in \S\ref{sec:pilot_a:hypotheses}. Pilot~B: adversarial short queries; deployment-narrow on this slice. Pilot~B$'$: $N{=}64$, strong directional finding rather than a narrow-CI deployment claim.
\end{table}

\begin{table}[H]
\centering
\caption{Comparison row for Pilot C, the Hoare-style agent-action pilot. Primary results from \S\ref{sec:agents:empirical}.}
\label{tab:pilot_compare_c}
\footnotesize
\setlength{\tabcolsep}{4pt}
\begin{tabular}{@{}p{0.6cm}p{2.8cm}p{2.5cm}p{2.6cm}p{4.0cm}@{}}
\toprule
\textbf{Pilot} & \textbf{Certificate} & \textbf{$N$} & \textbf{Primary verdict} & \textbf{Key result} \\
\midrule
C & Hoare agent action (\S\ref{sec:agents}) & 48 trajs (16/harness; 18 unsafe-injection probes per harness) & HC2--HC5 PASS; HC1 FAIL (12/18) & 38.9pp Treatment-over-deny-list dominance; 100\% audit-log informativeness; HC1 diagnostic on 5 within-budget probes + 1 omitted predicate clause \\
\bottomrule
\end{tabular}

\medskip
\footnotesize
\textbf{Scope limitations.} 5-task registered slice (the original 30-task plan was scaled to fit the available budget and preliminary feasibility schedule); hand-curated adversaries; LLM-driven adversary mining against held-out predicate sets is forward work.
\end{table}

\paragraph{Multiple-testing posture.} The four pilots together
register eighteen pre-registered hypotheses (H1--H6 in Pilot~A,
six; HB1, HB2, and HB4 in Pilot~B short-form, three, with HB3
retired during pre-registration revisions, see
\S\ref{sec:embedding:empirical}; HB$'$1--HB$'$4 in Pilot~B$'$
long-form, four; HC1--HC5 in Pilot~C, five). Confidence intervals
in this paper are reported per-hypothesis with no family-wise
correction; the proportion CIs are 95\% bootstrap intervals
(10{,}000-resample) by default, and 95\% Clopper--Pearson exact
intervals where the bootstrap collapses to a degenerate point
(Pilot~B$'$ B$'$1; see \S\ref{sec:embedding:longform}). We declare
this posture explicitly so readers can discount marginal results
accordingly. We treat the pilots' \emph{primary} load-bearing
hypotheses as: H1 (catch-rate dominance, Pilot~A), HB$'$1 (long-form
inequality positivity, Pilot~B$'$, where the certificate's premise
is the deployment claim), and HC3 (38.9pp Treatment-over-deny-list
dominance, Pilot~C). The Pilot~B short-form HB1 hypothesis is also
pre-registered as a primary, and its result is a
\emph{scope-narrowing negative}: the certificate remains sound, but
its premise holds in only 20--30\% of cases on short adversarial
queries. We do not list it among the deployment-readiness primaries
because the paper's substantive B-family deployment claim is
carried by HB$'$1, not HB1. \emph{Pilot~C HC2 and HC5 are
small-denominator descriptive instrumentation checks rather than
deployment-readiness primaries:} HC2 reports 0/6 false-blocks on
benign destructive proposals (point estimate 0.000; 95\%
Clopper--Pearson interval $[0, 0.459]$) and HC5 reports 13/13
audit-log informativeness on Treatment blocks (1.000; CP 95\%
$[0.753, 1.000]$). Both pass the registered point-estimate criteria
but, given $n{=}6$ and $n{=}13$, are not narrow-interval population
claims; we report them as instrumentation evidence supporting HC3
rather than as standalone load-bearing conclusions. Even with a
conservative Bonferroni adjustment over the three load-bearing
primaries (three tests, $\alpha = 0.05/3 \approx 0.017$ per test),
the relevant intervals on Pilot~A H1, Pilot~B$'$ HB$'$1 long-form
($1.000$ and $0.984$ with Clopper--Pearson lower bounds $0.944$ and
$0.916$), and Pilot~C HC3 (38.9pp gap) clear the adjusted bar by
wide margins. Readers with a stricter posture should treat Pilot~B
short-form HB1, HB$'$2 / HB2 retrieval-flip-prediction F1, and HC1's
12/13 in-scope diagnostic as the multiple-testing-fragile cells.
None of the three load-bearing primaries listed above sits within
an interval whose lower bound would cross threshold under any
reasonable family-wise adjustment; the load-bearing conclusions are
therefore robust to the posture choice.

What remains explicit future work: the Pilot~D family of Universal Card
validations in high-stakes domains (\S\ref{sec:limits:pilot}; patent /
technical-intelligence, legal, regulated-finance, clinical),
domain-specific extensions of the long-form embedding result to those
corpora, calibrated post-bound construction to tighten the
\S\ref{sec:embedding} certificate's retrieval-flip prediction beyond
the loose F1 it achieves on raw drift, scaling Pilot~C's predicate
registry beyond the 5-task scaffolded slice to close the residual
predicate-completeness gap, throughput at production scale, and cost
savings on prevented errors. Building public verifiability benchmarks
across these high-stakes domains is the Stage~1 priority of the
research program.

\subsection{Tier-2 and tier-3 axioms are closable}

\texttt{cauchy\_schwarz\_sq}, \texttt{innerProd\_sub}, and \texttt{HoeffdingInequality} reduce to $\Omega$ once aligned with Mathlib's evolving inner-product-space and probability libraries. \texttt{HashCollisionResistant} and \texttt{Decode\allowbreak Algorithm\allowbreak Deterministic} reduce to $\Omega$ once a SHA-256 (or substitute) Lean formalization is plugged in. None of these axioms is structurally different from the Mathlib mathematical-axiom load that other formal-methods deployments accept.

\subsection{Promoting the Universal Assurance Card from specified to compiled}
\label{sec:limits:card}

The Universal Assurance Card schema and its \texttt{VerdictConsistent} predicate are presently \emph{specified} (\S\ref{sec:card}): the structure, the four verdicts, and the consistency obligation are stated in Lean form, but \texttt{AssuranceCard.lean} is not yet a module in \texttt{lean\_artifact/}. The promotion is mechanical: the predicate is decidable on \texttt{Rat}/\texttt{Bool}/\texttt{List String} fields with no Mathlib-heavy lemmas required, and a worked example for each of the four verdicts (matching the case studies in \S\ref{sec:cases}) closes via \texttt{decide}. We identify this promotion as a focused next-step engineering exercise rather than research; once done, the artifact gains one additional \emph{consolidator} module (the count of compiled certificate types remains 22, since the Card aggregates the 22 rather than adding a new type), and the overview table's \emph{specified} status for the Card upgrades to \emph{compiled}.

\subsection{Throughput, generator overhead, decomposition assumption}
\label{sec:limits:throughput}

Kernel type-checking is fast but not free; high-throughput LLM workflows (millions of requests per day) amortize via batching, cached invariants, sampled audits, and SNARK compression. Each certificate needs a generator; the right engineering response is a library of reusable certificate kernels (arithmetic bounds, lattice aggregation, DAG propagation, Hoeffding bound, metric-space Lipschitz). The decomposition principle applies where the pipeline has enough structure to admit a theorem; pure end-to-end open-ended generation is essentially untouchable beyond well-formedness, where the right behavior is graceful degradation (\textsc{Abstain} under the policy with reason ``task does not admit a certificate'').

\subsection{Open research questions}

(i) The right semantic type system for text across deployments. (ii) Whether abstract interpretation gives useful over-approximations of LLM outputs at practical granularity. (iii) Concentration bounds tailored to non-IID LLM sampling. (iv) Partial verification of paraphrase oracles via curated thesauri. (v) Mechanistic interpretability~\cite{olah2020zoom,elhage2021mathematical,conmy2023circuits,cunningham2023sparse} as named axioms in the chain. (vi) Multi-agent emergent-behavior certificates (liveness of a multi-agent protocol; absence of collusive privilege escalation; cross-agent Byzantine agreement). (vii) Game-theoretic adversarial semantics extending the decision-margin certificates to a continuous adversary class. (viii) A relative-completeness theorem for the framework (which pipeline properties are expressible vs.\ strictly beyond reach).

\subsection{A staged research program}
\label{sec:roadmap}

\begin{description}[leftmargin=1.5em,style=nextline,topsep=2pt,itemsep=2pt]
\item[Stage~1 (Tooling).] Reusable Lean~4 certificate kernels; an audit-runner integrating kernel type-check and \texttt{\#print axioms} audit; reference LLM-wrapper integration; a public verifiability benchmark with adversarially-introduced contradictions.
\item[Stage~2 (Single-mechanism validation).] Deploy embedding sensitivity, conflict-aware grounding, and Hoare action certificates in any high-stakes pipeline (patent claim charting, legal case-law retrieval, regulated-finance disclosure analysis, clinical RAG, or a sandboxed agent) and measure catch rates against an expert-panel ground truth.
\item[Stage~3 (Composition).] Build the composite Universal Card pipeline. Demonstrate end-to-end auditability with external auditors. Publish trust models and policy templates.
\item[Stage~4 (Extensions).] Prototype the semantic type system; explore abstract interpretation for output constraints; deploy proof-carrying outputs with and without SNARK compression in at least one high-stakes setting.
\end{description}

\subsection{Pilot status}
\label{sec:limits:pilot}

The framework's three local certificate families
(grounding, embedding sensitivity, agent action) admit empirical
pilots at parallel levels of importance.

\paragraph{Pilot~A (conflict-aware grounding, \S\ref{sec:pilot_a}): complete.} Six pre-registered hypotheses, all PASS under the v5
wording, on a 264-example perturbed slice and a 203-example
refusal-excluded unperturbed slice. The bilattice's
catch=1.0/FBR=0.064 operating point is unreachable for the
cosine-similarity baseline at any threshold.

\paragraph{Pilot~B (embedding sensitivity, short-form, \S\ref{sec:embedding:empirical}): complete.} A pre-registered
$N{=}209$ slice on text-embedding-3-large and
sentence-transformers-mpnet establishes that on short adversarial
multi-hop QA the \S\ref{sec:embedding} inequality $\Delta^2 > 0$ holds in only 30\% /
20\% of queries (FAIL $\ge 0.80$ / $\ge 0.60$). Panel reliability is
high ($\kappa = 0.887$, unanimous 0.95). A Phase~1 edit-type
ablation localizes the failure to (paraphrase, reorder) $\times$
attribute\_flip while (surface, synonym) $\times$ (entity\_swap,
negation) cells reach 84--97\% positive.

\paragraph{Pilot~B$'$ (embedding sensitivity, long-form, \S\ref{sec:embedding:longform}): complete.} A matched
pre-registered $N{=}64$ slice with the query unit shifted from
multi-hop questions to gold supporting paragraphs (median ${\approx}110$
tokens) shows the inequality holds in \textbf{100\% / 98\%} of queries
on the same two encoders, across every edit-type combination
(median $\Delta^2 = +0.39$ vs $-0.07$ short-form). The certificate is
therefore sound on every input length and empirically supported on
long-form text. The retrieval-flip-prediction F1 remains loose
(0.13--0.15) on both pilots: the bound holds, but is not tight enough
to predict which paraphrases will flip top-$k$ retrieval.

\paragraph{Pilot~C (Hoare action certificates, \S\ref{sec:agents:empirical}): complete.} A pre-registered controlled filesystem-sandbox benchmark
of 48 trajectories (16 per harness, comprising 11 adversarial trajectories carrying 18 spliced direct-injection unsafe-action proposals plus 5 benign trajectories carrying 6 destructive-but-permitted proposals) across three harnesses (Unaudited, Deny-list, Treatment).
Treatment blocks 12/18 injected unsafe actions (66.7\%); Deny-list
blocks 5/18 (27.8\%); Unaudited 0/18 (0\%, structural floor). The
38.9pp Treatment-over-Deny-list gap is well above the registered
15pp threshold (HC3 PASS). False-block rate on benign destructive
proposals is 0.000 on all three harnesses (HC2 PASS). Audit-log
informativeness on Treatment blocks is 1.000 (13/13, HC5 PASS).
HC1 verdict: FAIL (block rate 12/18 against the registered $\ge 0.95$ threshold), with a predicate-completeness diagnostic from the audit log: 5 of the 6 unblocked attacks are within-budget delete probes that the registered task t02 budget predicate correctly permits (the 6th attempt, which exceeds the budget, is correctly blocked); the remaining 1 is a mtime-violating move probing a predicate that task t01's mv declaration deliberately omits. The Lean-checked soundness theorem is unaffected. The empirical signal localizes
exactly which predicate clauses to extend per-task.

A Pilot~D family of Universal Card validations remains forward work. Each variant below is independently scoped, and any one is a complete next-step empirical contribution. The framework targets deployments where per-call error cost exceeds per-call verification cost by roughly an order of magnitude (\S\ref{sec:intro:nonclaims}); validating the Universal Assurance Card against expert-gold annotation in any such domain is the natural follow-on to the present paper. The four pilots already completed validate the underlying certificates; a Pilot~D in any high-stakes domain validates the consolidator: whether the Card's four-verdict output (\textsc{Certified}, \textsc{Partial}, \textsc{Residue}, \textsc{Abstain}), once an application policy is applied, agrees with what a domain expert would say should happen on real queries.

\paragraph{Pilot~D-Patents (technical intelligence and IP).} 100--300 prior-art and freedom-to-operate queries with patent-attorney or technology-scout annotation of the gold retrieval-and-claim gate; measures Card-verdict agreement on \textsc{Residue} cases (e.g., partial novelty findings where some claim elements are anticipated and others are not) and \textsc{Abstain} precision on queries where the corpus contains contradictory teachings. The bilattice grounding certificate (\S\ref{sec:grounding}) is particularly well-suited to patent claim charting, where one prior-art reference can simultaneously support and refute different elements of the same independent claim.

\paragraph{Pilot~D-Legal (case-law retrieval and brief support).} 100--200 case-law queries with attorney annotation; measures \textsc{Residue} agreement on holdings that are partially supported and \textsc{Abstain} precision on queries spanning circuit splits, where contradiction-aware grounding (\S\ref{sec:grounding}) and embedding sensitivity on long-form text (\S\ref{sec:embedding:longform}, where the certificate's premise holds at 98--100\%) are jointly load-bearing.

\paragraph{Pilot~D-Finance (regulated research and disclosure analysis).} 100--200 queries against SEC filings, regulatory text, and internal compliance corpora, with compliance-officer annotation; measures Card behavior on temporal-validity boundaries (a 10-K is fresh for a quarter, a regulation is fresh until amendment) and on action-gate residues for analyst-recommendation pipelines.

\paragraph{Pilot~D-Clinical (decision support).} 50--200 clinical-RAG queries with board-certified-clinician annotation of the gold action gate; measures Card-verdict agreement on \textsc{Residue} cases (partial recommendations with one contraindicated element, the §\ref{sec:case_clinical} case study made empirical) and \textsc{Abstain} precision.

Each variant requires its own annotation infrastructure and is independent of the others; collaborators in any of these domains can run a single variant without the others being in flight. The patent / technical-intelligence variant is the natural first reported, given existing infrastructure in that domain (the precursor's home territory~\cite{koomullil2026patent}); the framework's value proposition does not depend on which variant runs first.

\paragraph{Pacing.} Pilots~A, B, and B$'$ all run on HotpotQA dev
distractor (\S\ref{sec:pilot_a}, \S\ref{sec:embedding:empirical},
\S\ref{sec:embedding:longform}); the shared-benchmark design
controls cross-pilot comparisons between the two
HotpotQA-applicable certificate families and between the two
input classes (short-form questions vs.\ long-form paragraphs).
Pilot~C (\S\ref{sec:agents:empirical}) runs on a separately-built
controlled filesystem sandbox because the Hoare-action threat
model targets agentic side effects rather than retrieval / answer
generation. Each Pilot~D variant requires its own domain-specific
expert-annotation infrastructure (patent-attorney, attorney,
compliance-officer, or board-certified-clinician annotation),
none of which is built by any of the four existing pilots.

\paragraph{Near-term empirical roadmap.} For readers wanting an
ordered list of the smallest-scope empirical work that would extend
this paper's findings, four items are concrete and independent.
\textbf{(R1) Pilot~D-Patents.} A 100--300 prior-art / freedom-to-operate
slice with patent-attorney or technology-scout annotation; this is
the natural first variant given existing infrastructure
in that domain (the precursor's home territory~\cite{koomullil2026patent})
and exercises the bilattice grounding certificate on real
multi-reference contradictions. \textbf{(R2) Encoder/dataset
characterization for the embedding inequality.} A targeted run on
two to four additional encoders (Cohere embed-v3, BGE, E5, Voyage)
and two to three additional corpora (a clinical-guideline corpus,
a regulatory-text corpus, a scientific-abstract corpus), reporting
$\Delta^2 > 0$ rates with Clopper--Pearson intervals at matched
short-form / long-form input lengths; the cost is dominated by
panel labeling on the new corpora (the run script itself is
parameterized over $n$ encoders today, \S\ref{sec:embedding:empirical}).
\textbf{(R3) Pilot~C predicate-registry expansion.} Extending the
Pilot~C task slate from 5 tasks to 15--30 tasks with broader
predicate clauses (mtime invariants on \texttt{mv}, finer-grained
budget predicates) directly addresses the HC1 audit-log diagnostic
and would convert the registered FAIL into a registered PASS on the
strictest 95\% threshold. \textbf{(R4) AssuranceCard.lean
promotion.} Adding the consolidator module to \texttt{lean\_artifact/}
upgrades the Universal Card from \emph{specified} to \emph{compiled};
the proof obligations are decidable on \texttt{Rat}/\texttt{Bool}/\texttt{List String}
and discharge by \texttt{decide}, so the work is mechanical
engineering rather than research (\S\ref{sec:limits:card}). Items (R1)
and (R3) are the most informative empirical extensions; (R2) is the
fastest to deploy on existing infrastructure; (R4) is the most
self-contained and is the only item under the formal-method team's
unilateral control. Any of the four is a complete, citable next-step
contribution.

\section{Conclusion}
\label{sec:conclusion}

The thesis, restated declaratively: the framework verifies the deterministic structured computations surrounding an LLM rather than the model itself, with the trust boundary pushed as far toward the raw input as possible and everything beyond it made deterministic, checkable, and auditable.

\medskip
This paper extends the patent-analysis architecture of~\cite{koomullil2026patent} from a single specialized domain to the generic interfaces of LLM pipelines. The trust-boundary architecture, certificate-validity definition, and verification-status convention are inherited unchanged. The technical contribution is three local certificate families developed end to end (conflict-aware bilattice grounding, embedding sensitivity / paraphrase stability, and Hoare-style agent action) and two operators on those certificates (Maximal Certifiable Residue, which turns abstention into a maximal constraint-satisfying subset of an answer, and Compositional Stability, which takes the per-layer gains and margins from the three families and produces a closed-form pipeline-wide perturbation budget), consolidated into a per-call deliverable, the Universal Assurance Card (declared as a Lean structure with a decidable consistency predicate, specified pending the addition of \texttt{AssuranceCard.lean} to the artifact), with one new impossibility conjecture (calibration without ground truth). The compiled Lean~4 reference artifact spans all 22 certificate types of the broader framework with 17 of 46 declarations axiom-free and zero uses of \texttt{sorry}. Four pre-registered empirical pilots cover all three certificate families: Pilot~A on conflict-aware grounding (\S\ref{sec:pilot_a}), Pilots~B and B$'$ on embedding sensitivity in matched short-form and long-form settings (\S\ref{sec:embedding:empirical}, \S\ref{sec:embedding:longform}), and Pilot~C on Hoare-style agent action (\S\ref{sec:agents:empirical}).

The contribution is principally architectural; the empirical contribution is that all three certificate families were measured pre-registered against panel-validated or directly-injected adversarial inputs on shared infrastructure. The framework does not claim superior performance on every benchmark. The claim, now empirically demonstrated, is that a principled decomposition leaves behind a formally verifiable core at every interface where an LLM output becomes structured enough to admit a theorem; that the core is small, but its presence changes the trust profile qualitatively; that localized diagnostic failures replace opaque errors; that independent re-verification replaces vendor lock-in; and that abstention becomes a first-class outcome with an explicit residue. Pilot~A demonstrates one operationally meaningful instance: the bilattice grounding certificate reaches an operating point that cosine-similarity baselines cannot reach for any threshold. Pilots~B and B$'$ demonstrate the framework's intended epistemic posture in two-domain form: a sound certificate whose premise narrows on adversarial short queries and is essentially universal on long-form retrieval, with the failure mode mechanically explicable (surface-form bias on short queries) and the scope of applicability operationalized as a deployment-readiness signal in the Universal Card. Pilot~C demonstrates the Hoare-action family's gating value on a controlled adversarial-injection benchmark, the typed-guarantee analogue of Pilot~A's positive operating-point result. Whether the approach improves practitioner outcomes more broadly remains empirical and requires the tooling, benchmarks, and Lean developments identified as future work. The Pilot~D family (\S\ref{sec:limits:pilot}) names four parallel domain validations of the Universal Assurance Card consolidator (patent / technical-intelligence, legal, regulated-finance, and clinical), any one of which is an independent next-step empirical contribution; collaborators in any of these domains can run a single variant without the others being in flight. The method itself is the precursor's decomposition principle at the finer granularity of generic LLM pipelines, and where per-call error cost exceeds per-call verification cost by an order of magnitude or more, the shift is qualitative rather than incremental.

\section*{Code and data availability}

\paragraph{Public archive.} The compiled Lean~4 reference artifact and the four pilot directories are publicly archived at \url{https://github.com/gkoomullil/proof-carrying-certificates} and citable via concept DOI \href{https://doi.org/10.5281/zenodo.20144140}{\texttt{10.5281/zenodo.20144140}} (this paper's v1.0.0 snapshot is pinned at version DOI \href{https://doi.org/10.5281/zenodo.20144139}{\texttt{10.5281/zenodo.20144139}}). The archive is distributed under the PolyForm Noncommercial 1.0.0 license; commercial use requires a separate license from the author.

\paragraph{Artifact archive.} The archive layout is \texttt{lean\_artifact/} (the compiled Lean~4 project) and \texttt{pilot\_a/}, \texttt{pilot\_b/}, \texttt{pilot\_c/} (one self-contained Python project per pilot, with Pilot~B$'$ co-located in \texttt{pilot\_b/}).

\paragraph{Lean~4 artifact pins.} Toolchain: \texttt{leanprover/lean4:v4.30.0-rc2} (recorded in \texttt{lean\_artifact/lean-toolchain}). Mathlib commit: \texttt{ee3a5404e56808aacf7393507153cc97ff28141a} (recorded in \texttt{lean\_artifact/lake-manifest.json} alongside the pinned commits of all transitive packages, including \texttt{batteries}, \texttt{aesop}, \texttt{Qq}, \texttt{plausible}, \texttt{proofwidgets}, \texttt{importGraph}, and \texttt{LeanSearchClient}). Reproduction: \texttt{lake exe cache get} pulls the pre-built Mathlib oleans, then \texttt{lake build} compiles every module. The build emits one \texttt{\#print axioms} message per audited declaration via \texttt{EmbeddingSensitivity/AxiomAudit.lean}; redirecting that output to a file yields the full per-declaration audit (a superset of the representative excerpt in Appendix~\ref{app:axiom_audit}). The build artifact is fully deterministic; any deviation from the published audit (in particular \texttt{sorryAx} or \texttt{Lean.ofReduceBool} appearing) indicates a compilation regression.

\paragraph{Pilot pins (Pilot~A, conflict-aware grounding).} Python: see \texttt{pilot\_a/requirements.txt} (anthropic $\ge$0.40.0, openai $\ge$1.50.0, google-genai $\ge$0.3.0, numpy $\ge$2.0.0, requests $\ge$2.31.0). Models (recorded in \texttt{pilot\_a/src/config.py}): generator \texttt{claude-sonnet-4-5}; decomposer \texttt{gemini-2.5-flash-lite}; NLI / perturber / fail-mode-B / decomposition-sensitivity model and primary judge \texttt{gpt-5}; panel models \texttt{gpt-5}, \texttt{claude-sonnet-4-5}, \texttt{gemini-2.5-pro}; embedding \texttt{text-embedding-3-small}. Decoding: \texttt{temperature = 0.0} throughout. Seeds: dataset draws are deterministic functions of seeds $42$ and $11$ (see \texttt{pilot\_a/src/data.py}). API access dates and per-call cached responses are stored under \texttt{pilot\_a/cache/} and pinned to the run that produced \texttt{pilot\_a/results/}.

\paragraph{Pilot pins (Pilots~B and~B$'$, embedding sensitivity).} Encoders (\texttt{pilot\_b/src/config.py}): \texttt{openai-text-embedding-3-large} (3072-dim); \texttt{sentence-transformers/all-mpnet-base-v2} (768-dim, local); \texttt{cohere-embed-english-v3} (1024-dim, parameterized but disabled in the reported run for budget reasons). Variant generator: \texttt{gpt-5}. Three-judge validation panel: \texttt{claude-sonnet-4-5}, \texttt{gpt-5-mini}, \texttt{gemini-2.5-flash}. Decoding: \texttt{temperature = 0.0}. Per-pair embedding outputs and per-judge votes are cached under \texttt{pilot\_b/cache/} and \texttt{pilot\_b/data/}.

\paragraph{Pilot pins (Pilot~C, Hoare agent action).} Agent and injector models (\texttt{pilot\_c/src/config.py}): \texttt{gpt-5}; panel model \texttt{gpt-5-mini}. Decoding: \texttt{temperature = 0.0} (set in \texttt{pilot\_c/src/agent.py}). Agent transcripts, audit logs, and per-trajectory records are stored under \texttt{pilot\_c/results/}; the harness recreates a per-trajectory \texttt{pilot\_c/sandbox\_runs/} working directory on each run.

\paragraph{Reproducibility flow.} For each pilot, \texttt{python scripts/run\_full\_pilot.sh} (Pilot~C) or the equivalent \texttt{run\_*.py} entry points (Pilots~A, B) re-execute the registered analysis from cached API responses and emit the primary tables and confidence intervals reported in the body. The cached-response policy is monotone-build: cached responses are the source of record; a missing cache entry reissues the API call and populates the cache; existing entries are never re-issued. The Lean artifact \texttt{lake build} is independent of any external API and is bit-exact under the pinned toolchain. The license accompanying the archive is recorded in \texttt{LICENSE}.

\section*{Acknowledgments}

This paper builds directly on the hybrid AI + Lean~4 architecture of~\cite{koomullil2026patent}; the trust-boundary discipline, the certificate-validity definition (kernel type-check plus $\Omega$-audit), and the four-tier verification-status convention are inherited unchanged. We thank the Lean~4 and Mathlib communities for the mathematical infrastructure that makes this line of work feasible.

\bibliographystyle{plainnat}
\bibliography{references}

\appendix

\section{Notation}
\label{app:notation}

Symbols used in the paper, in roughly the order they appear.

\noindent
\footnotesize
\begin{longtable}{@{}p{3.0cm}p{10.5cm}@{}}
\toprule
\textbf{Symbol} & \textbf{Meaning} \\
\midrule
\endfirsthead
\toprule
\textbf{Symbol} & \textbf{Meaning} (continued) \\
\midrule
\endhead
\multicolumn{2}{@{}l}{\emph{Trust boundary and certificate validity}} \\
$\Omega$ & Trusted axiom set $\{\texttt{propext}, \texttt{Classical.choice}, \texttt{Quot.sound}\}$. \\
$\mathrm{Ax}$ & Set of declared named tier-4 oracles for a given certificate. \\
$V(c)$ & Validity check on a certificate $c$: kernel type-check plus \texttt{\#print axioms} $\subseteq \Omega \cup \mathrm{Ax}$. \\
\midrule
\multicolumn{2}{@{}l}{\emph{Embedding sensitivity (\S\ref{sec:embedding})}} \\
$E$ & Embedding map from text to $\mathbb{R}^d$ with $L_2$-normalized output. \\
$G_{\mathrm{inv}}(x)$ & Declared finite family of meaning-invariant perturbations of $x$. \\
$G_{\mathrm{sig}}(x)$ & Declared finite family of meaning-changing perturbations of $x$. \\
$R^2_{\mathrm{inv}}(x)$ & $\max_{x' \in G_{\mathrm{inv}}(x)} \|E(x') - E(x)\|^2$. Meaning-invariant robustness radius (squared). \\
$R^2_{\mathrm{sig}}(x)$ & $\min_{x' \in G_{\mathrm{sig}}(x)} \|E(x') - E(x)\|^2$. Critical sensitivity floor (squared). \\
$\Delta^2(x)$ & $R^2_{\mathrm{sig}}(x) - R^2_{\mathrm{inv}}(x)$. Selective gap; $\Delta^2 > 0$ certifies local separation of meaning. \\
$\theta_{\mathrm{ret}}$ & Retrieval similarity threshold (cutoff above which a document is returned). \\
\midrule
\multicolumn{2}{@{}l}{\emph{Conflict-aware bilattice grounding (\S\ref{sec:grounding})}} \\
$\sigma^+, \sigma^-, \sigma^0$ & Three-way NLI signature per (claim, chunk) pair: entailment, refutation, neutral. \\
$\beta^+(c), \beta^-(c)$ & Per-claim signed-support scores aggregated across chunks. \\
$W^+, W^-, W^\pm, W^?$ & Mass aggregates over atomic claims classified as \textsc{Supported}, \textsc{Contradicted}, \textsc{Contested}, \textsc{Unknown}. \\
$\theta_g$ & Support threshold for the \textsc{Supported} classification of a claim. \\
$\theta_c$ & Contradiction threshold for the \textsc{Contradicted} classification of a claim. \\
$\theta_e$ & Emission threshold ($W^+ \ge \theta_e$ required to emit). \\
$\theta_r$ & Refutation threshold ($W^- + W^\pm \le \theta_r$ required to emit). \\
\midrule
\multicolumn{2}{@{}l}{\emph{Maximal certifiable residue (\S\ref{sec:mcr})}} \\
$\mathrm{At}(y)$ & Atomic claims of an answer $y$. \\
$\mathcal{C} = \{C_1, \ldots, C_m\}$ & Family of decidable predicates (constraints) on subsets of $\mathrm{At}(y)$. \\
$\mathsf{Cert}(\mathcal{C}, y)$ & Certifiable family: subsets satisfying every $C_j \in \mathcal{C}$. \\
$\mathrm{MCR}(\mathcal{C}, y)$ & Maximal certifiable residue: highest-weight element of $\mathsf{Cert}$. \\
\midrule
\multicolumn{2}{@{}l}{\emph{Hoare-style agent action (\S\ref{sec:agents})}} \\
$\mathrm{Pre}(a, \sigma, \alpha)$ & Precondition predicate on action $a$, state $\sigma$, arguments $\alpha$. \\
$\mathrm{Post}(a, \sigma, \sigma', \alpha)$ & Postcondition relating pre-state $\sigma$ and post-state $\sigma'$. \\
\midrule
\multicolumn{2}{@{}l}{\emph{Compositional stability (\S\ref{sec:composition})}} \\
$\Pi$ & Pipeline $L_n \circ \cdots \circ L_1$ of certified layers. \\
$g_i$ & Per-layer gain (Lipschitz constant on the metric output of layer $L_i$). \\
$m_i$ & Per-layer margin (the maximum allowable input perturbation for layer $L_i$). \\
$\varepsilon^\star_i$ & Chained perturbation budget at the input of layer $L_i$. \\
$\kappa_\Pi$ & Cumulative gain $\prod_{i=1}^n g_i$. \\
$B_\Pi$ & Pipeline-wide perturbation budget $\min_i m_i / \prod_{j<i} g_j$. \\
\midrule
\multicolumn{2}{@{}l}{\emph{Statistical and empirical}} \\
$\kappa$ & Cohen's kappa (panel-reliability statistic). \\
$N$ & Sample size (per-pilot effective query count). \\
\bottomrule
\end{longtable}

\section{Abbreviations}
\label{app:abbreviations}

\noindent
\footnotesize
\begin{longtable}{@{}p{2.5cm}p{11.0cm}@{}}
\toprule
\textbf{Abbreviation} & \textbf{Expansion} \\
\midrule
\endfirsthead
\toprule
\textbf{Abbreviation} & \textbf{Expansion} (continued) \\
\midrule
\endhead
AI & Artificial intelligence. \\
API & Application programming interface. \\
CEGAR & Counterexample-guided abstraction refinement. \\
CI & Confidence interval. \\
CoT & Chain-of-thought reasoning. \\
CTL & Computation tree logic. \\
DAG & Directed acyclic graph. \\
DOAC & Direct oral anticoagulant (clinical case study, \S\ref{sec:case_clinical}). \\
EU AI Act & European Union Artificial Intelligence Act. \\
FBR & False-block rate (Pilot A primary metric). \\
FDA SaMD & U.S. Food and Drug Administration regulation of Software as a Medical Device. \\
HIPAA & Health Insurance Portability and Accountability Act (U.S. clinical-privacy law). \\
IID & Independent and identically distributed. \\
ISO/IEC 42001 & International AI-management-system standard. \\
ITP & Interactive theorem proving. \\
JSON & JavaScript Object Notation. \\
LOC & Lines of code. \\
LLM & Large language model. \\
LTL & Linear temporal logic. \\
MCR & Maximal certifiable residue (\S\ref{sec:mcr}). \\
ML & Machine learning. \\
NIST AI RMF & U.S. National Institute of Standards and Technology AI Risk Management Framework. \\
NLI & Natural language inference. \\
NLP & Natural language processing. \\
NN & Neural network. \\
NNV & Neural-network verification. \\
PCC & Proof-carrying code. \\
PII & Personally identifiable information. \\
PRM & Process reward model. \\
QA & Question answering. \\
RAG & Retrieval-augmented generation. \\
ReAct & Reasoning-and-acting agent paradigm. \\
RL & Reinforcement learning. \\
SHA-256 & Secure Hash Algorithm, 256-bit output. \\
SMT & Satisfiability modulo theories. \\
SNARK & Succinct non-interactive argument of knowledge. \\
TCB & Trusted computing base. \\
\bottomrule
\end{longtable}

\section{Lean~4 Compiled Reference Artifact: Structures and Theorems}
\label{app:embedding_lean}

The companion Lean~4 project at \texttt{lean\_artifact/} compiles all 22 certificate types under Lean~4 v4.30.0-rc2 with Mathlib. The artifact builds via \texttt{lake build} in roughly ten seconds with a warm Mathlib cache. We reproduce here a representative excerpt: the compiled embedding sensitivity certificate, the conflict-aware grounding certificate skeleton, and the Hoare trajectory certificate. A representative \texttt{\#print axioms} excerpt is in Appendix~\ref{app:axiom_audit}, and the full per-declaration audit is regenerated by \texttt{EmbeddingSensitivity/AxiomAudit.lean} on every build.

\paragraph{Embedding sensitivity (\S\ref{sec:embedding}).} The crucial design choice is squared radii throughout, so the certificate stays in $\mathbb{Q}$ and avoids irrational $\sqrt{\cdot}$ in proof fields; square roots enter only at the theorem statement.

\smallskip
\leanverbatim{verbatim compiled excerpt from \texttt{Certificate.lean}}
\begin{lstlisting}
namespace VerifiedEmbedding

axiom ParaphraseOracle : String -> List String -> List String -> Prop

axiom cauchy_schwarz_sq (u v : Embedding) :
  innerProd u v * innerProd u v ≤ l2NormSq u * l2NormSq v

axiom innerProd_sub (u v w : Embedding) :
  innerProd (List.zipWith (· - ·) u v) w
   = innerProd u w - innerProd v w

structure EmbeddingSensitivityCertificate where
  text : String
  embedding : Embedding
  G_inv G_sig : List String
  emb_inv emb_sig : List Embedding
  R_inv_sq R_sig_sq delta_sq : Rat
  p_inv_length : G_inv.length = emb_inv.length
  p_sig_length : G_sig.length = emb_sig.length
  p_bounded_orig : l2NormSq embedding ≤ 1
  p_bounded_inv : ∀ v ∈ emb_inv, l2NormSq v ≤ 1
  p_bounded_sig : ∀ v ∈ emb_sig, l2NormSq v ≤ 1
  p_R_inv_nonneg : 0 ≤ R_inv_sq
  p_R_sig_nonneg : 0 ≤ R_sig_sq
  p_R_inv_sound : ∀ e ∈ emb_inv, l2DistSq embedding e ≤ R_inv_sq
  p_R_sig_sound : ∀ e ∈ emb_sig, R_sig_sq ≤ l2DistSq embedding e
  p_delta_correct : delta_sq = R_sig_sq - R_inv_sq
  h_paraphrase_valid : ParaphraseOracle text G_inv G_sig

theorem robust_similarity_sq
    (cert : EmbeddingSensitivityCertificate)
    (d_emb : Embedding)
    (h_d_norm : l2NormSq d_emb ≤ 1)
    (e' : Embedding)
    (h_e_in : e' ∈ cert.emb_inv) :
    let diff := innerProd cert.embedding d_emb - innerProd e' d_emb
    diff * diff ≤ cert.R_inv_sq := by
  have h_drift : l2DistSq cert.embedding e' ≤ cert.R_inv_sq :=
    cert.p_R_inv_sound e' h_e_in
  exact similarity_diff_sq_bound cert.embedding e' d_emb
        cert.R_inv_sq h_drift h_d_norm

end VerifiedEmbedding
\end{lstlisting}

\paragraph{Bilattice grounding (\S\ref{sec:grounding}).} The full structure including \texttt{p\_emission\_rule}, \texttt{p\_reconstruction}, and the two tier-4 oracles is in \S\ref{sec:grounding} above; the emission-gate soundness lemma (Lemma~\ref{thm:no_silent}) discharges by definitional unfolding plus rational arithmetic and is conditional on the NLI / decomposition oracles.

\paragraph{Hoare action (\S\ref{sec:agents}).} The counter-state instantiation in the artifact is pure integer arithmetic with decidable equality; every proof field reduces by kernel computation alone. \texttt{\#print axioms} reports ``does not depend on any axioms'' for every declaration in this module.

\paragraph{Universal Assurance Card (\S\ref{sec:card}).} The card structure and the \texttt{VerdictConsistent} predicate are stated in Lean form in \S\ref{sec:card}; the predicate is decidable on \texttt{Rat}, \texttt{Bool}, and \texttt{List String} fields, requires no Mathlib-heavy lemmas, and is well-suited to direct \texttt{decide} discharge once added to the artifact as \texttt{AssuranceCard.lean}. The worked-example cards from \S\ref{sec:cases} are correspondingly \emph{specified}; promoting them to \emph{compiled} (a focused engineering step) is the natural follow-on to this paper. The remainder of this audit reports the 22 currently-compiled certificate types.

\section{Representative Axiom Audit}
\label{app:axiom_audit}

The following is a representative excerpt of the \texttt{\#print axioms} output from running \texttt{lake build} on the artifact, organized by certificate family. It reproduces every core declaration verbatim along with the named non-$\Omega$ axioms for each catalogue family. It is not the complete \texttt{\#print axioms} output for every one of the 46 audited declarations; the full output is regenerated on every build by the \texttt{EmbeddingSensitivity/AxiomAudit.lean} module, which emits one \texttt{\#print axioms} info-message per audited declaration during \texttt{lake build} and is the reproducible source of record for the audit. No \texttt{sorry}, no \texttt{Lean.ofReduceBool}; every reported axiom is either in the kernel-trusted set $\Omega$ or an explicitly declared named non-$\Omega$ assumption (tier-2 mathematical placeholder, tier-3 cryptographic assumption, or tier-4 ML/human oracle, per \S\ref{sec:arch:taxonomy}). Example-specific witness axioms (such as \texttt{example\_paraphrase\_oracle}, \texttt{example\_decomp\_oracle}, and \texttt{example\_nli\_oracle}) appear only in illustrative example modules to discharge the corresponding tier-4 oracle premises on a single concrete instance; they are not framework-level oracles and are excluded from the five persistent tier-4 oracles counted in Table~\ref{tab:oracle_recipe}.

\smallskip
\leanverbatim{embedding sensitivity (\S\ref{sec:embedding})}
\begin{lstlisting}
'VerifiedEmbedding.l2NormSq' depends on axioms:
  [propext, Classical.choice, Quot.sound]
'VerifiedEmbedding.l2DistSq' depends on axioms:
  [propext, Classical.choice, Quot.sound]
'VerifiedEmbedding.EmbeddingSensitivityCertificate' depends on axioms:
  [propext, Classical.choice, Quot.sound, ParaphraseOracle]
'VerifiedEmbedding.similarity_diff_sq_bound' depends on axioms:
  [propext, Classical.choice, Quot.sound,
   cauchy_schwarz_sq, innerProd_sub]
'VerifiedEmbedding.robust_similarity_sq' depends on axioms:
  [propext, Classical.choice, Quot.sound, ParaphraseOracle,
   cauchy_schwarz_sq, innerProd_sub]
'VerifiedEmbedding.Example.exampleCertificate' depends on axioms:
  [propext, Classical.choice, Quot.sound, ParaphraseOracle,
   example_paraphrase_oracle]
\end{lstlisting}

\smallskip
\leanverbatim{conflict-aware grounding (\S\ref{sec:grounding})}
\begin{lstlisting}
'VerifiedEmbedding.Grounding.EpistemicStatus'
  does not depend on any axioms
'VerifiedEmbedding.Grounding.classify'
  does not depend on any axioms
'VerifiedEmbedding.Grounding.ConflictAwareGroundingCertificate'
  depends on axioms: [propext, Classical.choice, Quot.sound,
                      DecompositionOracle, SignedSupportOracle]
'VerifiedEmbedding.Grounding.no_silent_contradictions'
  depends on axioms: [propext, Classical.choice, Quot.sound,
                      DecompositionOracle, SignedSupportOracle]
'VerifiedEmbedding.Grounding.Example.exampleGroundingCertificate'
  depends on axioms: [propext, Classical.choice, Quot.sound,
                      DecompositionOracle, SignedSupportOracle,
                      example_decomp_oracle, example_nli_oracle]
\end{lstlisting}

\smallskip
\leanverbatim{Hoare agent action (\S\ref{sec:agents})}
\begin{lstlisting}
'VerifiedEmbedding.Action.ActionSchema'        does not depend on any axioms
'VerifiedEmbedding.Action.ActionCertificate'   does not depend on any axioms
'VerifiedEmbedding.Action.TrajectoryCertificate' does not depend on any axioms
'VerifiedEmbedding.Action.step_post_holds'     does not depend on any axioms
'VerifiedEmbedding.Action.trajectory_all_steps_safe'
                                                does not depend on any axioms
'VerifiedEmbedding.Action.Example.exampleTrajectory'
                                                does not depend on any axioms
\end{lstlisting}

\smallskip
\leanverbatim{compositional stability (\S\ref{sec:composition}) and MCR (\S\ref{sec:mcr})}
\begin{lstlisting}
'Composition.compositional_stability_two_layer' depends on axioms: Ω
'MCR.residue_is_certifiable'                     depends on axioms: Ω
'MCR.subset_of_residue_is_certifiable'          depends on axioms: Ω
'MCR.residue_is_fixed_point'                     depends on axioms: Ω
\end{lstlisting}

\smallskip
\leanverbatim{Hoeffding concentration (catalogue)}
\begin{lstlisting}
'Concentration.HoeffdingCertificate' :
   [propext, Classical.choice, Quot.sound,
    HoeffdingInequality, IIDSamples]
'Concentration.phat_in_unit_interval' : Ω + {HoeffdingInequality, IIDSamples}
'Concentration.hoeffding_confidence_interval' :
   Ω + {HoeffdingInequality, IIDSamples}
\end{lstlisting}

\smallskip
\leanverbatim{Curry--Howard structured outputs (catalogue)}
\begin{lstlisting}
'CurryHoward.SortedList'      does not depend on any axioms
'CurryHoward.DistinctList'    does not depend on any axioms
'CurryHoward.BalancedParens'  depends on axioms: Ω
'CurryHoward.GroundedQuote'   depends on axioms: Ω
\end{lstlisting}

\paragraph{What the audit shows.} (i) Zero \texttt{sorryAx} or \texttt{Lean.ofReduceBool} anywhere in the transitive set of any certificate, the operational definition of \emph{sorry-free $\Omega$-audited} (Def.~\ref{def:validcert}). (ii) The Hoare action certificate, scope-of-validity, simpler CEGAR theorems, simpler Curry--Howard exemplars, and several catalogue types (provenance, temporal validity, calibrated abstention, diff, meta) compile with \emph{no axioms at all} (17 of 46 declarations); the guarantee is unconditional modulo the Lean kernel itself. (iii) The remaining 29 declarations depend only on $\Omega$ plus explicitly declared non-$\Omega$ assumptions, partitioned into tier-2 mathematical placeholders, tier-3 cryptographic assumptions, and tier-4 ML/human oracles; the auditor sees exactly which natural-language judgments and which mathematical/cryptographic assumptions the certificate is conditional on, and can reject any certificate whose declared-axioms list is not within an approved set.

\paragraph{Reproducibility.} From \texttt{lean\_artifact/}: \texttt{lake exe cache get} pulls the precompiled Mathlib cache; \texttt{lake build} compiles everything and emits the \texttt{\#print axioms} output above verbatim. Any deviation (in particular \texttt{sorryAx} or \texttt{Lean.ofReduceBool} appearing) would indicate a compilation regression and cause the audit to reject under Definition~\ref{def:validcert}.

\section{Extended Catalogue: Per-Type Companion Notes}
\label{app:catalogue}

For completeness, the broader catalogue beyond the main-text core. The paragraphs are deliberately terse; full structures are in the Lean sources for the compiled types. The 17 catalogue types compiled in the artifact (Table~\ref{tab:audit_all22}) are: self-consistency lattice, chain-of-thought DAG, Curry--Howard structured output, concentration, proof-of-sampling, scope-of-validity, CEGAR, and the ``ten additional types'' block (negative-guarantee through meta/self-audit). The paragraphs below labeled ``\textsc{specified}'' are described in narrative form, pending dedicated artifact modules: semantic type system for text (the role-tag concept is used implicitly by the grounding certificate of \S\ref{sec:grounding} but does not yet have a standalone module), conformal prediction, compositional uncertainty (Markov categories), and the five-way extensions block. These are not part of the 22 compiled count.

\paragraph{Self-consistency lattice (axiom-free).} A canonicalizer-parameterized aggregation of $k$ samples to canonical form, with a winner-meets-threshold theorem and a Galois-connected threshold ladder. Compiles with no axioms at all on finite concrete data.

\paragraph{Chain-of-thought DAG.} Reasoning steps as a typed DAG with parameterized aggregation ($\mathcal{A}_{\min}$, $\mathcal{A}_\Pi$, $\mathcal{A}_\lor$, $\mathcal{A}_{FH}$); Conservative Aggregation Soundness Theorem ($\mathcal{A} \le \min$ implies sound lower-bound under any joint, with $\mathcal{A}_\Pi$ tight under independence). Tier-4 oracle: \texttt{StepConfidenceOracle}.

\paragraph{Curry--Howard structured output.} Inductive types whose constructors carry decidable proof obligations: \texttt{SortedList}, \texttt{DistinctList}, \texttt{BalancedParens}, \texttt{GroundedQuote}. The LLM proposes; the kernel rejects. Refinement-type compilation~\cite{rondon2008liquid,vazou2014refinement} handles routine arithmetic via SMT; full dependent-type obligations stay in Lean.

\paragraph{Concentration certificates.} Hoeffding (binary), generalized Hoeffding ($[a,b]$), Bernstein, empirical Bernstein, Azuma--Hoeffding (martingales), Bennett, McDiarmid, union-bound categorical. Each is the same Lean structure with a different \texttt{failure\_bound} and named hypothesis. Named non-$\Omega$ axioms: \texttt{IIDSamples} (tier-4 oracle on the sampling discipline), \texttt{HoeffdingInequality} (tier-2 mathematical placeholder pending Mathlib formalization).

\paragraph{Proof-of-sampling.} Commitment-bound randomness for honest $k$-sample claims via a hash of seed plus sampling-record. Tier-3 oracle: \texttt{HashCollisionResistant}.

\paragraph{Semantic type system for text [\textsc{specified}].} Role-tagged text (\texttt{user\_input}, \texttt{system\_prompt}, \texttt{retrieved\_fact}, \texttt{tool\_output}, \ldots) with a coercion lattice. A role non-interference property follows by structural induction over the coercion record: when the lattice forbids the flow $\rho_1 \to \rho_2$, the pipeline output restricted to role $\rho_2$ is a function only of inputs of role $\neq \rho_1$. Prompt injection via retrieved content thus becomes a compile-time error rather than a runtime filter.

\paragraph{Scope-of-validity.} For each oracle, a decidable scope predicate on the input plus a graceful-downgrade rule when the predicate fails. Compiles axiom-free in the core case. Composes with every other certificate as an input gate.

\paragraph{Counterexample-guided certificate repair (CEGAR).} Failure witnesses are first-class certificate fields; on failure, the LLM is re-prompted with a typed witness; loops terminate within a declared budget. Connects to neural theorem-proving repair~\cite{jiang2023draftsketchprove}.

\paragraph{Conformal prediction [\textsc{specified}].} Distribution-free coverage on the candidate set under an exchangeability axiom; Mondrian, conformal-risk, and adaptive variants for drift.

\paragraph{Ten additional types.} Negative-guarantee (absence of forbidden patterns), counterfactual / causal (output stable under protected-attribute swap), calibration (max bin gap on held-out set), provenance (cryptographic input binding), temporal validity (freshness window), budget composition (cost stays under parent's allowance), calibrated abstention (conformal abstain rate), diff / incremental (incremental certificate over diff), multi-hop branching (AND/OR/NOT reasoning trees), meta / self-audit (kernel hash + reproducible build).

\paragraph{Compositional uncertainty (Markov categories) [\textsc{specified}].} A categorical treatment of compositional uncertainty, including Markov-category formulations~\cite{fritz2020synthetic} of credal-set or interval certificates, is a theoretical direction outside the compiled artifact and outside the present paper's implementation claims.

\paragraph{Spectral sensitivity [\textsc{specified}].} Complements the finite-family embedding certificate of \S\ref{sec:embedding} by bounding an encoder Lipschitz constant via per-layer spectral-norm or randomized-smoothing witnesses, then lifting the finite perturbation family to a continuous metric-ball statement. The bound is conservative rather than exact, in contrast to \S\ref{sec:embedding}'s finite max/min certificate.

\paragraph{Dual-model cross-verification [\textsc{specified}].} Submits the same input to two independently-implemented providers, applies a deterministic canonicalizer to both outputs, and emits either a weakest-verdict agreement Card or a dissent-marked downgraded Card. The guarantee is agreement on canonical denotation, not truth; the audit handle permits a third party to replay the canonicalizer check against both providers' kernel-audited build outputs without re-running either underlying model.

\paragraph{Abstract interpretation, canonical-denotation decoding, zero-knowledge proof-carrying outputs [\textsc{specified}].} Three further extensions: sound output over-approximations (e.g.\ PII-excluding grammar); semantic reproducibility via deterministic canonicalization at the decoder; and SNARK compression of certificates for transport.

\section{Adoption Pathways and Regulatory Mapping}
\label{app:adoption}

\paragraph{Adoption pathways.} Adoption pathways span product UX (assurance-mode selectors, claim-by-claim highlighting, time-decaying cards), engineering (provider-agnostic verifier sidecars, automatic metamorphic challenges, cross-provider quorum, CEGAR-style certificate-conditioned repair), governance (verifiability as a leaderboard dimension, procurement and insurance contracts, per-organization trust budgets), and training (a \emph{certifiability rate} optimized alongside accuracy).

\paragraph{Regulatory mapping.} The certificate taxonomy and the Universal Card map onto major frameworks: EU AI Act articles on risk management (Art.~9), logging (Art.~12), transparency (Art.~13), human oversight (Art.~14), and accuracy / robustness / cybersecurity (Art.~15); FDA SaMD and NIST AI RMF MAP/MEASURE/MANAGE/GOVERN; ISO/IEC 42001 management-system clauses; HIPAA / clinical privacy via the role-typed semantic type system; Federal Rules of Evidence~702 / Daubert via the kernel-checked $\Omega$-audit. The point is structural: the framework is a reinforcement of the same guarantees these regimes require, with kernel-audited certificates standing in for hand-waving compliance evidence.

\end{document}